\pdfoutput=1 
\documentclass[11pt]{article}
\usepackage{fullpage}
\usepackage[pagebackref, colorlinks = true, linkcolor = blue, urlcolor  = blue, citecolor = purple, bookmarks, hypertexnames=false]{hyperref}
\usepackage{amssymb}
\usepackage{amsmath}
\usepackage{amsthm}
\usepackage{graphicx,color,colordvi}
\usepackage{bbm}
\usepackage{varioref}
\usepackage[capitalise]{cleveref}
\usepackage{stmaryrd}
\usepackage[utf8]{inputenc}
\usepackage[blocks]{authblk}
\usepackage{dsfont}
\usepackage{mathtools}
\usepackage{wrapfig}
\usepackage[english]{babel}
\usepackage{physics}
\usepackage{braket}
\usepackage[table]{xcolor}
\usepackage[center]{caption}
\usepackage{tikz}
\usetikzlibrary{shapes,arrows,patterns,positioning,shapes.geometric,arrows.meta,decorations.markings}
\usepackage{ifthen}
\usepackage{pgfplots}
\pgfplotsset{compat=1.9}
\usetikzlibrary{shapes,arrows.meta}
\usetikzlibrary{positioning}
\usetikzlibrary{shapes.geometric}
\RequirePackage[framemethod=default]{mdframed}
\usepackage[margin=1in]{geometry}
\usepackage{comment}
\usepackage{url}
\usepackage{ upgreek }
\usepackage{fancyref} 
\usepackage{float}
\usepackage{enumitem}
\usepackage{autonum}
\usepackage{subcaption}
\newfloat{algorithm}{t}{lop}
\usepackage{array,rotating ,makecell, multirow, tabularx}
\usepackage{scrextend}
\usepackage{mathrsfs}
\usepackage{soul}

\usepackage{IEEEtrantools}
\usepackage[normalem]{ulem}
\usepackage{cancel}

\DeclareMathOperator{\supp}{supp}

\DeclareMathOperator{\id}{id}
\DeclareMathOperator{\co}{co}
\DeclareMathOperator{\cb}{cb}

\newcommand{\1}{\ensuremath{\mathbbm{1}}}
\newcommand{\lamdba}{\lambda}

\newcommand{\SWAP}{\mathsf{SWAP}}

\newtheoremstyle{newdefinition}{}{}{\normalfont}{}{\bfseries}{}
{ }
{\thmname{#1} \thmnumber{#2}\thmnote{ (#3)}}

\newtheoremstyle{newplain}{}{}{\itshape}{}{\bfseries}{}{1em}
{\thmname{#1} \thmnumber{#2}\thmnote{ (#3)}}

\newtheoremstyle{newremark}{}{}{\normalfont}{}{\bfseries}{}{1em}
{\thmname{#1}}

\theoremstyle{newdefinition}
\newtheorem{definition}{Definition}[]

\theoremstyle{newplain}
\newtheorem{theorem}{Theorem}[section]
\newtheorem{lemma}{Lemma}[section]
\newtheorem{proposition}{Proposition}[section]
\newtheorem{corollary}{Corollary}[section]
\newtheorem{remark}{Remark}[section]

\newtheorem{conjecture}{Conjecture}[]
\newtheorem{counterexample}{Counterexample}[section]

\newtheoremstyle{myplain}{5pt}{5pt}{\itshape}{0pt}{\bfseries}{}{5pt plus 1pt minus 1pt}{}
\theoremstyle{myplain}
\newtheorem*{theorem*}{Theorem}
\newtheorem*{corollary*}{Corollary}
\newtheorem*{lemma*}{Lemma}

\DeclareMathOperator{\cH}{\mathcal{H}}

\DeclareMathOperator{\cS}{\mathcal{S}}

\DeclareMathOperator{\cK}{\mathcal{K}}

\DeclareMathOperator{\cB}{\mathcal{B}}

\DeclareMathOperator{\cX}{\mathcal{X}}
\DeclareMathOperator{\cY}{\mathcal{Y}}

\DeclareMathOperator{\cD}{\mathcal{D}}

\title{Two-Indexed Schatten Quasi-Norms with\\ Applications to Quantum Information Theory}
\author[1]{Jan Kochanowski\thanks{Email: jan.kochanowski@inria.fr}}
\author[2]{Omar Fawzi}
\author[1]{Cambyse Rouzé}

\affil[1]{Inria, Télécom Paris - LTCI, Institut Polytechnique de Paris, 91120 Palaiseau, France}

\affil[2]{Inria, ENS Lyon, UCBL, LIP, F-69342 Lyon Cedex 07, France}
\date{\today}

\begin{document}

\maketitle

\begin{abstract}
    We define 2-indexed $(q,p)$-Schatten quasi-norms for any $q,p > 0$ on operators on a tensor product of Hilbert spaces, naturally extending the norms defined by Pisier's theory of operator-valued Schatten spaces. We establish several desirable properties of these quasi-norms, such as relational consistency and the behavior on block diagonal operators, assuming that $|\frac{1}{q} - \frac{1}{p}| \leq 1$. In fact, we show that this condition is essentially necessary for natural properties to hold. Furthermore, for linear maps between spaces of such quasi-norms, we introduce completely bounded quasi-norms and co-quasi-norms. We prove that the $q \to p$ completely bounded co-quasi-norm is super-multiplicative for tensor products of quantum channels for $q \geq p>0$, extending an influential result of [Devetak, Junge, King, Ruskai, 2006]. Our proofs rely on elementary matrix analysis and operator convexity tools and do not require operator space theory. On the applications side, we demonstrate that these quasi-norms can be used to express relevant quantum information measures such as R\'enyi conditional entropies for $\alpha \geq \frac{1}{2}$ or the Sandwiched Rényi Umlaut information for $\alpha < 1$. Our multiplicativity results imply a tensorizing notion of reverse hypercontractivity, additivity of the completely bounded minimum output Rényi-$\alpha$-entropy for $\alpha\geq\frac{1}{2}$ extending another important result of [Devetak, Junge, King, Ruskai, 2006], and additivity of the maximum output R\'enyi-$\alpha$ entropy for any $\alpha \geq \frac{1}{2}$. 
\end{abstract}

\newpage

\setcounter{tocdepth}{3} 
\tableofcontents

\newpage

\section{Introduction}

In his seminal work \cite{Book.Pisier.1998}, Gilles Pisier introduced 2-indexed Schatten norms on operators on a tensor product of Hilbert spaces. They are indexed by $q,p\in[1,\infty]$ and act locally as Schatten-$q$-, respectively -$p$-norms. In particular, they satisfy
\begin{align}
    \|X\otimes Y\|_{(q,p)}= \|X\|_q\|Y\|_p
\end{align} 
and serve as the non-commutative generalizations to the $\ell_q[\ell_p]$-vector norms, defined for $q,p\in[1,\infty]$ as
\begin{align}\label{equ:classical}
    \|v\|_{\ell_q[\ell_p]}=\bigg(\sum_i\bigg(\sum_j|v_{ij}|^p\bigg)^\frac{q}{p} \bigg)^\frac{1}{q}.
\end{align}
While these norms naturally appeared in the context of \emph{operator-valued Schatten spaces}, these have recently found ample applications in quantum information theory. In particular completely bounded (cb) norms of linear maps are defined via such norms. Examples include the diamond norm between quantum channels \cite[Chapter 3]{book:Watrous.2018}, cb-minimum output (conditional) entropies of quantum channels \cite{Devetak.2006, Fawzi.2026, Kochanowski.2025} or the complete modified logarithmic Sobolev inequalities \cite{Bardet.2022, Bardet.2024}. See \cite{Beigi.2023} for a recent review of these norms and \cite{book.Junge.2010} for a complete monograph.

It is established \cite{Beigi.2023} that the optimized sandwiched Rényi-$\alpha$ entropy of a state $\rho_{AB}$, as well as the entanglement Rényi-$\alpha$ entropy \cite{Wang.2016} of some pure state $\psi_{AB}$, for $\alpha>1$, respectively satisfy
\begin{align}
\label{eq:uparrowentropy_norm}
    H_\alpha^\uparrow(A|B)_{\rho_{AB}}=\frac{\alpha}{1-\alpha}\log\|\rho_{BA}\|_{(1,\alpha)},\qquad \textup{ and }\qquad H_\alpha(A)_{\tr_B[\psi_{AB}]}=\frac{\alpha}{1-\alpha}\log\|\psi_{AB}\|_{(\alpha,1)}.
\end{align} 
Both of these objects $H_\alpha^\uparrow(A|B)_{\rho_{AB}}, H_\alpha(A)_{\tr_B[\psi_{AB}]}$ are well defined and find significance in quantum information theory, however, not just for $\alpha>1$, but in fact for any $\alpha\geq\frac{1}{2}$ or even $\alpha>0$ in the second case.
In fact there are plenty of quantum information theoretic objects and tasks which would be naturally described using such 2-indexed Schatten norms with indices $q,p\leq 1$. In particular entropic quantities such as the optimized sandwiched Rényi-$\alpha$ conditional entropies, cb-minimal output Rényi-$\alpha$ entropies, Rényi-$\alpha$-coherent information and the sandwiched $\alpha$-Umlaut information for $\alpha\in[\frac{1}{2},1)$, are prime candidates. See \cref{sec:intro:multiplicativity} and \cref{sec:Applications} for more details. 
Additionally there lacks a notion of completely bounded quasi-norms of linear maps between quasi-normed spaces which would give rise to, on the one hand, a notion of complete reverse hypercontractivity, extending reverse hypercontractivity \cite{Cubitt.2015,Beigi.2020} and mirroring complete hypercontractivity \cite{Beigi.2016, Bardet.2022}. Such norms are intimately linked to mixing time bounds of quantum Markov semigroups \cite{Bardet.2024} and one would expect the same to be true for the complete reverse hypercontractivity.
On the other hand, such a notion should give rise to a completely bounded minimum output Rényi-$\alpha$ entropies, among others, mirroring the $\alpha>1$ case of \cite{Devetak.2006}.

However, unlike in the commutative case of \eqref{equ:classical}, which readily extends to the quasi-normed regime of $0<q,p\leq 1$, the 2-indexed Schatten norms are by construction constrained to indices $1\leq q,p\leq\infty$ as operator space theory manifestly builds on hierarchies of norms. This includes the construction of Pisier norms via Haagerup tensor-products of row- and column operator spaces and explicit duality statements which fail to hold for (Schatten) quasi-norms \cite{Book.Pisier.1998, Beigi.2023}. A natural question then arises:

\begin{center}
\textit{What is the natural generalization of Pisier's 2-indexed Schatten norms when $q$ or $p$ is in $(0,1)$?}
\end{center}

In this work, we answer this question by introducing new quasi-norms for indices $0<q,p\leq\infty$ with $\big|\frac{1}{q}-\frac{1}{p}\big|\leq 1$, which we call \emph{2-indexed Schatten quasi-norms}.
We show that this `compatibility' condition, i.e. $\big|\frac{1}{q}-\frac{1}{p}\big|\leq 1$, is sufficient and in some sense necessary to admit natural properties one would expect such quasi-norms to possess. Hence, they form a natural extension of Pisier's theory to the quasi-normed regime and moreover, recover it when $1\leq q,p$. The 2-indexed Schatten quasi-norms allows us to establish new and extend old connections between entropic quantities and 2-indexed norms, for example \eqref{eq:uparrowentropy_norm} is extended to the whole range $\alpha \geq \frac{1}{2}$.
We also show that they naturally give rise to notions of completely bounded quasi-norms and co-quasi-norms of linear maps and prove multiplicativity properties for tensor products of linear maps.  
In particular, this suggests a notion of complete reverse hypercontractivity and shows additivity of the completely bounded minimal output Rényi entropy and maximal output Rényi entropy under tensor products of CP maps for any index $\frac{1}{2}\leq p$.
The proofs in this paper use matrix analysis techniques and do not require knowledge of operator space theory.

\subsection{Intuition from the Commutative Setting}

Before we present our definition and the main results of the paper, we give some intuition behind our approach to 2-indexed Schatten quasi-norms, by first considering the commutative case, i.e. the 2-indexed $\ell_q[\ell_p]$-(quasi-)norms, defined for any set of indices $0<q,p<\infty$ via
\begin{align}
\|v\|_{\ell_q[\ell_p]}:=\left(\sum_{i=1}^{D_1}\left(\sum_{j=1}^{D_2}|v_{ij}|^p\right)^\frac{q}{p} \right)^\frac{1}{q}\,,\qquad \forall v\in \mathbb{C}^{D_1\times D_2}.
\end{align}
Unfortunately this does not directly generalize to the non-commutative setting as for instance $\Tr_1[\tr_2[|X|^p]^\frac{q}{p}]^\frac{1}{q}$ does not define a norm, nor a quasi-norm \cite{Devetak.2006}. 
We can, however, re-express this definition into a variational formula using the generalized Hölder inequality.

\begin{theorem}[2-indexed $\ell_q$ (quasi-)norms]\label{thm:classical-2-indexed-spaces}
Let $0<p,q\leq\infty$ and set $\frac{1}{r}=\big|\frac{1}{q}-\frac{1}{p}\big|$, then for any $v\in\mathbb{C}^{D_1\times D_2}$ it holds that
\begin{align}
    \|v\|_{\ell_q[\ell_p]} = \begin{cases}
        \underset{{\substack{a,b\in\ell_{2r} \\a_i,b_i>0\forall i:v_i\ne 0}}}{\inf}\|a\|_{\ell_{2r}}\|b\|_{\ell_{2r}}\|(a^{-1}\otimes\mathbf{1})\circ v\circ (b^{-1}\otimes\mathbf{1})\|_{\ell_p} \quad &\text{ if } q\leq p \\
        \underset{{a,b\in\ell_{2r}}}{\sup}\|a\|^{-1}_{\ell_{2r}}\|b\|^{-1}_{\ell_{2r}}\|(a\otimes\mathbf{1})\circ v\circ (b\otimes\mathbf{1})\|_{\ell_p}  &\text{ if } q\geq p,
    \end{cases}
\end{align} where $\mathbf{1}$ is the vector of all $1$s. Here for any two $v,w\in\mathbb{C}^{D_1\times D_2}$, $v\circ w$ indicates the entrywise or Hadamard product of $v$ and $w$.
\end{theorem}
\noindent We give a proof of this in  \cref{app:lplqVariationalFormulas} for sake of completeness. Moving to operators, one can also relate  Schatten (quasi-)norms by a variational expression that is analogous to the one above.

\begin{proposition}\label{prop:SpVariationalFormulas}

Let $0<q,p$ and set $\frac{1}{r}=\big|\frac{1}{q}-\frac{1}{p}\big|$, then for any $X\in\cS_q(\cH)$ it holds that
\begin{align}
    \|X\|_q=\begin{cases}
        \underset{\substack{X=\Pi_aX\Pi_b\\a,b\in\cS^+_{2r}(\cH)}}{\inf}\|a\|_{2r}\|b\|_{2r}\|a^{-1}Xb^{-1}\|_p  &\text{ if } q\leq p, \\
       \underset{\substack{X=\Pi_aX\Pi_b\\a,b\in\cS_{2r}^+(\cH)}}{\sup}\|a\|^{-1}_{2r}\|b\|^{-1}_{2r}\|aXb\|_p  &\text{ if } q\geq p.
    \end{cases}
\end{align} and when $X\geq0$ one may choose $a=b$.
\end{proposition}

\noindent Above, we denoted by $\mathcal{S}_q(\cH)$ the $q$-Schatten space of operators on the Hilbert space $\cH$, and $\mathcal{S}_q^+(\cH)$ the subset of positive operators. $\Pi_a$, resp.~$\Pi_b$, denotes the projection onto the support of $a$, resp.~$b$. The proof of the above identities relies only on the generalized Hölder inequality and can be found in \cref{app:SpVariationalFormulas}. Both of these should serve as ample motivation to conjecture that one can define 2-indexed Schatten quasi-norms by effectively copying the commutative variational formulas and changing the vectors to matrices and the $\ell_{q}$ norms to Schatten $\cS_q$ norms. We show that this naïve approach works except when $q,p$ are too far apart in the sense that $\big|\frac{1}{q}-\frac{1}{p}\big|>1$. 

\subsection{Main Results}

\subsubsection{A New Family of 2-Indexed Schatten Quasi-Norms}

None of the tools which are core to the definition of Pisier's norms carry over to the quasi-normed setting. Duality clearly only holds for indices in $[1,\infty]$ and breaks down below. Moreover, the construction of operator spaces builds manifestly on normed spaces. 
To reach a viable definition we start, hence, with factorization formulas which extend the ones for 2-indexed Schatten norms \cite{Devetak.2006, Beigi.2016} and which do not manifestly require norm or operator space properties. We define the following 2-indexed Schatten (quasi-)norms (see \cref{sec:Recap:QuasiNorms} for a recap on quasi-norms).
\begin{definition}[2-indexed Schatten (quasi-)norms]\label{def:Definition1}
Let $q,p\in(0,\infty]$ be such that $\frac{1}{r}=\big|\frac{1}{q}-\frac{1}{p}\big|\leq 1$, and $\cH_1,\cH_2$ two Hilbert spaces. Then for $q\leq p$, we define on $\mathcal{B}(\mathcal{H}_1\otimes \mathcal{H}_2)$
\begin{align}
    \|X\|_{(q,p)}:=\inf_{X_{12}=a_1Z_{12}b_1}\|a\|_{2r}\|b\|_{2r}\|Z\|_p \label{equ:Def1.1}
\end{align}
where the infimum is over operators $a_1=a\otimes \1$, $b_1=b\otimes \1$ acting non-trivially on $\cH_1$. Similarly, for $q\geq p$ we define
\begin{align}
    \|X\|_{(q,p)}:=\sup_{\underset{}{a,b}}\|a\|^{-1}_{2r}\|b\|^{-1}_{2r}\|a_1X_{12}b_1\|_p \,.\label{equ:Def1.2}
\end{align}
\end{definition}

\noindent While this definition clearly reproduces the normed setting when $q,p\geq 1$ \cite{Book.Pisier.1998, Devetak.2006, Beigi.2016, Fawzi.2026} it is not at all obvious that it suffices to give rise to a suitably well behaved family of indeed quasi-norms and that the condition $\big|\frac{1}{q}-\frac{1}{p}\big|\leq 1$ is the sufficient and necessary condition to do so.

The main goal of this paper is to prove important properties of these functionals $\|\cdot\|_{(q,p)}$ which demonstrate that these are indeed the correct notions of 2-indexed Schatten quasi-norms naturally extending Pisier's norms. To do this we will prove in \cref{sec:Operator.valued.quasi.norms} that they define families of $\min\{q,p,1\}$-normable quasi-norms (see Definition \ref{definekappanormable}).
Note that in the norm setting of $q,p\geq1$ this condition is always trivially fulfilled as $\min\{q,p,1\}=1$.
\begin{theorem}\label{thm:QuasiNomrms}
Take $0<q,p$ with $\big|\frac{1}{q}-\frac{1}{p}\big|\leq1$ and set $\kappa:=\min\{q,p,1\}$.
Then the expressions $\|\cdot\|_{(q,p)}$ define $\kappa$-normable quasi-norms on $\mathcal{B}(\mathcal{H}_1\otimes\cH_2)$. We denote the spaces $\mathcal{B}(\mathcal{H}_1\otimes \cH_2)$ endowed with these quasi-norms as $\cS_q[\cH_1,\cS_p(\cH_2)]$.
\end{theorem}

When $q,p\geq 1$ it is well established that they give rise to norms \cite{Book.Pisier.1998}, and when $q=p$ it is well known that the Schatten-$p$-spaces are $p$-normable quasi-normed spaces, so in that sense, \cref{thm:QuasiNomrms} constitutes the natural generalization of these spaces.

 In order to prove \cref{thm:QuasiNomrms}, we first prove two key structural results of these quasi-norms. The first and maybe most important one is about relational consistency. Explicitly, in \cref{def:Definition1} we defined the $\|\cdot\|_{(q,p)}$-quasi-norm in terms of factorizations through the Schatten-$p$-quasi-norm. This naturally raises the question whether these relations can be inverted, i.e. whether the Schatten-$p$-(quasi-)norm can be factored through the $\|\cdot\|_{(q,p)}$-(quasi-)norm and further whether the $\|\cdot\|_{(q,t)}$ quasi-norms can be factored through $\|\cdot\|_{(p,t)}$ quasi-norms. In the case where all indices $q,p,t\geq 1$, this is known to be true by Pisier's theorem \cite[Theorem 1.5, Lemma 1.7]{Book.Pisier.1998} (see also \cite[Lemma 3.1]{Fawzi.2026} for $\cX=\cS_t(\cH_2)$). Here, we show that this property extends beyond the norm regime.
\begin{theorem}\label{thm:qt.to.pt}
Let $q,p,t>0$ be s.t. $\max\big\{\big|\frac{1}{t}-\frac{1}{q}\big|,\big|\frac{1}{t}-\frac{1}{p}\big|\big\}\leq 1$ (we call such a triple compatible) then it holds that 
\begin{align} \label{equ:thm:qt.to.pt.inf}
    \|X\|_{(q,t)}&= \inf_{\underset{X_{12}=a_1Z_{12}b_1}{a,b\geq 0, Z}}\|a\|_{2r}\|b\|_{2r}\|Z\|_{(p,t)}  & \text{ if } q\leq p \\
    \|X\|_{(q,t)}&= \sup_{a,b\geq 0}\|a\|^{-1}_{2r}\|b\|^{-1}_{2r}\|a_1X_{12}b_1\|_{(p,t)} & \text{ if } q\geq p \label{equ:thm:qt.to.pt.sup}
\end{align} 
where $\frac{1}{r}=\big|\frac{1}{q}-\frac{1}{p}\big|$ and for any $q>0$ it holds that
\begin{align}
    \|X\|_{(q,q)} = \|X\|_q.
\end{align}
Moreover, when $X\geq0$, then one can restrict the optimizations each to $a=b^*$.
\end{theorem}

\noindent This key result shows that \cref{def:Definition1} is in some sense built on minimal assumptions, yet sufficient to prove the much a more rich structure included in the relations in \cref{thm:qt.to.pt}. In particular it demonstrates that one can relate these quasi-norms with each other in the most natural way, as long as the indices of the involved quasi-norms are all compatible, i.e. not too far apart. This important structural result is known to be true for operator-valued Schatten norms. Note that \eqref{equ:thm:qt.to.pt.inf} with $p=\infty, q\geq1$ is usually referred to as Pisier's formula \cite[Theorem 1.5]{Book.Pisier.1998} and for $q,p,t\geq1$  \eqref{equ:thm:qt.to.pt.sup} follows from duality between operator-valued Schatten normed spaces with Hölder dual indices. In the quasi-normed setting, however, neither the duality nor the operator space construction of the $(\infty,t)$-norm exist so we have to prove it from the ground up. 

 Next, we study the behavior of our quasi-norms on block-diagonal operators. We demonstrate that they simplify to $\ell_q$-quasi-norms of the $\|\cdot\|_p$-quasi-norms of the blocks, as one would expect and as is true in the normed setting $1\leq p,q$ \cite[Corollary 1.3]{Book.Pisier.1998}. We also construct a counterexample to this statement, for when $q\geq p$ with $\frac{1}{r}=\frac{1}{p}-\frac{1}{q}> 1$ implying that the condition $\frac{1}{r}\leq 1$ is necessary in that case.

\begin{theorem}\label{thm:cq.additivity}
Given an orthonormal basis $\{|i\rangle\}$ of $\cH_1$, let $X=\bigoplus_iX_i=\sum_i|i\rangle\langle i|\otimes X_i\in \cB(\cH_1\otimes \cH_2)$ and $0<q,p$ s.t. $\frac{1}{r}=\big|\frac{1}{q}-\frac{1}{p}\big|\leq 1$, then
\begin{align}\label{equ:cq.additivityintro}
    \bigg\|\bigoplus_i X_i\bigg\|_{(q,p)}= \left(\sum_i\|X_i\|^q_p\right)^{\frac{1}{q}}.
\end{align}
Furthermore, for $q\geq p$ \eqref{equ:cq.additivityintro} holds for all $X$ \emph{if and only if} $\frac{1}{p}-\frac{1}{q}\leq1$.
\end{theorem}
\noindent Lastly, for indices $0<q,p\leq 1$ we prove a reverse Hölder inequality, which generalizes the well-known one for Schatten spaces \cite[Lemma 3.3]{Book.Tomamichel.2016} and relates the $\|\cdot\|_{(q,p)}$-quasi-norm to its reverse Hölder dual quasi-norm $\|\cdot\|_{(-q^\prime,-p^\prime)}$.

\begin{lemma}[Reverse Hölder's inequality]\label{lem:GenRevHölder}
Let $0< q,p\leq 1$ with $\big|\frac{1}{q}-\frac{1}{p}\big|\leq 1$ and let $q^\prime, p^\prime$ be their Hölder dual indices, that is $\frac{1}{q}+\frac{1}{q^\prime}=\frac{1}{p}+\frac{1}{p^\prime}=1$. Then it holds that
\begin{align}\label{equ:GenRevHölder}
    \|X\|_{(q,p)}=\inf_{\underset{\Pi_A=\Pi_{XX^*}, \Pi_B=\Pi_{X^*X}}{A,B\geq0}}\|A\|^\frac{1}{2}_{(-q^\prime,-p^\prime)}\|B\|^\frac{1}{2}_{(-q^\prime,-p^\prime)}\|A^{-\frac{1}{2}}XB^{-\frac{1}{2}}\|_1,
\end{align} 
which for positive $X\geq 0$ simplifies to 
\begin{align}\label{equ:RevHölder}
    \|X\|_{(q,p)}=\inf_{\underset{\ker(Y)\subset\ker(X)}{Y\geq0}}\Tr[Y^*X]\, \|Y^{-1}\|_{(-q^\prime,-p^\prime)}.
\end{align}
\end{lemma}
\noindent Note that in the above lemma the Hölder dual indices $-q^\prime,-p^\prime\in[1,\infty]$ for $\frac{1}{2}\leq q,p\leq1$ and $-q^\prime,-p^\prime\in(0,\infty]$ for $0<q,p\leq1$.
See \cref{sec:Operator.valued.quasi.norms} for more details and proofs.

\subsubsection{Some Entropic Quantities via 2-Indexed Schatten Quasi-Norms}

In \cref{sec:Applications} we consider applications of these families of quasi-norms. First in \cref{sec:Apl:Entropies} we relate the above constructed family of quasi-norms to quantum information theoretic entropic quantities. Two particular examples we mention below are the optimized sandwiched conditional Rényi entropy $H_\alpha^\uparrow(A|B)$ and the sandwiched Rényi Umlaut information $\tilde{U}_\alpha(A;B)$. We prove properties of these objects via their quasi-norm representations in \cref{sec:Apl:Entropies}. 

An almost immediate consequence of \cref{def:Definition1} is that the optimized sandwiched conditional entropy for any index $1>\alpha\geq\frac{1}{2}$ can be expressed as
\begin{align}
    H_\alpha^\uparrow(A|B)_\rho=\frac{\alpha}{1-\alpha}\log\|\rho_{BA}\|_{(1,\alpha)},
\end{align}
which extends the well known and used result for $\alpha>1$ \cite{Beigi.2023, Fawzi.2026}. We also introduce a new sandwiched Umlaut information for $\alpha<1$, which we define as
\begin{align}
       \tilde{U}_\alpha(A;B)_\rho&:=\inf_{\sigma_B\geq0, \|\sigma\|_1= 1}D_\alpha(\rho_A\otimes\sigma_B\|\rho_{AB}) =\frac{\alpha}{\alpha-1}\log\|(\1_B\otimes\rho_A^\frac{1}{2})\rho_{BA}^{\frac{1-\alpha}{\alpha}}(\1_B\otimes\rho_A^\frac{1}{2})\|_{(\frac{\alpha}{1-\alpha},\alpha)}.
\end{align}
This is a novel expression, since for $\alpha>1$ it does not correspond to a norm. For further properties of these objects and a connection to the quantum Umlaut information \cite{Girardi.2025} see \cref{sec:Apl:Entropies}.

\subsubsection{A Notion of Completely Bounded Quasi-Norms}


Using the 2-indexed Schatten quasi-norms for arbitrary indices $0<q,p\leq\infty$, we define completely bounded quasi-norms and co-quasi-norms of linear maps between quasi-normed Schatten spaces, which we believe to be of independent interest. In \cref{sec:cb.mixed.norms} we define the following \emph{mixed quasi-norm} and \emph{completely bounded quasi-norm} of some linear map $\Phi$
\begin{align}
    \|\Phi\|_{\cb,q\to p}&:=\sup_E\|\id_E\otimes\Phi\|_{(t,q)\to (t,p)}, \qquad \|\id\otimes\Phi\|_{(t,q)\to (t,p)}:=\sup_{X}\frac{\|\id\otimes\Phi(X)\|_{(t,p)}}{\|X\|_{(t,q)}}
\end{align}
 for any compatible $0<q,p,t$. Then, we prove in \cref{lem:cb.norm.via.t} that this completely bounded quasi-norm is independent of the index $t>0$ as long as it is compatible with both $p,q$. This directly extends the normed setting $1\leq q,p,t$, for which this independence of $t$ is known by \cite[Lemma 1.7]{Book.Pisier.1998}. In addition to considering the usual (mixed) quasi-norm above we are also interested in analogous quantities involving a minimization. These are sometimes called \textit{minimum moduli} or \textit{co(-quasi-)norms} \cite{Mbekhta.1996, amelin.1973}. We define the \textit{mixed co-quasi-norm} and \textit{completely bounded co-quasi-norm} of a CP map $\Phi$ as
\begin{align}
    \|\Phi\|^+_{\cb,\co,q\to p}&:=\inf_E\|\id_E\otimes\Phi\|^+_{(t,q)\to (t,p)}, \qquad \|\id\otimes\Phi\|^+_{(t,q)\to (t,p)}:=\inf_{X\geq0}\frac{\|\id\otimes\Phi\|_{(t,p)}}{\|X\|_{(t,q)}}
\end{align} for any compatible $0<q,p,t$. In \cref{lem:cb.co.norm.via.t}, we prove again independence of this quantity on the index $t>0$, as long as it is compatible with both $0<q,p$.

In order to derive multiplicativity results for these completely bounded (co-)quasi-norms we first extend the important non-commutative Minkowski inequality to our 2-indexed quasi-normed setting. In \cref{sec:app:Minkowski.Inequality} we prove the following contractivity statement for the SWAP map. 
\begin{theorem}[A quasi-norm Minkowski inequality, informal, \cref{thm:QuasiNormMinkowskiInequality}]
For $0<q\leq p$ s.t. $\big|\frac{1}{q}-\frac{1}{p}\big|\leq 1$ and any tripartite operator $X\equiv X_{123}\in\cB(\cH_1\otimes \cH_2\otimes \cH_3)$, it holds that
\begin{align}
    \|X_{12;3}\|_{(p,q)}\leq \|X_{132}\|_{(p,q,p)}\,,
\end{align}
where the $\|\cdot\|_{(p,q,p)}$-quasi-norm is defined in \eqref{def:3.indexed.special} as a natural extension to \eqref{equ:Def1.2}, and $\|X_{12;3}\|_{(p,q)}$ here is to be understood as the joint system $\cH_1\otimes \cH_2$ being endowed with a $p$-norm, and $\cH_3$ endowed with a $q$-norm.
\end{theorem}

\noindent In the norm setting $1\leq q,p$, this inequality was proved via interpolation and properties of certain operator space tensor products \cite[Corollary 1.10]{Book.Pisier.1998}. There it constitutes a useful tool to prove multiplicativity statements, see e.g. \cite{Devetak.2006, Fawzi.2026}. In this work we do the same and use this non-commutative Minkowski inequality to derive multiplicativity results for mixed cb quasi-norms and cb co-quasi-norms of linear maps in \cref{sec:app:quasinorm.additivity}.

\subsubsection{New Quantum Information Multiplicativity Results}\label{sec:intro:multiplicativity}

 Next, we show that maximal output and completely bounded minimal output Rényi entropies for any index $\frac{1}{2}\leq p\leq \infty$ are additive and that the notion of \emph{complete reverse hypercontractivity} is indeed a well defined and useful one, in the sense that the corresponding co-quasi-norm is supermultiplicative under tensorproducts of channels. For more details on this see \cref{rem:cb.reverse.hypercontractivity}. 

For multiplicativity statements of the completely bounded quasi-norms see \cref{sec:app:mixed.norm.additivity} and for the completely bounded co-quasi-norms see \cref{sec:app:mixed.conorm.additivity}. One of these prior mentioned results is the following, which should be considered one of the main multiplicativity results of this paper.
\begin{theorem}\label{thm:cb.reverse hypercontractive.bound}
Let $q\geq p>0$ be s.t. $\frac{1}{p}-\frac{1}{q}\leq 1$ and $\Phi\,,\Psi$ be CP maps, then it holds that
\begin{align}
    \|\Phi\otimes\Psi\|^+_{\cb,\co,q\to p}\geq \|\Phi\|^+_{\cb,\co,q\to p}\cdot \|\Psi\|^+_{\cb,\co, q\to p}.
\end{align} 
\end{theorem}

\noindent One can think of this result as a generalization of \cite[Theorem 11]{Devetak.2006} to arbitrary positive indices. The inequality above is, from a quantum information perspective the non-trivial one, as the LHS is defined via a minimization. Even though this inequality is achievable, we leave the proof of this fact to future work \cite{Kochanowski.in.prep}, as it requires somewhat more advanced technical tools, which we believe deserve their own proper exposition.

\cref{thm:cb.reverse hypercontractive.bound} is important for two reasons, on the one hand it is vital for the concept of \emph{complete reverse hypercontractivity} which builds manifestly on the supermultiplicativity of the above completely bounded $q\to p$-co-quasi-norm (for $1> q\geq p\geq \frac{1}{2}$). For a more detailed discussion of this see \cref{rem:cb.reverse.hypercontractivity}.

On the other hand we consider the special case of $q=1$. There \cref{thm:cb.reverse hypercontractive.bound}  implies additivity of the completely bounded minimal output Rényi-$p$ entropy for any $p\geq\frac{1}{2}$ extending the result of \cite{Devetak.2006} for $p>1$.
\begin{corollary}\label{cor:Additivity.mib.cb.entropy}
Let $\Phi:Q_1\to A_1\,,\Psi:Q_2\to A_2$ be CP maps and $p\geq\frac{1}{2}$, then 
\begin{align}
\inf_E\inf_{\rho_{EQ_1Q_2}}H^\uparrow_p(A_1A_2|E)_{(\id_E\otimes\Phi\otimes\Psi)(\rho_{})}= \inf_{E,\rho_{EQ_1}}H^\uparrow_p(A_1|E)_{(\id_E\otimes\Phi)(\rho_{})}+\inf_{E,\rho_{EQ_2}} H^\uparrow_p(A_2|E)_{(\id_E\otimes\Psi)(\rho_{})}.
\end{align}
\end{corollary}
\noindent For proofs and more details, see \cref{sec:app:mixed.conorm.additivity}.
\begin{remark}\label{ref:concurrent.work}
    While finalizing this paper, we were made aware of the concurrent work \cite{Li.2026} in which the authors prove, among others, the special case $q=1$ of \cref{thm:cb.reverse hypercontractive.bound} and thus \cref{cor:Additivity.mib.cb.entropy} independently. In this special case one can define the $\cb,\co,1\to p$ norm without a full theory of 2-indexed Schatten quasi-norms, as both the in- and output norms $\|\cdot\|_{(p,1)}=\|\tr_2[\cdot]\|_p$, $\|\cdot\|_{(p,p)}=\|\cdot\|_p$, respectively collapse to closed form expressions, which is the approach of \cite{Li.2026}. See \cref{cor:q.1.PartialTrace} and \cref{lem:cb.co.norm.via.t} for details on these special cases. Since our result is more general and adaptable to different settings, it requires this more advanced machinery.
\end{remark}

\noindent In \cref{sec:app:quasinorm.additivity}, we also prove that the regular (non-completely bounded) maximal output Rényi-$\alpha$ entropy is additive under tensor products of channels, a result we are unaware of having been previously established, except for the von Neumann case, where it is a direct consequence of the  subadditivity.
\begin{corollary}\label{cor:max.entropy.additivity}
Let $\frac{1}{2}\leq p\leq \infty$, then it holds for CP maps $\Phi:Q_1\to A_1\,, \Psi:Q_2\to A_2$ that
\begin{align}
    \sup_{\rho_{Q_1Q_2}}H_p(A_1A_2)_{(\Phi\otimes\Psi)(\rho_{Q_1Q_2})} = \sup_{\rho_{Q_1}}H_p(A_1)_{\Phi(\rho_{Q_1})}+\sup_{\rho_{Q_2}}H_p(A_2)_{\Psi(\rho_{Q_2})}.
\end{align}
\end{corollary}

We also extend, with \cref{thm:quasi.chain.rule} an influential chain rule \cite[Lemma 5]{Devetak.2006} to indices $\frac{1}{2}\leq q,p$, which among others implies \cref{cor:postivie.mixed.optimizers}: CP maps achieve their mixed $q\to p$ quasi-norms on positive operators for any $0<q,p$. This extends an important result of \cite{Watrous.2004} to the quasi-normed setting.



\section{Preliminaries}\label{sec:Preliminaries}

Before we begin we set our notation and central mathematical notions.
    
\subsection{Notation}\label{sec:Notation}
Hilbert spaces, denoted with $\cH, \cK$, will generically be assumed finite dimensional in this work, unless explicitly stated otherwise.
We denote by $\cD(\cH)$ the set of quantum states and $\mathcal{B}(\mathcal{H},\mathcal{K})$ that of linear operators from $\mathcal{H}$ to $\mathcal{K}$, with $\mathcal{B}(\mathcal{H})\equiv \mathcal{B}(\mathcal{H},\mathcal{H})$. $\mathcal{U}(\mathcal{H})$ stands for the group of unitary operators over $\mathcal{H}$. The identity operator on $\cB(\cH)$ is denoted with $\1$. We denote the adjoint of some operator $X\in\cB(\cH,\cK)$ with respect to the Hilbert inner product as $X^*\in\cB(\cK,\cH)$. An operator $a\in\cB(\cH)$ is called positive $a\geq0$, whenever there exists some $y\in\cB(\cH)$ s.t. $a=yy^*$.

Given some positive operator $a\geq0$ we denote its support projection as $\Pi_a$, that is the minimal projection s.t. $a=\Pi_aa\Pi_a$. Hence for any operator $X$ the projection onto its image (left support) is $\Pi_{XX^*}$ and the one onto its support (right support) as $\Pi_{X^*X}$. Given a bipartite operator $X\in\cB(\cH\otimes\cK)$, then its \emph{left},  respectively \emph{right marginal support projections} $L_X,R_X$ on $\cH$ are the minimal projections s.t.
\begin{align}\label{equ:full.to.marginal.support}
    X=(L_X\otimes\1)X(R_X\otimes\1).
\end{align} We will interchangeably call them projections onto the marginal image, respectively, support and the marginal in question will be clear from the context.
They are also equally characterized by being the minimal projections which satisfy $\Pi_{XX^*}\leq L_X\otimes\1, \quad \Pi_{X^*X}\leq R_X\otimes\1$.

When inverting an operator $a\geq0$ we will generically be using to the generalized Moore-Penrose inverse, which is a bounded operator $a^{-1}$ s.t. $a^{-1}a=aa^{-1}=\Pi_a$. Effectively it is the inverse on the support of that operator and $0$ elsewhere.

We use the convection that $|X|:=\sqrt{X^*X}$ for any operator $X$, which implies that its polar decompositions are $X=U|X|=|X^*|V$, for some unitaries $U,V$.

A linear map $\Phi:\cB(\cH)\to \cB(\cK)$ is called positive if it maps positive operators to positive ones. It is called completely positive (CP) if $\id_n\otimes\Phi:\cB(\mathbb{C}^n\otimes\cH)\to \cB(\mathbb{C}^n\otimes\cK)$ is positive for all $n\in\mathbb{N}$.
A map is called unital (U) if it preserves the identity, i.e. $\Phi(\1_{\cH})=\1_{\cK}$ and trace-preserving (TP) if it holds that $\Tr[\Phi(X)]=\Tr[X]$ $\forall X\in\cB(\cH)$. 
We call a CPTP map a quantum channel. We denote the identity channel as $\id:\cB(\cH)\to\cB(\cH)$.

\subsection{Quasi-Norms}\label{sec:Recap:QuasiNorms}

An important concept used throughout this work is that of a quasi-norm.

\begin{definition}\label{definekappanormable}
A $\kappa$-normable quasi-norm, for $0<\kappa\leq1$ on some vector space $\mathcal{V}$ over the complex numbers is a map 
$\|\cdot\|:\mathcal{V}\to [0,\infty)$ which is
\begin{enumerate}
    \item \emph{absolutely homogeneous of degree 1}: $\|\lambda v\|=|\lambda|\|v\|$ for all $\lambda\in\mathbb{C}, v\in\mathcal{V}$,
    \item \emph{positive definite}: $\|v\|=0$ iff $v=0$,
    \item \textit{$\kappa$-subadditive}: $\|v+w\|^\kappa\leq \|v\|^\kappa+\|w\|^\kappa$ for all $v,w\in\mathcal{V}$.
\end{enumerate}
The last property is called $\kappa$-subadditivity and implies in particular that
\begin{align}\label{equ:modulus.of.concavity}
    \|v+w\|\leq 2^{\frac{1}{\kappa}-1}(\|v\|+\|w\|)
\end{align} holds for any $v,w\in\mathcal{V}$, where $2^{\frac{1}{\kappa}-1}$ is called the modulus of concavity of the quasi-norm \cite{book:Kalton.03}. 
\end{definition}
The important Aoki-Rolewicz theorem \cite{Aoki.1942, Rolewicz.1972} essentially implies that any quasi-norm with a modulus of concavity $C=2^{\frac{1}{\kappa}-1}$ as in \eqref{equ:modulus.of.concavity} is equivalent to a $\kappa$-normable quasi-norm for that $0<\kappa\leq 1$. Hence any quasi-norm is equivalent to a $\kappa$-normable one, for some $0<\kappa\leq 1$.
Note that if a quasi-norm is $\kappa$-normable, then it is also $\mu$-normable for all $0<\mu\leq\kappa$ and that in particular any norm a $1$-normable quasi norm.

The Schatten spaces $\cS_p(\cH)$ and the sequence spaces $\ell_p(I)$ with $I\subseteq\mathbb{N}$ are examples of $p$-normable quasi-normed spaces for $0<p<1$ and normed spaces for $p\geq1$. These are for any $0<p$, in finite dimensions given by
\begin{align}
    \cS_p(\cH)=(\cB(\cH),\, \|\cdot\|_p), \qquad \|X\|_p:=\Tr[|X|^p]^\frac{1}{p} 
\end{align} and for finite sequences by
\begin{align}
    \ell_p(I):=({\mathbb{C}^{I}},\, \|\cdot\|_{\ell_p}) \qquad \|v\|_{\ell_p}:=\Big(\sum_{i\in I}|v_i|^p\Big)^\frac{1}{p}.
\end{align}

These quasi-norms satisfy the following important generalized Hölder inequality.
\begin{lemma}[Generalized Hölder's inequality \cite{borell1982positivity, tomamichel2014relating}]\label{lem:GeneralizedHölder}
Let $0<p,q,r\leq \infty$ be s.t. $\frac{1}{r}=\frac{1}{p}+\frac{1}{q}$, then for any two vectors $v\in\ell_p,\, w\in\ell_q$ it holds that
\begin{align}
    \|v\circ w\|_{\ell_r}\leq \|v\|_{\ell_p}\|w\|_{\ell_q},
\end{align} where $w\circ v$ is the Hadamard product of the two vectors, i.e. $(w\circ v)_i=w_iv_i$. Equality holds if and only if $|v_i|^p \propto |w_i|^q$ for all $i$.
Likewise for any two operators $X\in\cS_q ,Y\in\cS_p$ it holds that
\begin{align}
    \|XY\|_r\leq \|X\|_q\|Y\|_p.
\end{align}
\end{lemma}

\section{2-Indexed Schatten Quasi-Norms}\label{sec:Operator.valued.quasi.norms}

The goal of this section is to study the \textit{2-indexed Schatten quasi norms}. The main part of this section is to justify the proposed \Cref{def:Definition1}, while deriving some of its core properties.

\subsection{Basic Properties}\label{sec:basic.properties}
We begin with some basic consequences of \Cref{def:Definition1}. The first one is that, due to the isometric invariance of the Schatten (quasi-)norms, the 2-indexed Schatten quasi-norms are invariant under local isometries. For notational simplicity we present this for unitaries:
\begin{proposition}\label{prop:LocalIsomInv}
Let $q,p$ be s.t. $\big|\frac{1}{q}-\frac{1}{p}\big|\leq 1$ and let $U,V\in\mathcal{U}(\mathcal{H}_{1})$ or $\in\mathcal{U}(\mathcal{H}_{2})$ respectively. Then it holds that
\begin{align}
    \|U_1X_{12}V_1\|_{(q,p)} = \|X\|_{(q,p)} \qquad \text{ and }\qquad     \|U_2X_{12} V_2\|_{(q,p)} = \|X\|_{(q,p)}.
\end{align}
\end{proposition}
\begin{proof}
In the case where $q\leq p$ we have by definition
\begin{align}
   \|U_1X_{12}V_1\|_{(q,p)} &= \inf_{U_1X_{12}V_1=a_1Z_{12}b_1}\|a\|_{2r}\|b\|_{2r}\|Z\|_p \\ &= \inf_{X_{12}=U_1^*a_1Z_{12}b_1V_1^*}\|a\|_{2r}\|b\|_{2r}\|Z\|_p \\
   &= \inf_{X_{12}=c_1Z_{12}d_1}\|c\|_{2r}\|d\|_{2r}\|Z\|_p \\
   &= \|X\|_{(q,p)}.
\end{align} 
Similarly,
\begin{align}
    \|U_2X_{12}V_2\|_{(q,p)} &= \inf_{U_2X_{12}V_2=a_1Z_{12}b_1}\|a\|_{2r}\|b\|_{2r}\|Z\|_p \\ 
    &= \inf_{{X_{12}=a_1U_2^*Z_{12}V_2^*b_1}}\|a\|_{2r}\|b\|_{2r}\|Z\|_p \\ &= \inf_{{X_{12}=a_1W_{12}b_1}}\|a\|_{2r}\|b\|_{2r}\|U_2W_{12}V_2\|_p \\ 
    &= \inf_{X_{12}=a_1W_{12}b_1}\|a\|_{2r}\|b\|_{2r}\|W\|_p \\
    &=\|X\|_{(q,p)}.
\end{align}
The proof on the case $q\geq p$ is analogous:
\begin{align}
    \|U_1X_{12}V_1\|_{(q,p)} &= \sup_{a,b}\|a\|^{-1}_{2r}\|b\|^{-1}_{2r}\|a_1U_1X_{12}V_1b_1\|_p \\ &= \sup_{c,d}\|cU^*\|^{-1}_{2r}\|V^*d\|^{-1}_{2r}\|c_1X_{12}d_1\|_p \\
    &= \|X\|_{(q,p)}
\end{align} and
\begin{align}
    \|U_2X_{12}V_2\|_{(q,p)} &= \sup_{a,b}\|a\|^{-1}_{2r}\|b\|^{-1}_{2r}\|a_1U_2X_{12}V_2b_1\|_p \\ &= \sup_{a,b}\|a\|^{-1}_{2r}\|b\|^{-1}_{2r}\|U_2a_1X_{12}b_1V_2\|_p \\ &= \sup_{a,b}\|a\|^{-1}_{2r}\|b\|^{-1}_{2r}\|a_1X_{12}b_1\|_p\\
    &= \|X\|_{(q,p)}.
\end{align}
\end{proof}
\noindent A direct consequence of the local isometric invariance and the polar decomposition is that we may assume the optimizers $a,b$ in \cref{def:Definition1} to be positive.
\begin{corollary}[Positive optimizers]
Let $0<q,p$ be compatible indices, then
\begin{align}
\|X\|_{(q,p)}&=\left\{
\begin{aligned}
&\inf_{\underset{X_{12}=a_1Z_{12}b_1}{a,b \geq0\, , Z}}\|a\|_{2r}\|b\|_{2r}\|Z\|_p&q\le p \\&\sup_{{a,b\geq0}}\|a\|^{-1}_{2r}\|b\|^{-1}_{2r}\|a_1X_{12}b_1\|_p&q\ge p
    \end{aligned}\right.~~.
\end{align}
\end{corollary}
\noindent Throughout the rest of this paper, we will be assuming the optimizers and the operators over which we optimize to be positive semidefinite.
\begin{proof}
In the case of the infimum this is seen by considering the polar decompositions  $a=|a^*|U_a,b=V_b|b|$ for a given decomposition $a,b,Z$. Then clearly $(|a^*|,|b|,(U_a)_1Z_{12}(V_b)_1)$ is also a valid decomposition with, by the above proof, equal objective value, but now positive $a,b$. In the case of the supremum use the polar decompositions $a=V_a|a|,b=|b^*|U_b$ and observe that replacing $a,b$ by $|a|,|b^*|$ doesn't change the objective value.
\end{proof}
\noindent Another important observation is that we may assume the optimizers to preserve the marginal image and support of $X$. This also allows for a simpler form of the norm which leads to an explicit minimization over $Z$.

\begin{lemma}\label{lem:SupportConditions}
Given $X\in\cB(\cH_1\otimes\cH_2)$, denote with $L_X,R_X$ its left and right marginal support projections, i.e. the minimal projections s.t. $X_{12}=(L_X)_1X_{12}(R_X)_1$. Then for $0<q\leq p$ s.t. $\frac{1}{r}=\big|\frac{1}{q}-\frac{1}{p}\big|\leq 1$, it holds that
\begin{align}
    \|X\|_{(q,p)}&=\inf_{\underset{\Pi_a=L_X, \Pi_b=R_X}{a,b \geq0}}\|a\|_{2r}\|b\|_{2r}\|a^{-1}_1X_{12}b_1^{-1}\|_p,
\end{align} 
where $\Pi_{a,b}$ is the projection onto the support of $a,b$ respectively, and $a^{-1},b^{-1}$ are the generalized Moore-Penrose inverses of $a,b$., i.e. $aa^{-1}=a^{-1}a=\Pi_a$. Likewise for $0<p\leq q$ s.t. $\frac{1}{r}=\big|\frac{1}{q}-\frac{1}{p}\big|\leq 1$, it holds that
\begin{align}
    \|X\|_{(q,p)}&=\sup_{\underset{\Pi_a=L_X,\Pi_b=R_X}{a,b \geq0}}\|a\|^{-1}_{2r}\|b\|^{-1}_{2r}\|a_1X_{12}b_1\|_p.
\end{align}
Further for $q\geq p$, when $r>1$ then the stronger statement that the optimal $a,b$ satisfy $\Pi_a=L_X, \Pi_b=R_X$ holds.
\end{lemma}

\begin{remark}
Effectively \cref{lem:SupportConditions} and the condition $\Pi_a=L_X,\Pi_b=R_X$ implies that we can treat the optimization variables $a,b$ as strictly positive matrices on a fixed subspace governed entirely by $X$. Further the marginal image/ support of $X$ is inherited by $a_1X_{12}b_1$ and $a_1^{-1}X_{12}b_1^{-1}$ for such $a,b$ allowing for iteration. 
For full rank $X$ this means in particular that $a,b>0$ can be assumed.  We will use this powerful fact throughout the rest of the paper, sometimes  without explicitly writing this condition out for notational simplicity. 
\end{remark}

\begin{proof}
For the case $q\leq p$, on the one hand, if $a,b$ are s.t. $X_{12}=(\Pi_a)_1X_{12}(\Pi_b)_1=a_1a_1^{-1}X_{12}b_1^{-1}b_1$, which is equivalent to $\Pi_a\geq L_X, \Pi_b\geq R_X$, then the triple $a,b,Z:=a^{-1}_1X_{12}b_1^{-1}$ satisfies $X_{12}=a_1Z_{12}b_1$ and thus
\begin{align}
    \|X\|_{(q,p)}\leq \inf_{\underset{X=(\Pi_a)_1X_{12}(\Pi_b)_1}{a,b \geq0}}\|a\|_{2r}\|b\|_{2r}\|a_1^{-1}X_{12}b_1^{-1}\|_p.
\end{align} On the other hand, given any compatible triple $a,b,Z$, then
\begin{align}
    X_{12}=a_1Z_{12}b_1 \implies a^{-1}_1X_{12}b^{-1}_1 = (\Pi_a)_1Z_{12}(\Pi_b)_1 \text{ and } \Pi_a\geq L_X, \Pi_b\geq R_X
\end{align} but since $\|(\Pi_a)_1Z_{12}(\Pi_b)_1\|_p\leq \|\Pi_a\|_\infty \|Z\|_p \|\Pi_b\|_\infty\leq \|Z\|_p$  and this holds for any such compatible triple, it follows that
\begin{align}
\inf_{\underset{\Pi_a\geq L_X, \Pi_b\geq R_X}{a,b \geq0}}\|a\|_{2r}\|b\|_{2r}\|a^{-1}_1X_{12}b^{-1}_1\|_p \leq \inf_{\underset{X=a_1Z_{12}b_1}{a,b \geq0}}\|a\|_{2r}\|b\|_{2r}\|Z\|_p = \|X\|_{(q,p)}.
\end{align} 

\noindent To see that it suffices to consider $\Pi_a=L_X, \Pi_b=R_X$, take for some $a,b,Z$ with $\Pi_a\geq L_X, \Pi_b\geq R_X$ and $X=a_1Zb_1$ define $a^\prime:=(L_Xa^2L_X)^\frac{1}{2}, b^\prime:=(R_Xb^2R_X)^\frac{1}{2}$ which satisfy $\Pi_{a^\prime}=L_X\leq \Pi_a$ and $\Pi_{b^\prime}=R_X\leq \Pi_b$. By the polar decomposition of $L_Xa$ there exist unitaries s.t.
\begin{align}
    L_Xa = |(L_Xa)^*|U=a^\prime U, \qquad bR_X = V|bR_X|= Vb^\prime
\end{align}
Since $X= (L_X)_1X(R_X)_1$ this yields the factorization
\begin{align}
   X_{12}=(L_Xa)_1Z_{12}(bR_X)_1=a^\prime_1\,(UZV)_{12}\,b^\prime_1, \qquad\text{where } (UZV)_{12}:=U_1Z_{12}V_1,
\end{align}
and consequently, using $a^{\prime\,-1}a^\prime=L_X$ and $b^\prime b^{\prime\,-1}=R_X$, 
\begin{align}
\|a^{\prime\,-1}_1X_{12}b^{\prime\,-1}_1\|_p=\|(L_X)_1U_1Z_{12}V_1(R_X)_1\|_p\leq \|U_1Z_{12}V_1\|_p=\|Z\|_p=\|a^{-1}_1X_{12}b^{-1}_1\|_p,
\end{align}
where we used Hölders inequality valid for any $p>0$, with $\|L_X\|_\infty=\|R_X\|_\infty=1$ and unitary invariance of the Schatten quasi-norm. Furthermore
\begin{align}
    \|a^\prime\|_{2r}=\|L_Xa^2L_X\|^\frac{1}{2}_{r}=\|(L_Xa)(L_Xa)^*\|^\frac{1}{2}_{r}=\|L_Xa\|_{2r}\leq \|L_X\|_\infty\|a\|_{2r}=\|a\|_{2r},
\end{align}
and analogously $\|b^\prime\|_{2r}=\|bR_X\|_{2r}\leq\|b\|_{2r}$. Hence replacing $(a,b)$ by $(a^\prime,b^\prime)$ does not increase the objective, which finishes the argument. In particular, the infimum may be restricted to those $a,b$ whose support projections are the left and right marginal support projections of $X$, $\Pi_a=L_X$ and $\Pi_b=R_X$, i.e.~the minimal projections such that $X=(\Pi_a\otimes\1)X(\Pi_b\otimes\1)$.

For the case $q\geq p$, restricting the optimization to be only over $a,b\geq0$ s.t. $\Pi_a\leq L_X$ and $\Pi_b\leq R_X$ does not decrease the objective value by  an analogous argument to the above. For some $a,b$, setting $a^\prime:=aL_X, b^\prime:=R_Xb$ can only increase the objective value, since
\begin{align}
    \|a^\prime\|_{2r}=\|aL_X\|_{2r}\leq \|a\|_{2r}\|L_X\|_\infty=\|a\|_{2r}
\end{align} and the same holds for $b$, hence
\begin{align}
\|a\|^{-1}_{2r}\|b\|_{2r}^{-1}\|a_1X_{12}b_1\|_p &=  \|a\|^{-1}_{2r}\|b\|_{2r}^{-1}\|(a L_X)_1X_{12}(R_Xb)_1\|_p \\ &\leq \|a^\prime\|^{-1}_{2r}\|b^\prime\|_{2r}^{-1}\|a^\prime_1X_{12}b^\prime_1\|_p.
\end{align}
By polar decomposition and local unitary independence it follows that replacing $a^\prime = U|a^\prime|,\, b^\prime=|b^{\prime*}|V$ achieve the same value and $|a^\prime|=\sqrt{a^{\prime*} a^\prime}=\sqrt{L_Xa^2L_X}\geq0$ satisfies $\Pi_{|a^\prime|}\leq L_X$ and similarly for $\Pi_{|b^\prime|}\leq R_X$. The fact that we can restrict optimization now to only operators $a,b$ with $\Pi_a=L_X, \Pi_b=R_X$ stems from the norm continuity of objective functional in $a,b$ and the fact that all $a,b$ with $\Pi_a\leq L_X$ and $\Pi_b\leq R_X$, can be approximated (in norm) by the prior $a,b$. So we have
\begin{align}
    \sup_{a,b\geq0}\|a\|^{-1}_{2r}\|b\|^{-1}_{2r}\|a_1X_{12}b_1\|_p &= \sup_{\underset{\Pi_a\leq L_X, \Pi_b\leq R_X}{a,b\geq0}}\|a\|^{-1}_{2r}\|b\|^{-1}_{2r}\|a_1X_{12}b_1\|_p \\&= \sup_{\underset{\Pi_a= L_X, \Pi_b= R_X}{a,b\geq0}}\|a\|^{-1}_{2r}\|b\|^{-1}_{2r}\|a_1X_{12}b_1\|_p.
\end{align}
For a proof of the claim about the optimizers of this maximum in the case where $r>1$, see \cref{app:MarginalSupports}.
\end{proof}

\noindent Lastly, it turns out that these factorization expressions simplify when the operator $X\geq0$ is positive. This is a property which is well established \cite{Devetak.2006} and exploited in the normed setting, in particular implying a fruitful connection between operator space theory and optimized sandwiched Rényi conditional entropies \cite{Devetak.2006, Beigi.2023, Fawzi.2026}.
\begin{lemma}\label{lem:Positive.Symmetrie}
Let $0<q,p$ be s.t. $\big|\frac{1}{q}-\frac{1}{p}\big|\leq 1$, then if $X\geq0$ it holds that
\begin{align}
     \|X\|_{(q,p)}&=\inf_{\underset{X=(a\otimes\1)Z(a\otimes\1)}{a\geq0\, , Z}}\|a\|^2_{2r}\|Z\|_p = \inf_{\underset{\Pi_a=L_X=R_X}{a,b\geq0}}\|a\|^2_{2r}\|a^{-1}_1X_{12}a^{-1}_1\|_p  &\text{ for } q\leq p, \\
    \|X\|_{(q,p)}&=\sup_{\underset{\Pi_a=L_X=R_X}{a\geq0}}\|a\|^{-2}_{2r}\|a_1X_{12}a_1\|_p,  &\text{ for } q\geq p.
\end{align}
\end{lemma}

\begin{proof}
The proof of the case $q\geq p$ is the same as for the normed setting, see e.g.~\cite{Devetak.2006}. For convenience we repeat it here. Denote $A:=a_1,\, B:=b_1$ and let $X = xx^*$. By generalized Hölder's inequality it follows that
\begin{align}
    \|AXB\|_p=\|Axx^*B\|_p\leq \|Ax\|_{2p}\|x^*B\|_{2p}=\|AXA^*\|^\frac{1}{2}_p\|B^*XB\|^\frac{1}{2}_p\leq \max\{\|AXA^*\|_p, \|B^*XB\|_p\}, 
\end{align} from which it directly follows that we may restrict the optimization to $B=A^*$.

The case of $q\leq p$ is a little more involved. In the normed setting, see e.g.~\cite[Proposition 5.2 (ii)]{Bardet.2024} the proof of this fact goes manifestly through operator-valued Schatten spaces and in particular the space $\cS_\infty[\cS_p]$. As this quasi-norm is not well defined for $q<1$, as then $q$ is not compatible with $\infty$ since $\frac{1}{q}-\frac{1}{\infty}>1$, we find another proof. We will show that if $a,b$ are suitable operators involved in the minimization, then $c:=\frac{a+b}{2}$ is also suitable and achieves the same objective value. Note first that since $X=X^*$ it follows that $L_X=R_X$. To do this we first rewrite
\begin{align}
    \|X\|^p_{(q,p)}&=\inf_{\underset{\Pi_a=L_X=\Pi_b}{a,b\geq0}}\|a\|^p_{2r}\|b\|^p_{2r}\|a_1^{-1}Xb_1^{-1}\|_p^p \\ &=\inf_{\underset{\Pi_a=L_X=\Pi_b}{a,b\geq0,\, \|a\|_1=\|b\|_1\leq1}}\underbrace{\|(a^{-\frac{1}{2r}}\otimes\1)X(b^{-\frac{1}{2r}}\otimes\1)\|_p^p}_{=:G(a,b)} \\&\equiv \inf_{\underset{\Pi_a=L_X=\Pi_b}{a,b\geq0,\, \|a\|_1=\|b\|_1\leq1}}G(a,b).
\end{align}
Due to \cite[Theorem 1.1 (2)]{Zhang.2020} $(a,b)\mapsto G(a,b)$ is jointly convex as $0\leq \frac{1}{r}\leq 1$, so in particular
\begin{align}
    G(c,c)= G\left(\frac{a+b}{2}, \frac{b+a}{2}\right)\leq \frac{1}{2}(G(a,b)+G(b,a))=G(a,b),
\end{align} where the last equality follows due to self-adjointness of $X=X^*$, which implies $G(a,b)=G(b,a)$.
In addition we have due to $\Pi_a=\Pi_b$ that $c\geq0$ with $\Pi_c=\Pi_a=\Pi_b=L_X$ and $\|c\|_1= \|\frac{a+b}{2}\|_1\leq \frac{1}{2}(\|a\|_1+\|b\|_1)\leq 1$ for all $a,b$ with $\|a\|_1=\|b\|_1\leq 1$, so taking the infimum yields
\begin{align}
   \inf_{\underset{\Pi_a=\Pi_b=L_X}{a,b\geq0,\, \|a\|_1=\|b\|_1\leq1}}G(a,b) &\geq  
   \inf_{\underset{\Pi_c=L_X=R_X}{c\geq0,\, \|c\|_1\leq1}}G(c,c),
\end{align} which is what we wanted to show. 
\end{proof}

\noindent An important consequence of \cref{lem:Positive.Symmetrie} and \cref{lem:SupportConditions} is that for $X\geq0$, the 2-indexed Schatten-quasi-norm acts like a partial trace if the second index is $1$. 
\begin{corollary}\label{cor:q.1.PartialTrace}
Let $q\geq \frac{1}{2}$ and $X\geq0$, then it holds that
\begin{align}\label{equ:q.1.PartialTrace}
    \|X\|_{(q,1)}= \|\tr_2[X]\|_{q}.
\end{align}
\end{corollary}
\begin{proof}
For $\frac{1}{2}\leq q<1$ observe that by positivity of $X$ and \cref{lem:Positive.Symmetrie} it follows that
\begin{align}
    \|X\|_{(q,1)}&=\inf_{a\geq0}\|a\|^2_{2r}\|a_1^{-1}X_{12}a_1^{-1}\|_1 \\  &= \inf_{a\geq0}\|a\|^2_{2r}\Tr_{12}[a_1^{-1}X_{12}a_1^{-1}] \\ &=  \inf_{a\geq0}\|a\|^2_{2r}\Tr_1[a^{-1}\tr_2[X_{12}]a^{-1}] \\
    &= \|\tr_2[X]\|_q.
\end{align} 
For the norm case $q\geq 1$, which is proven exactly as above we refer to \cite{Devetak.2006}.
\end{proof}

\subsection{Simplification for Block Diagonal Operators}

A natural property one would ask for a sensible notion of 2-indexed Schatten quasi-norms is that on block-diagonal operators, also called classical-quantum operators in quantum information theory, they behave as $\ell_q$-norms of the $\cS_p$ norms of the individual blocks. This is a basic property of the Pisier norm \cite[Corollary 1.3]{Book.Pisier.1998}, which however does not seem to be trivially derived from \cref{def:Definition1}. In fact, we demonstrate that the condition $|\frac{1}{q}-\frac{1}{p}|\leq1$ is sufficient and necessary for this property to hold. 
\begin{theorem*}[\cref{thm:cq.additivity}]
Let $\operatorname{dim}(\cH_1)=d$, let $X=\bigoplus_{i=1}^dX_i$ be a block-diagonal element of $\cB(\cH_1\otimes \cH_2)$, and $0<q,p$ s.t. $\frac{1}{r}=\big|\frac{1}{q}-\frac{1}{p}\big|\leq 1$, then
\begin{align}\label{equ:cq.additivity}
    \bigg\|\bigoplus_{i=1}^d X_i\bigg\|_{(q,p)}= \left(\sum_{i=1}^d\|X_i\|^q_p\right)^{\frac{1}{q}}.
\end{align}
Furthermore for $q\geq p$, \eqref{equ:cq.additivity} holds \emph{only if} $\frac{1}{p}-\frac{1}{q}\leq1$.
\end{theorem*}
\noindent Before we give a proof of this statement we present an important consequence.
\begin{corollary}\label{tensorproduct}
Let $0<q,p$ s.t. $\frac{1}{r}=\big|\frac{1}{q}-\frac{1}{p}\big|\leq 1$, then for any $X\in\cB(\cH_1),Y\in\cB(\cH_2)$, it holds that
\begin{align}
    \|X\otimes Y\|_{(q,p)}=\|X\|_q\|Y\|_p.
\end{align}
\end{corollary}
\begin{proof}
Given any $X,Y$, let $X=UDV$ be its singular value decomposition with $D=\sum_ix_i|i\rangle\langle i|$, then by local isometric invariance, Proposition \ref{prop:LocalIsomInv}, and \cref{thm:cq.additivity} we have
\begin{align}
  \|X\otimes Y\|_{(q,p)} \!=\! \|D\otimes Y\|_{(q,p)}\! =\! \|\bigoplus_i(x_iY)\|_{(q,p)}\!=\! \left(\sum_i\|x_iY\|^q_{p}\right)^\frac{1}{q} \!=\! \left(\sum_i|x_i|^q\right)^\frac{1}{q}\|Y\|_p\!=\!\|X\|_q\|Y\|_p.
\end{align}
\end{proof}

\begin{proof}[Proof of Theorem \ref{thm:cq.additivity}]
The core idea of the proof consists in showing that the optimizers in the variational expressions that define these quasi-norms can be chosen block diagonal. On the one hand, the minimization, respectively maximization, give natural upper, respectively lower bounds on $\left\|\bigoplus_i X_i\right\|_{(q,p)}$ in terms of the RHS of \eqref{equ:cq.additivity}. To get the reverse inequalities we will use joint concavity, respectively convexity to establish the monotonicity of the norms with respect to the pinching map onto the diagonal blocks of the first system. Moreover, the needed joint concavity, resp. convexity are known to hold if and only if $0\leq\frac{1}{r}\leq 1$.

\medskip

Consider first the case $q\leq p$, then by \cref{equ:Def1.1} we have on the one hand that
\begin{align}
    \left\|\oplus_i X_i\right\|_{(q,p)} &=\inf_{a,b\geq0}\|a\|_{2r}\|b\|_{2r}\|a_1^{-1}(\oplus_iX_i)b_1^{-1}\|_{p} \\ 
    &\leq \inf_{\underset{X_i\neq 0\implies a_i,b_i>0}{a=\oplus_ia_i,b=\oplus_ib_i\geq0}}\|(a_i)_i\|_{\ell_{2r}}\|(b_i)_i\|_{\ell_{2r}}\|\oplus_ia^{-1}_iX_ib^{-1}_i\|_p \\ 
    &= \inf_{\underset{X_i\neq 0\implies a_i,b_i>0}{\{a_i,b_i\geq0\}_i}}\|(a_i)_i\|_{\ell_{2r}}\|(b_i)_i\|_{\ell_{2r}}\| (|a^{-1}_ib^{-1}_i|\|X_i\|_p)_i\|_{\ell_p} \\ 
    &= \| (\|X_i\|_p)_i \|_{\ell_q}\\
    &= \left(\sum_i\|X_i\|_p^q\right)^\frac{1}{q}.
\end{align}
For the other direction, consider first the pinching map onto the computational basis of the first system, which is well known to be a mixed unitary channel:
\begin{align}
    \Pi(\cdot)=\sum_{i=1}^d|i\rangle\langle i|\cdot|i\rangle\langle i|=\frac{1}{d}\sum_{j=1}^d U_j\cdot U_j^*, 
\end{align} where $U_j=\sum_k\omega^{jk}|k\rangle\langle k|$ with $\omega=\exp{2\pi i/d}$.
Next observe that the map 
\begin{align}
    (A\geq0\,,\,B\geq0)\mapsto  \|A^{-\frac{1}{2r}}XB^{-\frac{1}{2r}}\|^p_p= \Tr[(A^{-\frac{1}{2r}}XB^{-1}X^*A^{-\frac{1}{2r}})^\frac{p}{2}]
\end{align} is jointly convex due to \cite[Theorem 1.1. (2)]{Zhang.2020} with $0\leq \frac{1}{r}\leq 1$, and invariant under the replacement $(A,B)\mapsto (U_jAU^*_j,U_jBU^*_j)$ since for any $j\in[d]$, $(U_j)_1$ and $X=\oplus_{i}X_i$ commute. These two facts together imply that
\begin{align}\label{equ:decom.proof.1}
    \|(\Pi\otimes\id)(A)^{-\frac{1}{2r}}X(\Pi\otimes\id)(B)^{-\frac{1}{2r}}\|^p_p 
    \leq \|A^{-\frac{1}{2r}}XB^{-\frac{1}{2r}}\|^p_p.
\end{align} 

\noindent To finish the proof we use the above with $A=a_1,\, B=b_1$ to get 
\begin{align}
    \|\oplus_iX_i\|^p_{(q,p)} &= \inf_{{a,b\geq0}}\|a\|^p_{2r}\|b\|^p_{2r}\|a_1^{-1}(\oplus_iX_i)b_1^{-1}\|^p_{p} \\
    &= \inf_{{a,b\geq0, \|a\|_1,\|b\|_1\leq1}}\|a_1^{-\frac{1}{2r}}(\oplus_iX_i)b_1^{-\frac{1}{2r}}\|^p_{p} \\ 
    &\geq \inf_{\underset{X_i\neq0 \implies a_i,b_i>0}{a,b\geq0, \Tr[\Pi(a)],\Tr[\Pi(b)]\leq1}}\|(\Pi (a_1))^{-\frac{1}{2r}}(\oplus_iX_i)(\Pi( b_1))^{-\frac{1}{2r}}\|^p_{p} \\ &
    = \inf_{\underset{X_i\neq0 \implies a_i,b_i>0}{\{a_i,b_i\geq0\}, \sum_ia_i,\sum_ib_i\leq1}}\|\oplus_ia_i^{-\frac{1}{2r}}b_i^{-\frac{1}{2r}}X_i\|^p_{p}
    \\ &= \inf_{\underset{X_i\neq0 \implies a_i,b_i>0}{\{a_i,b_i\geq0\}, \sum_ia_i,\sum_ib_i\leq1}} \| (a^{-\frac{1}{2r}}_ib^{-\frac{1}{2r}}_i\|X_i\|_p)_i\|^p_{\ell_p} \\ 
    &= \inf_{\underset{X_i\neq0 \implies a_i,b_i>0}{(a_i)_i,(b_i)_i\in \ell^+_{2r}}} \|(a_i)_i\|^p_{\ell_{2r}}\|(b_i)_i\|_{\ell_{2r}}^p\|(a_i^{-1}b_i^{-1}\|X_i\|_p)\|^p_{\ell_p} \\
    &= \|(\|X_i\|_p)_i \|^p_{\ell_q} \\
    &= \left(\sum_i\|X_i\|_p^q\right)^\frac{p}{q}. 
\end{align}
where we used that the pinching is trace preserving, i.e. $\Tr[a]=\Tr[\Pi(a)]$ for $a\geq0$ and denoted $\Pi(a)=\oplus_ia_i,\, a_i\in\mathbb{C}$. The case $q\geq p$ is analogous: first, due to \cref{equ:Def1.2} we have
\begin{align}
    \left\|\oplus_i X_i\right\|_{(q,p)} &=\sup_{a,b\geq0}\|a\|^{-1}_{2r}\|b\|^{-1}_{2r}\|a_1(\oplus_iX_i)b_1\|_{p} \\ 
    &\geq \sup_{a=\oplus_ia_i,b=\oplus_ib_i\geq0}\|(a_i)_i\|^{-1}_{\ell_{2r}}\|(b_i)_i\|^{-1}_{\ell_{2r}}\|\oplus_ia_iX_ib_i\|_p \\ 
    &= \sup_{\{a_i,b_i\geq0\}_i}\|(a_i)_i\|^{-1}_{\ell_{2r}}\|(b_i)_i\|^{-1}_{\ell_{2r}}\| (a_ib_i\|X_i\|_p)_i\|_{\ell_p} \\ 
    &= \| (\|X_i\|_p)_i \|_{\ell_q} \\
    &= \left(\sum_i\|X_i\|_p^q\right)^\frac{1}{q}\,.
\end{align}
The reverse inequality follows also analogously to the previous case from joint concavity of the map 
\begin{align}
   (A,B) \mapsto \|A^\frac{1}{2r}XB^{\frac{1}{2r}}\|_p^p
\end{align} 
as per \cite[Theorem 1.1 (1)]{Zhang.2020}, since $0\leq \frac{1}{r}\leq 1$ and $p\leq r$, which is true since $\frac{1}{r}=\frac{1}{p}-\frac{1}{q}\leq \frac{1}{p}$. 

\medskip

 Next, we argue that in the regime $q\geq p$, if we took \cref{def:Definition1} without the $\frac{1}{r}\leq1$ condition, then \eqref{equ:cq.additivity} does not hold whenever $\frac{1}{r}>1$. Indeed, for $q>p>0$, with $\frac{1}{r}=\frac{1}{p}-\frac{1}{q}>1$, consider the following $4\times 4$ counterexample. Let $X:=X_0\oplus X_1$, with 
\begin{align}
    X_0:=|0\rangle\langle0|, \quad X_1:=|1\rangle\langle 1|,
\end{align} then we have $\|X_0\|_p=\|X_1\|_p=1$ and hence the RHS of \eqref{equ:cq.additivity} becomes $2^\frac{1}{q}$.
On the other hand recall that we take the LHS of \eqref{equ:cq.additivity} as
\begin{align}
    \|X\|_{(q,p)}&=\sup_{a,b}\|a\|^{-1}_{2r}\|b\|_{2r}^{-1}\|a_1(X_0\oplus X_1)b_1\|_p.
\end{align}
Now defining $a:=\begin{pmatrix}
    1 & 1 \\ 0&0
\end{pmatrix}= b^*$, then clearly $\|a\|_{2r}=\|b\|_{2r}=\sqrt{2}$ since $a,b$ have only one singular value $\sqrt{2}$ and we can bound
\begin{align}
    \|X\|_{(q,p)}&=\sup_{a,b}\|a\|^{-1}_{2r}\|b\|_{2r}^{-1}\|a_1(X_0\oplus X_1)b_1\|_p \\
    &\geq 2^{-1}\left\|\begin{pmatrix}
        1&1\\0&0
    \end{pmatrix}\begin{pmatrix}
        X_0&0\\0&X_1
    \end{pmatrix}\begin{pmatrix}
        1&0\\1&0
    \end{pmatrix}\right\|_p \\ &=2^{-1}\left\|\begin{pmatrix}
        X_0+X_1&0\\0&0
    \end{pmatrix}\right\|_p\\
    &= 2^{\frac{1}{p}-1},
\end{align} where the last equality comes from $X_0+X_1=\1_2$ with $\|\1_2\|_p=2^\frac{1}{p}$.
However, now due to $\frac{1}{p}-\frac{1}{q}>1$ it follows that $2^{\frac{1}{q}}<2^{\frac{1}{p}-1}$ which disproves \eqref{equ:cq.additivity} in this case.
\end{proof}

\subsection{Relational consistency}\label{sec:RelCon}

Next, we prove that the relations defining our quasi-norms can be inverted in the sense that ~the $\|\cdot\|_{(q,t)}$ quasi-norms can be factored through $\|\cdot\|_{(p,t)}$ quasi-norms in the most natural way.

\begin{theorem*}[Theorem \ref{thm:qt.to.pt}]
Let $q,p,t>0$ be s.t.~$\max\big\{\big|\frac{1}{t}-\frac{1}{q}\big|,\big|\frac{1}{t}-\frac{1}{p}\big|\big\}\leq 1$ and take \cref{def:Definition1} as a given, then it holds that 
\begin{align} 
    \|X\|_{(q,t)}&=\left\{
    \begin{aligned}
   & \inf_{\underset{X_{12}=a_1Z_{12}b_1}{a,b\geq 0, Z}}\|a\|_{2r}\|b\|_{2r}\|Z\|_{(p,t)}  & \text{ if } q\leq p \\
   &  \sup_{a,b\geq 0}\|a\|^{-1}_{2r}\|b\|^{-1}_{2r}\|a_1X_{12}b_1\|_{(p,t)} & \text{ if } q\geq p
   \end{aligned}\right.\qquad ,
\end{align}
where $\frac{1}{r}=\big|\frac{1}{q}-\frac{1}{p}\big|$ and for any $q$,
\begin{align}
    \|X\|_{(q,q)} = \|X\|_q.
\end{align}
Moreover, when $X\geq0$, then one can restrict the optimizations each to $a=b^*$.
\end{theorem*}

\begin{remark}
    Before be prove this theorem, which is one of our main structural results of this section we note that it implies that \cref{def:Definition1} and \cref{thm:qt.to.pt} are equivalent. This is significant since it shows that \cref{def:Definition1} is built on in some sense minimal simple assumptions, yet sufficient to prove the much more rich structure that is inbuilt to \cref{thm:qt.to.pt}. In particular it demonstrates that one can relate these quasi-norms with each other in the most natural way, as long as the indices of the involved quasi-norms are all compatible, i.e. not too far apart. This reproduces an important structural result which is known to be true for operator-valued Schatten norms \cite[Lemma 3.1]{Fawzi.2026} when $1\leq q,p,t$.
\end{remark}

\noindent To prove Theorem \ref{thm:qt.to.pt}, we first gather several useful implications of our \cref{def:Definition1}, which we will then suitably combine.

\begin{lemma}\label{lem:inverseDef1}
Given $0<q,p$ s.t. $\frac{1}{r}=\big|\frac{1}{q}-\frac{1}{p}\big|\leq 1$ then it holds that
\begin{align}
    \|X\|_{p} &= \|X\|_{(p,p)}, \label{equ:InvDef1.1} \\ 
 &= \sup_{{a,b\geq0}}\|a\|^{-1}_{2r}\|b\|^{-1}_{2r}\|a_1X_{12}b_1\|_{(q,p)}, \quad &&q\leq p \label{equ:InvDef1.2}\\
 &= \inf_{{a,b\geq0}}\|a\|_{2r}\|b\|_{2r}\|a_1^{-1}X_{12}b_1^{-1}\|_{(q,p)}, \quad &&q\geq p. \label{equ:InvDef1.3}
\end{align}
Moreover, when $X\geq0$, then one can restrict the optimizations each to $a=b^*$.
\end{lemma}

\begin{remark}
The factorization formulas in Lemma \ref{lem:inverseDef1} should be read as inverses of the ones with which we defined the 2-indexed quasi-norms \eqref{equ:Def1.1}, \eqref{equ:Def1.2}, implying that the regular Schatten-$p$-quasi-norms can be factorized through the $\|\cdot\|_{(q,p)}$-quasi-norms in the above way. 
\end{remark}

\begin{proof}[Proof of \cref{lem:inverseDef1}]
\cref{equ:InvDef1.1} follows directly from \eqref{equ:Def1.1} and \eqref{equ:Def1.2}: we observe that when $p=q$, then $\frac{1}{r}=\frac{1}{q}-\frac{1}{q}=0$ and so $r=\infty$. Then choosing the factorizations $a=b=\1, Z=X$ in \eqref{equ:Def1.1} yields $\|X\|_{(p,p)}\leq \|X\|_p$ and the factorization $a=b=\1$ in \eqref{equ:Def1.2} yields the converse $\|X\|_{(p,p)}\geq \|X\|_p$.

\medskip

Next we prove \eqref{equ:InvDef1.2}, i.e.~the case $q\leq p$. This is a little more involved and essentially follows from concatenating with \eqref{equ:Def1.1} to get a minimax expression, which under suitable rewriting allows for an application of Sion's minimax theorem, which we recall in \cref{thm:Sion} for completeness. The minimax expression naturally gives the lower bound implicit in \eqref{equ:InvDef1.2}, while the maximin gives the upper bound implicit in \eqref{equ:InvDef1.2}. Concretely, consider the right hand side of \eqref{equ:InvDef1.2}, then plugging in the expression derived in Lemma \ref{lem:SupportConditions}, we get 
\begin{align}
   \text{RHS of }{\eqref{equ:InvDef1.2}} &= \sup_{{a,b\geq0}}\|a\|^{-1}_{2r}\|b\|^{-1}_{2r}\|a_1X_{12}b_1\|_{(q,p)} \\ &= \sup_{{a,b\geq0}}\inf_{{c,d\ge 0}}\|a\|^{-1}_{2r}\|b\|^{-1}_{2r}\|c\|_{2r}\|d\|_{2r}\|c_1^{-1}a_1X_{12}b_1d^{-1}_1\|_{p} \\ &\leq \|X\|_p,
\end{align} where the upper bound comes from replacing the infimum over $c,d$ with $c=a, d=b$. 
Here implicitly we dropped writing the conditions on $a,b,c,d$ over which we optimize, namely that, by \cref{lem:SupportConditions} we can restrict $a,b\geq0$ 
s.t. $\Pi_a=L_X, \Pi_b=R_X$. Since $a_1X_{12}b_1$ again have left and right marginal support projections $L_{a_1X_{12}b_1}=L_X, R_{a_1X_{12}b_1}=R_X$, the restriction for $c,d\geq0$ is also $\Pi_c=L_X, \Pi_d=R_X$.  
Next, we prove that the infimum and supremum above can be exchanged, which then by an analogous argument  yields the converse inequality. 

In principle, in order to be able to apply Sion's minimax theorem, we would need to show that the objective function of the minimax is jointly concave in the operators $a,b$ and jointly convex in $c,d$. As is, the functional above, however, is provably not jointly concave in $a,b$. To get around this we introduce a suitable change of variables. Again, for notational simplicity we drop mentioning the support and image conditions explicitly. Then
\begin{align}
   \text{RHS of }{\eqref{equ:InvDef1.2}} &= \sup_{a,b\geq0}\inf_{\underset{}{c,d\geq0}}\|a\|^{-1}_{2r}\|b\|^{-1}_{2r}\|c\|_{2r}\|d\|_{2r}\|\underbrace{(c_1^{-1}a_1}_{C_1^{-1}}X_{12}\underbrace{b_1d^{-1}_1}_{D_1^{-1}}\|_{p} \\ 
   &= \sup_{\underset{\|a\|_1\leq1, \|b\|_1\leq1}{a,b\geq0}}\inf_{C,D}\|a^{\frac{1}{2r}}CC^*a^{\frac{1}{2r}}\|^\frac{1}{2}_r\|b^{\frac{1}{2r}}D^*Db^{\frac{1}{2r}}\|^\frac{1}{2}_r\|C_1^{-1}X_{12}D_1^{-1}\|_p \\ 
   &=\sup_{\underset{\|a\|_1\leq1, \|b\|_1\leq1}{a,b\geq0}}\inf_{\underset{X=(\Pi_{\tilde{C}})_1X(\Pi_{\tilde{D}})_1}{\tilde{C},\tilde{D}\geq0}}\|a^{\frac{1}{2r}}\tilde{C}a^{\frac{1}{2r}}\|^\frac{1}{2}_r\|b^{\frac{1}{2r}}\tilde{D}b^{\frac{1}{2r}}\|^\frac{1}{2}_r\|\tilde{C}^{-\frac{1}{2}}_1X_{12}\tilde{D}_1^{-\frac{1}{2}}\|_p \\
   &= \!\!\left(\!\sup_{\underset{\|a\|_1\leq1, \|b\|_1\leq1}{a,b\geq0}}\inf_{\underset{X=(\Pi_{\tilde{C}})_1X_{12}(\Pi_{\tilde{D}})_1}{\tilde{C},\tilde{D}\geq0}}\!\!\!\!\!\!\left\{ \frac{q}{2r}\|a^\frac{1}{2r}\tilde{C}a^\frac{1}{2r}\|_r^r+\frac{q}{2r}\|b^\frac{1}{2r}\tilde{D}b^\frac{1}{2r}\|_r^r+\frac{q}{p}\|\tilde{C}_1^{-\frac{1}{2}}X_{12}\tilde{D}_1^{-\frac{1}{2}}\|^p_p
   \right\}\!\right)^{\!\!\!\frac{1}{q}}
\end{align}
In the second equality we defined $C^{-1}:=c^{-1}a, D^{-1}:=bd^{-1} \Leftrightarrow c=aC, d=Db$ as $\Pi_c= \Pi_a=L_X, \Pi_d= \Pi_b=R_X$ and used that, $\|c\|_{2r}=\|cc^*\|^\frac{1}{2}_r$ and $\|d\|_{2r}=\|d^*d\|^\frac{1}{2}_r$. It holds for $\tilde{C}:=CC^*$ and $\tilde{D}:=D^*D$ that $X=(\Pi_{\tilde{C}})_1X(\Pi_{\tilde{D}})_1$ since $\Pi_{\tilde{C}}=\Pi_{c}= \Pi_a=L_X$ and likewise $\Pi_{\tilde{D}}=\Pi_{d}= \Pi_b=R_X$. We further redefined $a\mapsto a^\frac{1}{2r}$, $b\mapsto b^\frac{1}{2r}$, and used the scale invariance under rescaling of $a,b$.
The last equality follows from a general operator-valued optimization which we postpone to \cref{lem:Inf.Prod.to.Sum}, with $g_a(\tilde{C})=\|a^\frac{1}{2r}\tilde{C}a^\frac{1}{2r}\|_r^r$, $h_b(\tilde{D})=\|b^\frac{1}{2r}\tilde{D}b^\frac{1}{2r}\|_r^r$, and $f(\tilde{C},\tilde{D})=\|\tilde{C}_1^{-\frac{1}{2}}X_{12}\tilde{D}_1^{-\frac{1}{2}}\|^p_p$, and thus $r_1=r,\, r_2=\frac{p}{2}$ and $\frac{1}{r_0}=\frac{2}{r}+\frac{2}{p}=\frac{2}{q}$, i.e. $r_0=\frac{q}{2}$, which follows since $\frac{1}{r}=\frac{1}{q}-\frac{1}{p}$. 
To apply Sion's minimax \cref{thm:Sion} first note upper semi-continuity and concavity in $a,b$ due to \cref{lem:continuity} for fixed $\tilde{C},\tilde{D}$.
Next to establish the suitable joint convexity we see that 
\begin{align}
 \tilde{C}\mapsto   g_a(\tilde{C}) = \|a^\frac{1}{2r}\tilde{C}a^\frac{1}{2r}\|_r^r= \Tr[(a^\frac{1}{2r}\tilde{C}a^\frac{1}{2r})^r]= \Tr[(a^\frac{1}{2r}\tilde{C}a^\frac{1}{r}\tilde{C}a^\frac{1}{2r})^\frac{r}{2}]
\end{align} 
is convex and the same obviously also applies to $h_b(\tilde{D})$ in terms of $\tilde{D}$. 
Lastly the function
\begin{align}
   f(\tilde{C},\tilde{D})&=\|\tilde{C}_1^{-\frac{1}{2}}X_{12}\tilde{D}_1^{-\frac{1}{2}}\|^p_p = \Tr\Big[\big|\tilde{C}_1^{-\frac{1}{2}}X_{12}\tilde{D}_1^{-\frac{1}{2}}\big|^p\Big] \\&= \Tr\Big[\big(\tilde{C}_1^{-\frac{1}{2}}X_{12}\tilde{D}_1^{-1}X_{12}^*\tilde{C}_1^{-\frac{1}{2}}\big)^\frac{p}{2}\Big] 
\end{align} is jointly convex in $(\tilde{C},\tilde{D})$ by \cite[Theorem 1.1 (2)]{Zhang.2020} since $p>0$. 
Since $\frac{q}{2r},\frac{q}{p}\geq0$ these individual concavity/convexity results  imply that  
\begin{align}
  (a,b,\tilde{C},\tilde{D})\mapsto  \frac{q}{2r}\|a^\frac{1}{2r}\tilde{C}a^\frac{1}{2r}\|_r^r+\frac{q}{2r}\|b^\frac{1}{2r}\tilde{D}b^\frac{1}{2r}\|_r^r+\frac{q}{p}\|\tilde{C}_1^{-\frac{1}{2}}X_{12}\tilde{D}_1^{-\frac{1}{2}}\|^p_p
\end{align}
 is jointly convex in $(\tilde{C},\tilde{D})$, jointly concave in $(a,b)$, and continuous in $(a,b)$ on positive trace sub normalized operators, for every fixed $\tilde{C},\tilde{D}$. The cone of positive semidefinite operators is clearly convex and adding the linear support/image constraint does not change this fact. The set of sub normalized trace-class operators is convex and norm-compact in finite dimensions. 
Hence Sion's minimax theorem can be applied and after unraveling the changes of variables yields:
\begin{align}
    \text{RHS of }{\eqref{equ:InvDef1.2}} & = \!\!\left(\!\inf_{\underset{X=(\Pi_{\tilde{C}})_1X_{12}(\Pi_{\tilde{D}})_1}{\tilde{C},\tilde{D}\geq0}}\sup_{\underset{\|a\|_1\leq1, \|b\|_1\leq1}{a,b\geq0}}\!\!\!\!\!\!\left\{ \frac{q}{2r}\|a^\frac{1}{2r}\tilde{C}a^\frac{1}{2r}\|_r^r+\frac{q}{2r}\|b^\frac{1}{2r}\tilde{D}b^\frac{1}{2r}\|_r^r+\frac{q}{p}\|\tilde{C}_1^{-\frac{1}{2}}X_{12}\tilde{D}_1^{-\frac{1}{2}}\|^p_p
   \right\}\!\right)^{\!\!\!\frac{1}{q}}   \\
   &=\left(\inf_{\underset{X=(\Pi_{\tilde{C}})_1X_{12}(\Pi_{\tilde{D}})_1}{\tilde{C},\tilde{D}\geq0}}\left\{ \frac{q}{2r}\|\tilde{C}\|_\infty^r+\frac{q}{2r}\|\tilde{D}\|_\infty^r+\frac{q}{p}\|\tilde{C}_1^{-\frac{1}{2}}X_{12}\tilde{D}_1^{-\frac{1}{2}}\|^p_p
   \right\}\right)^\frac{1}{q} \\
   &=\inf_{\underset{X=(\Pi_{\tilde{C}})_1X_{12}(\Pi_{\tilde{D}})_1}{\tilde{C},\tilde{D}\geq0}}\|\tilde{C}\|^\frac{1}{2}_\infty\|\tilde{D}\|^\frac{1}{2}_\infty\|\tilde{C}_1^{-\frac{1}{2}}X_{12}\tilde{D}_1^{-\frac{1}{2}}\|_p
   \\
   &=\inf_{\underset{X=(\Pi_{\tilde{C}})_1X_{12}(\Pi_{\tilde{D}})_1}{\tilde{C},\tilde{D}\geq0}}\|\tilde{C}\|_\infty\|\tilde{D}\|_\infty\|\tilde{C}_1^{-1}X_{12}\tilde{D}_1^{-1}\|_p\\
   &= \|X\|_{p}.
\end{align}
where the third equality follows again from \cref{lem:Inf.Prod.to.Sum}, and the last follows from  \cref{def:Definition1} \eqref{equ:Def1.1} with \cref{equ:InvDef1.1}.

\noindent In the case where $X\geq0$, then the claimed statement involves only $a=b^*\geq0$. Thus $a_1X_{12}a_1\geq0$ and by \cref{lem:Positive.Symmetrie} we can then choose $c=d^*\geq0$, which implies in particular that $C=c^{-1}a=D^*=bd^{-1}$. Hence $\tilde C=CC^*=D^*D=\tilde D$ which with the above proof yields the desired symmetric statement.

\medskip

Completely analogously to above, we prove \eqref{equ:InvDef1.3} in the case $q\geq p$. Without loss of generality we may assume that  $1\geq p$, as else this is the well known \cite[(3.18)]{Devetak.2006}. We continue by concatenating with \eqref{equ:Def1.2} to get a maximin expression, which under analogous suitable rewriting allows for an application of Sion's minimax \cref{thm:Sion}. The maximin expression naturally gives the upper bound implicit in \eqref{equ:InvDef1.3}, while the minimax gives the lower bound implicit in \eqref{equ:InvDef1.3}. Again for notational simplicity we drop writing the marginal support conditions from \cref{lem:SupportConditions}.
Due to \eqref{equ:Def1.2} we get
\begin{align}
    \text{RHS of }{\eqref{equ:InvDef1.3}} &= \inf_{\underset{}{a,b\geq0}}\|a\|_{2r}\|b\|_{2r}\|a_1^{-1}Xb_1^{-1}\|_{(q,p)} \\
    &= \inf_{a,b\geq0}\sup_{c,d\geq0}\|a\|_{2r}\|b\|_{2r}\|c\|_{2r}^{-1}\|d\|_{2r}^{-1}  \|(ca^{-1}\otimes\1)X(b^{-1}d\otimes\1)\|_{p}\\ &\geq \|X\|_p
\end{align} where the lower bound comes from replacing the supremum over $c,d$ with $c=a, d=b$. 
We will prove below, that the infimum and supremum can be exchanged, which then by an analogous argument to above yields the converse inequality. 

The argument that the infimum and supremum can be exchanged follows analogously, but slightly differently as before. Dropping again to write the support conditions from \cref{lem:SupportConditions} we have
\begin{align}
   \text{RHS of }{\eqref{equ:InvDef1.3}} &= \inf_{a,b\geq0}\sup_{c,d\geq0}\|a\|_{2r}\|b\|_{2r}\|c\|_{2r}^{-1}\|d\|_{2r}^{-1}  \|(ca^{-1}\otimes\1)X(b^{-1}d\otimes\1)\|_{p} \\ 
   &=\inf_{a,b\geq0}\sup_{c,d\geq0} \|a^\frac{1}{2r}\tilde{C}a^\frac{1}{2r}\|_r^{-\frac{1}{2}}\|b^\frac{1}{2r}\tilde{D}b^\frac{1}{2r}\|_r^{-\frac{1}{2}}\|(\tilde{C}^\frac{1}{2}\otimes\1)X(\tilde{D}^\frac{1}{2}\otimes\1)\|_p \\ 
   &= \left(\inf_{a,b\geq0}\sup_{\tilde{C},\tilde{D}\geq0}\left\{ -\frac{q}{2r}\|a^\frac{1}{2r}\tilde{C}a^\frac{1}{2r}\|_r^r-\frac{q}{2r}\|b^\frac{1}{2r}\tilde{D}b^\frac{1}{2r}\|_r^r+\frac{q}{p}\|(\tilde{C}^\frac{1}{2}\otimes\1)X(\tilde{D}^\frac{1}{2}\otimes\1)\|_p^p \right\}\right)^\frac{1}{q},
\end{align} where the last equality here follows from \cref{lem:Sup.Prod.to.Sum} with $g_a(\tilde{C})=\|a^\frac{1}{2r}\tilde{C}a^\frac{1}{2r}\|_r^r$, $h_b=\|b^\frac{1}{2r}\tilde{D}b^\frac{1}{2r}\|_r^r$, $f(\tilde{C},\tilde{D})=\|(\tilde{C}^\frac{1}{2}\otimes\1)X(\tilde{D}^\frac{1}{2}\otimes\1)\|_p^p$ and $r_1=r>1 ,\, r_2=-\frac{p}{2}<0$, which implies $\frac{1}{r_0}=\frac{2}{r}-\frac{2}{p}=-\frac{2}{q}<0$, as $\frac{1}{r}=\frac{1}{p}-\frac{1}{q}$. In the second line we defined $C:=ca^{-1}, D:=b^{-1}d \Leftrightarrow c=Ca, d=bD$ since $\Pi_a=\Pi_c=L_X, \Pi_b=\Pi_d=R_X$, by \cref{lem:SupportConditions} and set $\tilde{C}:=C^*C, \tilde{D}:=DD^*$.

To establish the required joint convexity observe that exactly as before $g_a(\tilde{C})$ is convex in $\tilde{C}$ and by \cref{lem:continuity} concave and continuous in $a$, which means that $-\frac{q}{2r}g_a(\tilde{C})$ is concave in $\tilde{C}$ and convex and continuous in $a$. The same applies to $h_b(\tilde{D})$. 
For
\begin{align}
f(\tilde{C},\tilde{D})& =\|(\tilde{C}^\frac{1}{2}\otimes\1)X(\tilde{D}^\frac{1}{2}\otimes\1)\|_p^p= \Tr[(\tilde{C}^\frac{1}{2}\otimes\1)X(\tilde{D}^\frac{1}{2}\otimes\1))^p] \\ &= \Tr[(\tilde{C}^\frac{1}{2}\otimes\1)X(\tilde{D}^1\otimes\1)X^*(\tilde{C}^\frac{1}{2}\otimes\1))^\frac{p}{2}]  
\end{align} is jointly concave in $\tilde{C},\tilde{D}$ by \cite[Theorem 1.1 (1)]{Zhang.2020} as $0<p\leq1$.
Hence as a sum it follows that
\begin{align}
    -\frac{q}{2r}\|a^\frac{1}{2r}\tilde{C}a^\frac{1}{2r}\|_r^r-\frac{q}{2r}\|b^\frac{1}{2r}\tilde{D}b^\frac{1}{2r}\|_r^r+\frac{q}{p}\|(\tilde{C}^\frac{1}{2}\otimes\1)X(\tilde{D}^\frac{1}{2}\otimes\1)\|_p^p
\end{align} is jointly convex in $a,b$ and jointly concave in $\tilde{C},\tilde{D}$. With the exact same discussion as above it follows that Sion's minimax theorem can be applied and the desired statement follows which concludes the proof of \cref{lem:inverseDef1} via
\begin{align}
   \text{RHS of }{\eqref{equ:InvDef1.3}} &= 
   \left(\sup_{\tilde{C},\tilde{D}\geq0}\inf_{a,b\geq0}\left\{ -\frac{q}{2r}\|a^\frac{1}{2r}\tilde{C}a^\frac{1}{2r}\|_r^r-\frac{q}{2r}\|b^\frac{1}{2r}\tilde{D}b^\frac{1}{2r}\|_r^r+\frac{q}{p}\|(\tilde{C}^\frac{1}{2}\otimes\1)X(\tilde{D}^\frac{1}{2}\otimes\1)\|_p^p \right\}\right)^\frac{1}{q} \\
   &= \left(\sup_{\tilde{C},\tilde{D}\geq0}\left\{ -\frac{q}{2r}\|\tilde{C}\|_\infty^r-\frac{q}{2r}\|\tilde{D}\|_\infty^r+\frac{q}{p}\|(\tilde{C}^\frac{1}{2}\otimes\1)X(\tilde{D}^\frac{1}{2}\otimes\1)\|_p^p \right\}\right)^\frac{1}{q} \\
   &= \sup_{\tilde{C},\tilde{D}\geq0}\|\tilde{C}\|^{-\frac{1}{2}}_\infty\|\tilde{D}\|^{-\frac{1}{2}}_\infty\|(\tilde{C}^\frac{1}{2}\otimes\1)X(\tilde{D}^\frac{1}{2}\otimes\1)\|_p \\
   &= \sup_{\tilde{C},\tilde{D}\geq0}\|\tilde{C}\|^{-1}_\infty\|\tilde{D}\|^{-1}_\infty\|(\tilde{C}\otimes\1)X(\tilde{D}\otimes\1)\|_p = \|X\|_p 
\end{align} where the last equality is just \cref{def:Definition1} \eqref{equ:Def1.2}.

\noindent In the case where $X\geq0$, then the claimed statement involves only $a=b^*\geq0$. Thus again $a_1X_{12}a_1\geq0$ and by \cref{lem:Positive.Symmetrie} we can then choose $c=d^*\geq0$, which implies in particular that $C=c^{-1}a=D^*=bd^{-1}$. Hence $\tilde C=CC^*=D^*D=\tilde D$ which with the above proof yields the desired symmetric statement.
\end{proof}

\noindent We are now in a position to prove \cref{thm:qt.to.pt} and in particular \eqref{equ:thm:qt.to.pt.inf} and \eqref{equ:thm:qt.to.pt.sup} starting from \cref{def:Definition1} and \cref{lem:inverseDef1}. 
\begin{proof}[Proof of \cref{thm:qt.to.pt}]
In principle we need to consider 6 different orders of the indices $q,p,t$. We will consider them one by one and the proof strategy for most of them will be analogous, suitably concatenating the factorization formulas from  \cref{def:Definition1} and \cref{lem:inverseDef1} to relate the desired $\|\cdot\|_{(q,t)}$ and $\|\cdot\|_{(p,t)}$, redefining the optimization variables, and in two cases swapping an occurring infimum and supremum. These two cases will be more intricate and first require proving joint concavity/ convexity of 2-indexed quasi-norm functionals. In the following, we denote by $m,s\geq 1$ the differences
\begin{align}
    \frac{1}{m}&:=\left|\frac{1}{t}-\frac{1}{q}\right|, \qquad
    \frac{1}{s}:=\left|\frac{1}{t}-\frac{1}{p}\right|, \qquad
    \frac{1}{r}:=\left|\frac{1}{q}-\frac{1}{p}\right|. 
\end{align}
For notational simplicity we will not be writing the support conditions of \cref{lem:SupportConditions} explicitly and only commenting on them when necessary. 

\medskip 

\noindent \underline{Case 1: $q\leq t\leq p$:}  From \Cref{lem:SupportConditions} and \cref{equ:InvDef1.3} it follows that
\begin{align}
    \|X\|_{(q,t)} &= \inf_{\underset{}{a,b\geq0}}\|a\|_{2m}\|b\|_{2m}\|a_1^{-1}X_{12}b_1^{-1}\|_{t} \\
    &= \inf_{\underset{}{a,b\geq0}}\inf_{\underset{}{c,d\geq0}}\|a\|_{2m}\|b\|_{2m}\|c\|_{2s}\|d\|_{2s}\|\underbrace{c_1^{-1}a^{-1}_1}_{A_1^{-1}=U_1^*\tilde{A}^{-\frac{1}{2}}_1}X_{12}\underbrace{b_1^{-1}d^{-1}_1}_{B_1^{-1}=\tilde{B}_1^{-\frac{1}{2}}V^*_1}\|_{(p,t)} \\
    &\overset{(1)}{=} \inf_{{\tilde{A},\tilde{B}\geq0}}\inf_{a,b\geq0} \|a\|_{2m}\|a^{-1}\tilde{A}a^{-1}\|^\frac{1}{2}_{s}\|b\|_{2m}\|b^{-1}\tilde{B}b^{-1}\|_s^\frac{1}{2} \|\tilde{A}_1^{-\frac{1}{2}}X_{12}\tilde{B}_1^{-\frac{1}{2}}\|_{(p,t)} \\
    &\overset{(2)}{=}\inf_{\underset{}{\tilde{A},\tilde{B}\geq0}}  \|\tilde{A}\|^\frac{1}{2}_{r}\|\tilde{B}\|^{\frac{1}{2}}_{r}\|\tilde{A}_1^{-\frac{1}{2}}X_{12}\tilde{B}_1^{-\frac{1}{2}}\|_{(p,t)}\\ 
    &=\inf_{\underset{}{\tilde{A},\tilde{B}\geq0}}  \|\tilde{A}\|_{2r}\|\tilde{B}\|_{2r}\|\tilde{A}_1^{-1}X_{12}\tilde{B}_1^{-1}\|_{(p,t)},
\end{align} 
where (1) follows from local unitary invariance of the $(p,t)$ quasi-norm along with defining $A:=ac$ and $\tilde{A}:=AA^*=acca$, implying $cc=a^{-1}AA^*a^{-1}$ since by \cref{lem:SupportConditions} $\Pi_a=\Pi_c=L_X$ and the analogous holds for $\Pi_c=\Pi_d=R_X$. It also implies that $\Pi_{\tilde{A}}= \Pi_c= \Pi_a=L_X$. The polar decomposition  of the generalized inverse of $A$ is given by $A^{-1}=U^* \tilde{A}^{-\frac{1}{2}}$. Everything is analogous for $B:=db$ with $\tilde{B}:=B^*B=dbbd$, it follows that $dd=b^{-1}BB^*b^{-1}$ and the polar decomposition of the inverse is $B^{-1}=\tilde{B}^{-\frac{1}{2}}V^*$ and that $\Pi_{\tilde{B}}= \Pi_d= \Pi_b=R_X$.
The equality in (2) follows from \cref{prop:SpVariationalFormulas} and the fact that $r\leq s$ with $\frac{1}{m}=\frac{1}{r}-\frac{1}{s}$. 
In the last line we redefined $\tilde{A}\rightarrow \tilde{A}^2, \tilde{B}\rightarrow \tilde{B}^2$.

\medskip

\noindent
\underline{Case 2: $q\leq p\leq t$:} Here $\frac{1}{m}=\frac{1}{r}+\frac{1}{s}$ and both $\|\cdot\|_{(q,t)}$ and $\|\cdot\|_{(p,t)}$ are given by \eqref{equ:Def1.1}. This case goes analogously to the proof of \cite[Theorem 4.5]{Beigi.2023} or \cite[Lemma 3.1]{Fawzi.2026} with no exchange of infimum and supremum required. We show the two inequalities directly, starting with

``$\leq$'': Let $X_{12}=a_1Z_{12}b_1$ with $a,b\geq0$, and let $Z_{12}=c_1W_{12}d_1$ be any factorization as in \eqref{equ:Def1.1} for $\|Z\|_{(p,t)}$. Then $X_{12}=(ac)_1W_{12}(db)_1$, so by \eqref{equ:Def1.1} for $\|X\|_{(q,t)}$ and the generalized Hölder inequality with $\frac{1}{2m}=\frac{1}{2r}+\frac{1}{2s}$,
\begin{align}
    \|X\|_{(q,t)}\leq \|ac\|_{2m}\|db\|_{2m}\|W\|_t\leq \|a\|_{2r}\|b\|_{2r}\,\|c\|_{2s}\|d\|_{2s}\|W\|_t .
\end{align}
Taking the infimum over $c,d,W$ gives $\|X\|_{(q,t)}\leq\|a\|_{2r}\|b\|_{2r}\|Z\|_{(p,t)}$ via \eqref{equ:Def1.1}. Taking then the infimum over $a,b,Z$ yields the $\leq$-part of what was to show.

``$\geq$'': Let $X_{12}=a_1W_{12}b_1$ with $a,b\geq0$. Since $\frac{m}{r}+\frac{m}{s}=1$ we may split $a=A c$ and $b=dB$ with the commuting positive operators
\begin{align}
    A:=a^{\frac{m}{r}},\quad c:=a^{\frac{m}{s}},\qquad B:=b^{\frac{m}{r}},\quad d:=b^{\frac{m}{s}},
\end{align}
which satisfy $\|A\|_{2r}\|c\|_{2s}=\|a\|_{2m}^{\frac{m}{r}}\|a\|_{2m}^{\frac{m}{s}}=\|a\|_{2m}$ and likewise $\|B\|_{2r}\|d\|_{2s}=\|b\|_{2m}$ (for $s=\infty$, i.e. $p=t$, set $c=\Pi_a$, $d=\Pi_b$). Then $X_{12}=A_1Z_{12}B_1$ with $Z:=c_1W_{12}d_1$, and \eqref{equ:Def1.1} for $\|Z\|_{(p,t)}$ gives
\begin{align}
    \|A\|_{2r}\|B\|_{2r}\|Z\|_{(p,t)}\leq \|A\|_{2r}\|B\|_{2r}\|c\|_{2s}\|d\|_{2s}\|W\|_t=\|a\|_{2m}\|b\|_{2m}\|W\|_t .
\end{align}
Taking the infimum over $a,b,W$ yields the claim. If $X\geq0$, we may take $a=b$ by \cref{lem:Positive.Symmetrie}, and then $A=B$, which gives the symmetric version.

\medskip

\noindent
\underline{Case 3: $p\le q\leq t$:} This case is more challenging since it involves a switching of the infimum and supremum.
WLOG we may assume $p<1$, since in the case where $p\geq 1$ it follows that $q,t\geq1$ are as well and we are in the normed setting in which the claim is known to be true \cite[Lemma 3.1]{Fawzi.2026}. Further if $p= t$, i.e. $s=\infty$ we are done by \cref{lem:inverseDef1}. Hence in the following consider $1>p\neq t$. Before we deal with the claimed equality we first establish the following claim: The map
\begin{align}
    (C,D)\mapsto \|(C^{\frac{1}{2}}\otimes\1)X(D^{\frac{1}{2}}\otimes\1)\|^{\frac{1}{2}}_{(p,t)}
\end{align} is jointly concave in $C,D\geq0$ for $1>p\leq t$.\footnote{This range of $p$ is sufficient for this argument, but not exhaustive.} 
To do this we resort to the result of Case 1, which was independently proven. Indeed,
Since $p<1$ define $P:=\frac{p}{1-p}\geq t\geq p$, then $\frac{1}{p}-\frac{1}{P}=1$ and $\frac{1}{t}-\frac{1}{P}\leq 1$. By Case 1
\begin{align}
    &\|(C^{\frac{1}{2}}\otimes\1)X(D^{\frac{1}{2}}\otimes\1)\|^{\frac{1}{2}}_{(p,t)}\\
    &\qquad = \inf_{a,b\geq0}\|a\|_2^\frac{1}{2}\|b\|_2^\frac{1}{2}\|(a^{-1}C^{\frac{1}{2}}\otimes\1)X(D^{\frac{1}{2}}b^{-1}\otimes\1)\|^{\frac{1}{2}}_{(P,t)} \\
    &\qquad = \inf_{A,B}\|ACA^*\|_1^\frac{1}{4}\|B^*DB\|_1^\frac{1}{4}\|(A^{-1}\otimes\1)X(B^{-1}\otimes\1)\|_{(P,t)}^\frac{1}{2} \\
    &\qquad =\inf_{\tilde{A},\tilde{B}\geq0}\|\tilde{A}^\frac{1}{2}C\tilde{A}^\frac{1}{2}\|_1^\frac{1}{4}\|\tilde{B}^\frac{1}{2}D\tilde{B}^\frac{1}{2}\|_1^\frac{1}{4}\|(\tilde{A}^{-\frac{1}{2}}\otimes\1)X(\tilde{B}^{-\frac{1}{2}}\otimes\1)\|_{(P,t)}^\frac{1}{2} \\
    &\qquad = \inf_{\tilde{A},\tilde{B}\geq0}\left\{\frac{1}{4}\|\tilde{A}^\frac{1}{2}C\tilde{A}^\frac{1}{2}\|_1+\frac{1}{4}\|\tilde{B}^\frac{1}{2}D\tilde{B}^\frac{1}{2}\|_1+\frac{1}{2}\|(\tilde{A}^{-\frac{1}{2}}\otimes\1)X(\tilde{B}^{-\frac{1}{2}}\otimes\1)\|_{(P,t)}\right\},
\end{align} where the last equality is \cref{lem:Inf.Prod.to.Sum} with $r_1=1,\, r_2=\frac{1}{2}$. Now since $0\leq C\mapsto \|\tilde{A}^\frac{1}{2}C\tilde{A}^\frac{1}{2}\|_1$ is linear it is in particular concave, hence the last sum is jointly concave in $C,D$ and this joint concavity is preserved by the infimum.

Having this joint concavity we may now proceed with the proof by applying \cref{equ:Def1.2} and \cref{equ:Def1.1}
\begin{align}
    &\|X\|_{(q,t)} \\
    &\qquad = \inf_{\underset{}{a,b\geq0}}\|a\|_{2m}\|b\|_{2m}\|a_1^{-1}Xb_1^{-1}\|_{t} \\
    &\qquad = \inf_{\underset{}{a,b\geq0}}\sup_{\underset{}{c,d\geq0}}\|a\|_{2m}\|b\|_{2m}\|c\|^{-1}_{2s}\|d\|^{-1}_{2s}\|(\underbrace{ca^{-1}}_{C}\otimes\1)X(\underbrace{b^{-1}d}_{D}\otimes\1)\|_{(p,t)} \\ 
    &\qquad =\inf_{\underset{}{a,b\geq0}}\sup_{\underset{}{C,D\geq0}}\|a\|_{2m}\|b\|_{2m}\|aC^*Ca\|^{-\frac{1}{2}}_{s}\|bDD^*b\|^{-\frac{1}{2}}_{s}\|(C\otimes\1)X(D\otimes\1)\|_{(p,t)} \\
    &\qquad =\left\{\inf_{\underset{}{a,b\geq0}}\frac{1}{2s-1}\sup_{\underset{}{C,D\geq0}}\!-\frac{1}{2}\|a^\frac{1}{2m}\tilde{C}a^{\frac{1}{2m}}\|^{s}_{s}-\frac{1}{2} \|b^\frac{1}{2m}\tilde{D}b^{\frac{1}{2m}}\|^{s}_{s}+2s\|(\tilde{C}^\frac{1}{2}\otimes\1)X(\tilde{D}^\frac{1}{2}\otimes\1)\|^\frac{1}{2}_{(p,t)}\right\}^{\frac{2s-1}{s}}
\end{align} 
where the last equality comes from \cref{lem:Sup.Prod.to.Sum} with $r_1=s,\, r_2=-\frac{1}{4}$. As in all previous proofs, the change of variables calculations are justified by \cref{lem:SupportConditions}. By the before proved claim we have that the last summand is jointly concave in $\tilde{C},\tilde{D}$ and by \cite[Lemma B.1]{Rubboli.2025} it follows that
\begin{align}
  -\frac{1}{2}\|a^\frac{1}{2m}\tilde{C}a^{\frac{1}{2m}}\|^{s}_{s} = -\frac{1}{2}\|\tilde{C}^\frac{1}{2}a^\frac{1}{m}\tilde{C}^\frac{1}{2}\|^{s}_{s}
\end{align} is concave in $\tilde{C}$ as $1\leq s$ and convex in $a$ as $0\leq \frac{1}{m}\leq 1$ and $s\leq m$, which follows from $q\geq p$. Along with \cref{lem:continuity} this implies that we can use Sion's minimax theorem to get
\begin{align}
   \|X\|_{(q,t)} &= \sup_{\underset{}{C,D\geq0}}\inf_{\underset{}{a,b\geq0}}\|a\|_{2m}\|b\|_{2m}\|aC^*Ca\|^{-\frac{1}{2}}_{s}\|bDD^*b\|^{-\frac{1}{2}}_{s}\|(C\otimes\1)X(D\otimes\1)\|_{(p,t)} \\
   &= \sup_{\underset{}{C,D\geq0}}\|C^*C\|^{-\frac{1}{2}}_{r}\|DD^*\|^{-\frac{1}{2}}_{r}\|(C\otimes\1)X(D\otimes\1)\|_{(p,t)}
   \\ &= \sup_{\underset{}{C,D\geq0}}\|C\|^{-1}_{2r}\|D\|^{-1}_{2r}\|(C\otimes\1)X(D\otimes\1)\|_{(p,t)}, 
\end{align} where the second equality is \cref{prop:SpVariationalFormulas} 
with $\frac{1}{m}=\frac{1}{s}-\frac{1}{r}\geq0$, via
\begin{align}
    \inf_{a\geq0}\|a\|_{2m}\|aC^*Ca\|^{-\frac{1}{2}}_{s} &= (\sup_{a\geq0}\|a\|_{2m}^{-2}\|aC^*Ca\|_s)^{-\frac{1}{2}} = \|C^*C\|_r^{-\frac{1}{2}}
\end{align} implicitly using \cref{lem:SupportConditions}.

\medskip

\noindent
\underline{Case 4: $t\leq p\leq q$:} This case goes exactly analogous to case 2. Here $\frac{1}{m}=\frac{1}{r}+\frac{1}{s}$ and both $\|\cdot\|_{(q,t)}$ and $\|\cdot\|_{(p,t)}$ are given by \eqref{equ:Def1.2}. Again we show the two inequalities directly starting with

``$\geq$'': Let $a,b\geq0$ and $c,d\geq0$. By \eqref{equ:Def1.2} for $\|X\|_{(q,t)}$, applied to the operators $ca$ and $bd$, and the generalized Hölder inequality $\|ca\|_{2m}\leq\|c\|_{2s}\|a\|_{2r}$, $\|bd\|_{2m}\leq\|b\|_{2r}\|d\|_{2s}$,
\begin{align}
    \|X\|_{(q,t)}\geq \|ca\|_{2m}^{-1}\|bd\|_{2m}^{-1}\|(ca)_1X_{12}(bd)_1\|_t\geq \|a\|^{-1}_{2r}\|b\|^{-1}_{2r}\,\|c\|_{2s}^{-1}\|d\|_{2s}^{-1}\|c_1(a_1X_{12}b_1)d_1\|_t .
\end{align}
Taking the supremum over $c,d$ and using \eqref{equ:Def1.2} for $\|a_1X_{12}b_1\|_{(p,t)}$ gives
\begin{align}
    \|X\|_{(q,t)}\geq\|a\|^{-1}_{2r}\|b\|^{-1}_{2r}\|a_1X_{12}b_1\|_{(p,t)},
\end{align}
and it remains to take the supremum over $a,b$.

``$\leq$'': Let $a,b\geq0$ and split $a=cA$, $b=Bd$ with $A:=a^{\frac{m}{r}}$, $c:=a^{\frac{m}{s}}$, $B:=b^{\frac{m}{r}}$, $d:=b^{\frac{m}{s}}$ as in Case 2, so that $\|a\|_{2m}=\|A\|_{2r}\|c\|_{2s}$ and $\|b\|_{2m}=\|B\|_{2r}\|d\|_{2s}$. By \eqref{equ:Def1.2} for $\|A_1X_{12}B_1\|_{(p,t)}$ we have
\begin{align}
    \|a\|^{-1}_{2m}\|b\|^{-1}_{2m}\|a_1X_{12}b_1\|_t&=\|A\|^{-1}_{2r}\|B\|^{-1}_{2r}\,\|c\|^{-1}_{2s}\|d\|^{-1}_{2s}\|c_1(A_1X_{12}B_1)d_1\|_t\\
    &\leq \|A\|^{-1}_{2r}\|B\|^{-1}_{2r}\|A_1X_{12}B_1\|_{(p,t)} .
\end{align}
Taking the supremum over $a,b$ yields the claim. If $X\geq0$, we may take $a=b$ by \cref{lem:Positive.Symmetrie}, and then $A=B$, which gives the symmetric version.

\medskip

\noindent
\underline{Case 5: $p\leq t\leq q$:} This case is analogous to Case 1 and involves two suprema which need not even be switched:
\begin{align}
    \|X\|_{(q,t)}&=\sup_{a,b\geq0}\|a\|_{2m}^{-1}\|b\|_{2m}^{-1}\|(a\otimes\1)X(b\otimes\1)\|_{t} \\
    &= \sup_{a,b\geq0}\sup_{c,d\geq0}\|a\|_{2m}^{-1}\|b\|_{2m}^{-1}\|c\|_{2s}^{-1}\|d\|_{2s}^{-1}\|(\underbrace{ca}_A\otimes\1)X(\underbrace{bd}_B\otimes\1)\|_{(p,t)} \\
    &= \sup_{A,B}\sup_{c,d\geq0}\|c^{-1}AA^*c^{-1}\|_{m}^{-\frac{1}{2}}\|d^{-1}B^*Bd^{-1}\|_{m}^{-\frac{1}{2}}\|c\|_{2s}^{-1}\|d\|_{2s}^{-1}\|(A\otimes\1)X(B\otimes\1)\|_{(p,t)} \\ &= \sup_{A,B}\|AA^*\|_{r}^{-\frac{1}{2}}\|B^*B\|_{r}^{-\frac{1}{2}}\|(A\otimes\1)X(B\otimes\1)\|_{(p,t)} \\ 
    &=\sup_{A,B\geq0}\|A\|_{2r}^{-1}\|B\|_{2r}^{-1}\|(A\otimes\1)X(B\otimes\1)\|_{(p,t)}
\end{align} where the second to last equality follows from 
\begin{align}
    \sup_{c\geq0}\|c\|_{2s}^{-1}\|c^{-1}AA^*c^{-1}\|_{m}^{-\frac{1}{2}} = (\inf_{\underset{}{c\geq0}}\|c\|_{2s}^2\|c^{-1}AA^*c^{-1}\|_{m})^{-\frac{1}{2}} = \|AA^*\|_r^{-\frac{1}{2}},
\end{align} by \cref{prop:SpVariationalFormulas} with $\frac{1}{r}-\frac{1}{m}=\frac{1}{s}\geq0$. Here we implicitly used \cref{lem:SupportConditions}. In the last equality we used the polar decompositions $A=U|A|, B=|B^*|V$, local unitary invariance of the $\|\cdot\|_{(p,t)}$-quasi-norm, and the fact that $\|AA^*\|_r=\|A^*A\|_r=\|A\|^2_{2r}$, analogous for $B$, to reduce to $A,B\geq0$.

\underline{Case 6: $t\leq q\leq p$:} This case is somewhat analogous to Case 3 and requires a swapping of the appearing infimum and supremum. In order to deal with this we first establish the following claim: the map
\begin{align}
    (C,D)\mapsto \|(C^{-\frac{1}{2}}\otimes\1)X(D^{-1/2}\otimes\1)\|_{(p,t)}
\end{align} is jointly convex on $C,D\geq0$ for $p\geq t$. To do this we resort to the result of Case 5: let $P:=\frac{p}{1+p}\leq t\leq p$, then $\frac{1}{P}-\frac{1}{p}=1$ and $|\frac{1}{P}-\frac{1}{t}|=1+\frac{1}{p}-\frac{1}{t}\leq 1$.
Due to Case 5 we know that for this $P$ it holds that
\begin{align}
    &\|(C^{-\frac{1}{2}}\otimes\1)X(D^{-1/2}\otimes\1)\|_{(p,t)}\\
    &\qquad= \sup_{a,b\geq0}\|a\|_2^{-1}\|b\|_2^{-1}\|(\underbrace{aC^{-\frac{1}{2}}}_{A}\otimes\1)X(\underbrace{D^{-1/2}b}_{B}\otimes\1)\|_{(P,t)} \\
    &\qquad=\sup_{A,B}\|ACA^*\|^{-\frac{1}{2}}_1\|B^*DB\|_1^{-\frac{1}{2}}\|(A\otimes\1)X(B\otimes\1)\|_{(P,t)} \\
    &\qquad= \sup_{\tilde{A},\tilde{B}\geq0}\|\tilde{A}^\frac{1}{2}C\tilde{A}^\frac{1}{2}\|^{-\frac{1}{2}}_1\|\tilde{B}^\frac{1}{2}D\tilde{B}^\frac{1}{2}\|_1^{-\frac{1}{2}}\|(\tilde{A}^{\frac{1}{2}}\otimes\1)X(\tilde{B}^\frac{1}{2}\otimes\1)\|_{(P,t)} \\
    &\qquad= \sup_{\tilde{A},\tilde{B}\geq0}\left\{-\frac{1}{2}\|\tilde{A}^\frac{1}{2}C\tilde{A}^\frac{1}{2}\|_1-\frac{1}{2}\|\tilde{B}^\frac{1}{2}D\tilde{B}^\frac{1}{2}\|_1+2\|(\tilde{A}^{\frac{1}{2}}\otimes\1)X(\tilde{B}^\frac{1}{2}\otimes\1)\|_{(P,t)}^\frac{1}{2} \right\}
\end{align} where the last line follows from \cref{lem:Sup.Prod.to.Sum} with $r_1=1, r_2=-\frac{1}{4}$. Now observe that
$C\mapsto -\|\tilde{A}^\frac{1}{2}C\tilde{A}^\frac{1}{2}\|_1= -\Tr[\tilde{A}^\frac{1}{2}C\tilde{A}^\frac{1}{2}]$ is linear and hence in particular convex for any $\tilde{A}\geq0$ and the same holds for $D$. The claim now follows since the sum of a function convex in $C$ and one in $D$ is jointly convex in $(C,D)$ and this joint convexity is preserved by the supremum.

Without loss of generality, assume $p\neq t$, and thus $s<\infty$, since else we are done by \cref{lem:inverseDef1}.
Having established this we may now proceed as in the other case by applying \cref{equ:Def1.2} and \cref{equ:Def1.1} to get
\begin{align}
    &\|X\|_{(q,t)} \\
    &\quad = \sup_{a,b\geq0}\|a\|_{2m}^{-1}\|b\|_{2m}^{-1}\|(a\otimes\1)X(b\otimes\1)\|_t \\ 
    &\quad= \sup_{a,b\geq0}\inf_{\underset{}{c,d\geq0}}\|a\|_{2m}^{-1}\|b\|_{2m}^{-1}\|c\|_{2s}\|d\|_{2s}\|(\underbrace{c^{-1}a}_{C^{-1}}\otimes\1)X(\underbrace{bd^{-1}}_{D^{-1}}\otimes\1)\|_{(p,t)} \\
    &\quad = \sup_{a,b\geq0}\inf_{\underset{}{C,D}}\|a\|_{2m}^{-1}\|b\|_{2m}^{-1}\|aCC^*a\|^\frac{1}{2}_{s}\|bD^*Db\|^\frac{1}{2}_{s}\|(\underbrace{c^{-1}a}_{C^{-1}}\otimes\1)X(D^{-1}\otimes\1)\|_{(p,t)} \\
    &\quad = \sup_{a,b\geq0}\inf_{\underset{}{\tilde{C},\tilde{D}\geq0}}\|a\|_{2m}^{-1}\|b\|_{2m}^{-1}\|a\tilde{C}a\|^\frac{1}{2}_{s}\|b\tilde{D}b\|^\frac{1}{2}_{s}\|(\tilde{C}^{-\frac{1}{2}}\otimes\1)X(\tilde{D}^{-\frac{1}{2}}\otimes\1)\|_{(p,t)} \\ 
    &\quad = \left\{\sup_{\underset{\|a\|_1,\|b\|_1\leq1}{a,b\geq0}}\!\frac{1}{s+1}\inf_{\underset{}{\tilde{C},\tilde{D}\geq0}}\frac{1}{2}\|a^\frac{1}{2m}\tilde{C}a^\frac{1}{2m}\|^s_{s}+\frac{1}{2}\|b^\frac{1}{2m}\tilde{D}b^\frac{1}{2m}\|^s_{s}+s\|(\tilde{C}^{-\frac{1}{2}}\otimes\1)X(\tilde{D}^{-\frac{1}{2}}\otimes\1)\|_{(p,t)} \right\}^\frac{s+1}{s}
\end{align} 
where the last equality comes from \cref{lem:Inf.Prod.to.Sum} with $r_1=s<\infty,\, r_2=\frac{1}{2}$. As in the previous proofs, the change of variables calculations are justified by \cref{lem:SupportConditions}.
By the previous claim, the third summand is jointly convex in $\tilde{C},\tilde{D}$. By \cref{lem:continuity} $a\mapsto\|a^\frac{1}{2m}\tilde{C}a^\frac{1}{2m}\|^s_{s}$ is continuous and concave.
Additionally by \cite[Lemma B.1 3.]{Rubboli.2025} it is convex in $\tilde{C}$ as $1\leq s$. Thus the whole sum is jointly convex in $\tilde{C},\tilde{D}$ and jointly concave and continuous in $a,b$ and we may apply Sion's minimax theorem, yielding
\begin{align}
  \|X\|_{(q,t)} &= \inf_{\underset{}{\tilde{C},\tilde{D}\geq0}}\sup_{a,b\geq0}\|a\|_{2m}^{-1}\|b\|_{2m}^{-1}\|a\tilde{C}a\|^\frac{1}{2}_{s}\|b\tilde{D}b\|^\frac{1}{2}_{s}\|(\tilde{C}^{-\frac{1}{2}}\otimes\1)X(\tilde{D}^{-\frac{1}{2}}\otimes\1)\|_{(p,t)} \\ 
  &= \inf_{\underset{}{\tilde{C},\tilde{D}\geq0}}\|\tilde{C}\|^\frac{1}{2}_{r}\|\tilde{D}\|^\frac{1}{2}_{r}\|(\tilde{C}^{-\frac{1}{2}}\otimes\1)X(\tilde{D}^{-\frac{1}{2}}\otimes\1)\|_{(p,t)} \\ 
  &= \inf_{\underset{}{\tilde{C},\tilde{D}\geq0}}\|\tilde{C}\|_{2r}\|\tilde{D}\|_{2r}\|(\tilde{C}^{-1}\otimes\1)X(\tilde{D}^{-1}\otimes\1)\|_{(p,t)},
\end{align} where the second equality follows from $\frac{1}{s}-\frac{1}{m}=\frac{1}{r}$ and \cite[Lemma 3.1 (ii)]{Fawzi.2026} with $\cX=\mathbb{C}$ and the last equality is a renaming of $\tilde{C}\to \tilde{C}^2, \tilde{D}\to \tilde{D}^2$.

In the case of $X\geq0$, in all the cases 1,3,5,6 we can take $a=b^*\geq0$ and $c=d^*\geq0$ by \cref{lem:Positive.Symmetrie} and \cref{lem:inverseDef1}. This implies in all chases that one can choose $C=D^*$ or $C^{-1}=D^{*\,-1}$, which gives the respective symmetric claimed statement for $X\geq0$. This completes the whole proof.
\end{proof}

\subsection{Quasi-Norm Property}\label{sec:QuasiNorms.Proof}

In the following we show that the 2-indexed functionals $\|\cdot\|_{(q,p)}$ we defined in \cref{def:Definition1} are $\kappa$-normable quasi-norms for $\kappa=\min\{q,p,1\}$ (cf.~\Cref{definekappanormable}). 
\begin{theorem*}[\cref{thm:QuasiNomrms}]
Let $0<q,p$ with $\big|\frac{1}{q}-\frac{1}{p}\big|\leq1$ and set $\kappa:=\min\{q,p,1\}$.
Then the functionals $\|\cdot\|_{(q,p)}$ defined in \cref{def:Definition1} are $\kappa$-normable quasi-norms on $\cB(\cH_1\otimes\cH_2)$. 
\end{theorem*}

\noindent In order to prove \cref{thm:QuasiNomrms}, we first establish the following factorization which relates the $\|\cdot\|_{(q,p)}$-quasi-norm through the $q$-quasi-norm with a supremum, whenever $q\leq p$. 

\begin{lemma}\label{lem:technical.sup.lemma}
Let $0<q\leq p$ and $q\leq 1$, such that $\frac{1}{r}=\frac{1}{q}-\frac{1}{p}\leq 1$, then it holds for $X\in\cB(\cH_1\otimes\cH_2)$ that
\begin{align}
    \|X\|_{(q,p)}=\sup_{A,B}\|A^*A\|^{-\frac{1}{2}}_{(\infty,r)}\|BB^*\|^{-\frac{1}{2}}_{(\infty,r)}\|AXB\|_{q}
\end{align}
\end{lemma}

\begin{proof}
We prove this by concatenating the suitable factorization formulas to first relate the $(q,p)$-quasi norm to the $(p,p)$-(quasi-)norm and then to the $(q,q)$-quasi-norm via \cref{prop:SpVariationalFormulas}. The result will then follow after an application of \cref{lem:Sup.Prod.to.Sum} and suitable semi-continuity and joint convexity/concavity arguments allowing an application of Sion's mininax theorem. More precisely, by Theorem \ref{thm:qt.to.pt},
\begin{align}
    \|X\|_{(q,p)} &=\inf_{\underset{}{a,b\geq0}}\|a\|_{2r}\|b\|_{2r}\|a_1^{-1}X_{12}b_1^{-1}\|_{p} \\ 
    &=\inf_{\underset{}{a,b\geq0}}\sup_{c,d\geq0}\|a\|_{2r}\|b\|_{2r}\|c\|_{2r}^{-1}\|d\|_{2r}^{-1}\|c_{12}a_1^{-1}X_{12}b_1^{-1}d_{12}\|_{q} \\ 
    &= \inf_{\underset{}{a,b\geq0, \|a\|_1,\|b\|_1\leq 1}}\sup_{\underset{\textup{}}{C,D}}\|a_1^\frac{1}{2r}C_{12}^*C_{12}a_1^\frac{1}{2r}\|_{r}^{-\frac{1}{2}}\|b_1^\frac{1}{2r}D_{12}D_{12}^*b_1^\frac{1}{2r}\|_{r}^{-\frac{1}{2}}\|CXD\|_{q} \\ 
    &= \inf_{\underset{}{a,b\geq0, \|a\|_1,\|b\|_1\leq 1}}\sup_{\underset{\textup{}}{\tilde{C},\tilde{D}\geq0}}\|a_1^\frac{1}{2r}\tilde{C}_{12}a_1^\frac{1}{2r}\|_{r}^{-\frac{1}{2}}\|b_1^\frac{1}{2r}\tilde{D}_{12}b_1^\frac{1}{2r}\|_{r}^{-\frac{1}{2}}\|\tilde{C}^\frac{1}{2}X\tilde{D}^\frac{1}{2}\|_{q} \\
    &= \!\Bigg(\inf_{\underset{a,b\geq0}{\|a\|_1,\|b\|_1\leq 1 }}\sup_{\underset{\textup{}}{\tilde{C},\tilde{D}\geq0}}\Big\{\!\!-\frac{p}{2r}\|a_1^\frac{1}{2r}\tilde{C}_{12}a_1^\frac{1}{2r}\|_{r}^{r}-\frac{p}{2r}\|b_1^\frac{1}{2r}\tilde{D}_{12}b_1^\frac{1}{2r}\|_{r}^{r} +\frac{p}{q
    }\|\tilde{C}^\frac{1}{2}X\tilde{D}^\frac{1}{2}\|^{q}_{q}\Big\} \Bigg)^\frac{1}{p}
\end{align}
where the last line follows with Lemma \ref{lem:Sup.Prod.to.Sum} 
As before,  $\|a_1^\frac{1}{2r}\tilde{C}_{12}a_1^\frac{1}{2r}\|_{r}^{r}$ is convex in $\tilde{C}$ and concave and continuous in $a$ due to \cref{lem:continuity} for any $\tilde{C}\geq0$. The same applies to the norm involving $b$ and $\tilde{D}$. Likewise 
\begin{align}
    \|\tilde{C}^\frac{1}{2}X\tilde{D}^\frac{1}{2}\|^q_{q} = \Tr[(\tilde{C}^\frac{1}{2}X\tilde{D}X^*\tilde{C}^\frac{1}{2})^{\frac{q}{2}}]
\end{align} is jointly concave in $(\tilde{C},\tilde{D})$ due to \cite[Theorem 1.1 (1)]{Zhang.2020} for $q\leq1$. As before these imply that Sion's minimax theorem can be applied. Thus we get
\begin{align}
    \|X\|_{(q,p)} &= \sup_{C,D}\inf_{\underset{}{a,b\geq0, \|a\|_1,\|b\|_1\leq 1}}\|a^\frac{1}{2r}_1C_{12}^*C_{12}a_1^\frac{1}{2r}\|_{r}^{-\frac{1}{2}}\|b_1^\frac{1}{2r}D_{12}D_{12}^*b_1^\frac{1}{2r}\|_{r}^{-\frac{1}{2}}\|CXD\|_{q} \\ 
    &= \sup_{C,D}\|C^*C\|^{-\frac{1}{2}}_{(\infty,r)}\|DD^*\|^{-\frac{1}{2}}_{(\infty,r)}\|CXD\|_{q},
\end{align} 
where the last equality follows from \cref{lem:Positive.Symmetrie}, or equivalently the factorization formulas for operator valued Schatten spaces, e.g.~\cite[Lemma 1.7]{Book.Pisier.1998}. 
\end{proof}
\noindent We are now ready to prove the quasi-norm property.
\begin{proof}[Proof of \cref{thm:QuasiNomrms}]
The positive homogeneity and positive definiteness follow more directly from the fact that the underlying Schatten norms possess these properties. Thus, it is easily checked that \eqref{equ:Def1.1} and \eqref{equ:Def1.2} preserve these.
The non-trivial part of this proof is proving the $\kappa$-normability for $\kappa<1$ since else we are in the well established normed setting \cite{Book.Pisier.1998, Beigi.2023, Devetak.2006}. Consider first the case $q\geq p$ so $\kappa=p<1$. Now by \eqref{equ:Def1.2} we have
\begin{align}
    \|X+Y\|^p_{(q,p)}&=\sup_{a,b\geq0}\|a\|^{-p}_{2r}\|b\|^{-p}_{2r}\|a_1(X+Y)_{12}b_1\|^p_{p} \\ &\leq \sup_{a,b\geq0}\|a\|^{-p}_{2r}\|b\|^{-p}_{2r}\big(\|a_1X_{12}b_1\|_p^p+\|a_1Y_{12}b_1\|^p_{p}\big) \\ 
    &\leq \sup_{a,b\geq0}\|a\|^{-p}_{2r}\|b\|^{-p}_{2r}\|a_1X_{12}b_1\|_p^p + \sup_{a,b\geq0}\|a\|^{-p}_{2r}\|b\|^{-p}_{2r}\|a_1Y_{12}b_1\|_p^p \\ 
    &= \|X\|^p_{(q,p)}+\|Y\|^p_{(q,p)},
\end{align} 
where the first inequality follows from the $p$-normability of the Schatten-$p$-quasi norm and the second one by the subadditivity of the supremum.

\medskip

\noindent We now consider the case $q\leq p$ so $\kappa=q<1$. 
To proceed we use the factorization formula established in \cref{lem:technical.sup.lemma} and the $q$-normability of the Schatten $q$-quasi-norm to get
\begin{align}
     \|X+Y\|^q_{(q,p)}&=\sup_{A,B}\|A^*A\|^{-\frac{q}{2}}_{(\infty,r)}\|BB^*\|^{-\frac{q}{2}}_{(\infty,r)}\|A(X+Y)B\|^q_{q} \\
     &\leq \sup_{A,B}\|A^*A\|^{-\frac{q}{2}}_{(\infty,r)}\|BB^*\|^{-\frac{q}{2}}_{(\infty,r)}\big(\|AXB\|_q^q+\|AYB\|_q^q\big) \\
     &\leq \sup_{A,B}\|A^*A\|^{-\frac{q}{2}}_{(\infty,r)}\|BB^*\|^{-\frac{q}{2}}_{(\infty,r)}\|AXB\|_q^q+ \sup_{A,B}\|A^*A\|^{-\frac{q}{2}}_{(\infty,r)}\|BB^*\|^{-\frac{q}{2}}_{(\infty,r)}\|AYB\|_q^q
     \\ &= \|X\|^q_{(q,p)}+ \|Y\|^q_{(q,p)}.
\end{align}
\end{proof}

\subsection{A 2-Indexed Reverse Hölder Inequality}

In the normed setting the 2-indexed Schatten norms satisfy, as a defining property, a duality relation, see e.g.~\cite[Proposition 4.3 ii)]{Beigi.2023}. The latter breaks down for strictly quasi-norms with indices $<1$. For Schatten-quasi-norms there exists, however, a reverse Hölder's duality which we here extend to two-indexed Schatten quasi-norms. We do this by factoring through the trace-norm. For positive operators this directly yields a reverse Hölder's inequality for 2-indexed Schatten-quasi-norms, generalizing \cite[Lemma 3.4, (3.13)]{Book.Tomamichel.2016}.

\begin{lemma*}[Generalized reverse Hölder's inequality \cref{lem:GenRevHölder}]
Let $0< q,p\leq 1$ with $\big|\frac{1}{q}-\frac{1}{p}\big|\leq 1$ and let $q^\prime, p^\prime$ be their Hölder dual indices, that is $\frac{1}{q}+\frac{1}{q^\prime}=\frac{1}{p}+\frac{1}{p^\prime}=1$. Then it holds that
\begin{align}
    \|X\|_{(q,p)}=\inf_{\underset{\Pi_A=\Pi_{XX^*}, \Pi_B=\Pi_{X^*X}}{A,B\geq0}}\|A\|^\frac{1}{2}_{(-q^\prime,-p^\prime)}\|B\|^\frac{1}{2}_{(-q^\prime,-p^\prime)}\|A^{-\frac{1}{2}}XB^{-\frac{1}{2}}\|_1.
\end{align} 
For positive $X\geq 0$ this simplifies to 
\begin{align}
    \|X\|_{(q,p)}=\inf_{\underset{X=\Pi_YX\Pi_Y}{Y\geq0}}\Tr[Y^*X]\, \|Y^{-1}\|_{(-q^\prime,-p^\prime)}.
\end{align} where $X=\Pi_YX\Pi_Y$ means that $\ker(Y)\subset\ker(X)$.
\end{lemma*}
\noindent Note that in the above lemma the Hölder dual indices $-q^\prime,-p^\prime\in (0,\infty]$ for $0<q,p\le 1$ and $-q^\prime,-p^\prime\in[1,\infty]$ for $\frac{1}{2}\leq q,p\leq 1$.

\begin{proof}[Proof of \cref{lem:GenRevHölder}]
We consider first the case $0< q\leq p\leq 1$ with $\frac{1}{r}:=\frac{1}{q}-\frac{1}{p}\leq 1$, then we have by Lemma \ref{lem:SupportConditions}
\begin{align}
 \|X\|_{(q,p)} &=\inf_{\underset{}{a,b\geq0}}\|a\|_{2r}\|b\|_{2r}\|a_1^{-1}X_{12}b_1^{-1}\|_p \\
 &=\inf_{\underset{}{a,b\geq0}}\inf_{c,d\geq0}\|a\|_{2r}\|b\|_{2r}\|c\|_{-2p^\prime}\|d\|_{-2p\prime}\|c_{12}^{-1}a_1^{-1}X_{12}b_1^{-1}d_{12}^{-1}\|_1 \\ 
 &= \inf_{a,b\geq0}\inf_{C,D}\|a\|_{2r}\|b\|_{2r}\|a_1^{-1}C_{12}C_{12}^*a_1^{-1}\|_{-p^\prime}^\frac{1}{2}\|b_1^{-1}D_{12}^*D_{12}b_1^{-1}\|_{-p^\prime}^\frac{1}{2}\|C^{-1}XD^{-1}\|_1 \\ 
 &=\inf_{C,D}\inf_{a,b\geq0}\|a\|_{2r}\|b\|_{2r}\|a_1^{-1}C_{12}C_{12}^*a_1^{-1}\|_{-p^\prime}^\frac{1}{2}\|b_1^{-1}D_{12}^*D_{12}b_1^{-1}\|_{-p^\prime}^\frac{1}{2}\|C^{-1}XD^{-1}\|_1 \\ 
 &=\inf_{C,D}\|CC^*\|_{(-q^\prime,-p^\prime)}^\frac{1}{2}\|D^*D\|_{(-q^\prime,-p^\prime)}^\frac{1}{2}\|C^{-1}XD^{-1}\|_1 \\ 
 &=\inf_{\tilde{C},B\geq0}\|\tilde{C}\|_{(-q^\prime,-p^\prime)}^\frac{1}{2}\|\tilde{D}\|_{(-q^\prime,-p^\prime)}^\frac{1}{2}\|\tilde{C}^{-\frac{1}{2}}X\tilde{D}^{-\frac{1}{2}}\|_1,
\end{align} 
where in the second line, we used the variational formulas of \cref{prop:SpVariationalFormulas}. Above, we suppressed writing the support conditions of \cref{lem:SupportConditions} on $a,b,c,d$ which are $\Pi_a=L_X, \Pi_b=R_X$, $\Pi_c=\Pi_{XX^*}\leq L_X\otimes\1$ is the projection onto the image of $X$ and $\Pi_d=\Pi_{X^*X}\leq R_X\otimes\1$ the projection onto the support of $X$. This makes $\tilde{C}:=CC^*, \tilde{D}:=D^*D$, with $C^{-1}=c^{-1}a_1^{-1}, D^{-1}=b_1^{-1}d^{-1}$ well defined and such that $\Pi_{\tilde{C}}=\Pi_{XX^*}$ and $\Pi_{\tilde{D}}=\Pi_{X^*X}$, i.e.  $X=\Pi_{\tilde{C}}X\Pi_{\tilde{D}}$ as claimed in the statement of the lemma. To get the second to last line we used that $\frac{1}{r}=\frac{1}{q}-\frac{1}{p}=\frac{1}{-q^\prime}-\frac{1}{-p^\prime}$ and $-q^\prime\leq -p^\prime$ together with \cref{lem:Positive.Symmetrie}. 

\medskip

\noindent Next we consider the other case $0< p\leq q\leq 1$ with $\frac{1}{r}:=\frac{1}{p}-\frac{1}{q}\leq 1$. This is slightly more involved and requires an application of Sion's minimax theorem. Let $P:=\min\{1,\frac{p}{1-p}\}$ consider first the case $P=1$ ($p\geq \frac{1}{2}$), then by \cref{def:Definition1} we have
\begin{align}
    \|X\|_{(q,p)}&=\sup_{a,b\geq0}\|a\|_{2r}^{-1}\|b\|_{2r}^{-1}\|a_1X_{12}b_1\|_p \\
    &= \sup_{a,b\geq0}\inf_{c,d\geq0}\|a\|_{2r}^{-1}\|b\|_{2r}^{-1}\|c\|_{-2p^\prime}\|d\|_{-2p^\prime}\|c_{12}^{-1}a_1X_{12}b_1d_{12}^{-1}\|_1 \\ 
    &= \sup_{a,b\geq0}\inf_{\tilde{C},\tilde{D}\geq0}\|a\|_{2r}^{-1}\|b\|_{2r}^{-1}\|a_1\tilde{C}_{12}a_1\|^\frac{1}{2}_{-p^\prime}\|b_1'\tilde{D}_{12}^*b_1\|^\frac{1}{2}_{-p^\prime}\|\tilde{C}^{-\frac{1}{2}}X\tilde{D}^{-\frac{1}{2}}\|_1 \\
    &= \left(\sup_{\underset{\|a\|_1,\|b\|_1= 1}{a,b\geq0}}\!\!\!\!\inf_{\tilde{C},\tilde{D}\geq0}\!\left\{\!\frac{p}{-2p^\prime}\|a_1^\frac{1}{2r}\tilde{C}_{12}a_1^\frac{1}{2r}\|_{-p^\prime}^{-p^\prime}\!+\!\frac{p}{-2p^\prime}\|b_1^\frac{1}{2r}\tilde{D}_{12}b_1^\frac{1}{2r}\|_{-p^\prime}^{-p^\prime}\!+\!p\|\tilde{C}^{-\frac{1}{2}}X\tilde{D}^{-\frac{1}{2}}\|_1\!\right\}\!\right)^\frac{1}{p},
\end{align}
where in the second line, we used once again the variational formulas of \cref{prop:SpVariationalFormulas}, and \cref{lem:Inf.Prod.to.Sum} in the last.
As above, by \cref{lem:SupportConditions}, we can assume that $\Pi_a=L_X, \Pi_b=R_X, \Pi_c=\Pi_{XX^*}\leq L_X\otimes\1, \Pi_d=\Pi_{X^*X}\leq R_X\otimes\1$, which makes $C:=a_1^{-1}c, D:=db_1^{-1}$ well defined and s.t. $\tilde{C}:=CC^*, \tilde{D}:=D^*D$ satisfy $\Pi_{\tilde{C}}=\Pi_{XX^*}, \Pi_{\tilde D}=\Pi_{X^*X}$. 
To apply Sion's minimax theorem we notice that due to \cite[Theorem 1.1 2)]{Zhang.2020} the last summand $\|\tilde{C}^{-\frac{1}{2}}X\tilde{D}^{-\frac{1}{2}}\|_1$ is jointly convex in $\tilde{C},\tilde{D}$ and due to \cite[Lemma B.1 3)]{Rubboli.2025} the first two are as  $-p^\prime\geq 1 \Leftrightarrow p\geq \frac{1}{2}$. Due to \cref{lem:continuity}, which applies since $\frac{1}{r}=\frac{1}{-p^\prime}-\frac{1}{-q^\prime}\leq \frac{1}{-p^\prime}$, the first two summand possess the desired continuity and concavity in $a,b$ and we can apply Sion's minimax theorem (see \Cref{thm:Sion}). 
Thus,
\begin{align}
    \|X\|_{(q,p)}&=\left(\inf_{\tilde{C},\tilde{D}\geq0}\sup_{\underset{\|a\|_1,\|b\|_1=1}{a,b\geq0}}\left\{\frac{p}{-2p^\prime}\|a_1^\frac{1}{2r}\tilde{C}_{12}a_1^\frac{1}{2r}\|_{-p^\prime}^{-p^\prime}+\frac{p}{-2p^\prime}\|b_1^\frac{1}{2r}\tilde{D}_{12}b_1^\frac{1}{2r}\|_{-p^\prime}^{-p^\prime}+p\|\tilde{C}^{-\frac{1}{2}}X\tilde{D}^{-\frac{1}{2}}\|_1\right\}\right)^\frac{1}{p} \\
    &=\left(\inf_{\tilde{C},\tilde{D}\geq0}\left\{\frac{p}{-2p^\prime}\|\tilde{C}\|_{(-q^\prime,-p^\prime)}^{-p^\prime}+\frac{p}{-2p^\prime}\|\tilde{D}\|_{(-q^\prime,-p^\prime)}^{-p^\prime}+p\|\tilde{C}^{-\frac{1}{2}}X\tilde{D}^{-\frac{1}{2}}\|_1\right\}\right)^\frac{1}{p} \\
    &=\inf_{\tilde{C},\tilde{D}\geq0}\|\tilde{C}\|^{\frac{1}{2}}_{(-q^\prime,-p^\prime)}\|\tilde{D}\|^{\frac{1}{2}}_{(-q^\prime,-p^\prime)}\|\tilde{C}^{-\frac{1}{2}}X\tilde{D}^{-\frac{1}{2}}\|_1,
\end{align} where we used \cite[Lemma 3.1]{Fawzi.2026} with $\frac{1}{r}=\frac{1}{-p^\prime}-\frac{1}{-q^\prime}$ as before and \cref{lem:Inf.Prod.to.Sum} to turn the sum back into a product. 
Now in the case $P<1$ repeating the above steps we get
 \begin{align}
     \|X\|_{(q,p)}&= \Bigg(\sup_{\underset{\|a\|_1,\|b\|_1= 1}{a,b\geq0}}\!\!\!\!\inf_{\tilde{C},\tilde{D}\geq0}\!\bigg\{\!\frac{P}{2(1+P)}\|a_1^\frac{1}{2r}\tilde{C}_{12}a_1^\frac{1}{2r}\|_{1}\!+\!\frac{P}{2(1+P)}\|b_1^\frac{1}{2r}\tilde{D}_{12}b_1^\frac{1}{2r}\|_{1}\\&\qquad\qquad\qquad\qquad+\frac{1}{1+P}\|\tilde{C}^{-\frac{1}{2}}X\tilde{D}^{-\frac{1}{2}}\|^P_P\!\bigg\}\!\Bigg)^\frac{1+P}{P},
 \end{align} to which, similarly to above, Sions minimax theorem applies as the first two summands are clearly convex in $\tilde C, \tilde D$ and the last one is jointly convex by \cite[Theorem 1.1 (2)]{Zhang.2020} since $0<P$ and continutity and concavity in $a,b$ of the whole expression is exactly as above. Hence as above we follow
 \begin{align}\label{equ:proof.intermediate}
     \|X\|_{(q,p)}=\inf_{\tilde{C},\tilde{D}\geq0}\|\tilde{C}\|^{\frac{1}{2}}_{(s,1)}\|\tilde{D}\|^{\frac{1}{2}}_{(s,1)}\|\tilde{C}^{-\frac{1}{2}}X\tilde{D}^{-\frac{1}{2}}\|_P,
 \end{align} where $\frac{1}{s}=\frac{1}{q}-\frac{1}{P}\in[0,1]$. Notice that for positive $X\geq0$ by \cref{lem:Positive.Symmetrie} and \cref{lem:inverseDef1} we could have chosen $a=b^*\geq0$ yielding ultimately $\tilde{C}=\tilde{D}$ in the above statement. Note also that this works equally for any $0<p<1$ so \eqref{equ:proof.intermediate} holds hence for arbitrary $0<P$. We will now use this to prove the in the Lemma claimed statement.
 
 Applying \cref{prop:SpVariationalFormulas} to this ultimately yields
 \begin{align}
     \|X\|_{(q,p)}&=\inf_{\tilde{C},\tilde{D}\geq0}\|\tilde{C}\|^{\frac{1}{2}}_{(s,1)}\|\tilde{D}\|^{\frac{1}{2}}_{(s,1)}\|\tilde{C}^{-\frac{1}{2}}X\tilde{D}^{-\frac{1}{2}}\|_P \\&=\inf_{\tilde{C},\tilde{D}\geq0}\inf_{E,F\geq0}\|\tilde{C}\|^{\frac{1}{2}}_{(s,1)}\|\tilde{D}\|^{\frac{1}{2}}_{(s,1)}\|EE^*\|_{-2P^\prime}\|F^*F\|_{-2P^\prime}\|E^{-1}\tilde{C}^{-\frac{1}{2}}X\tilde{D}^{-\frac{1}{2}}F^{-1}\|_1 \\&=\inf_{A,B}\inf_{\tilde{C},\tilde{D}\geq0}\|\tilde{C}\|^{\frac{1}{2}}_{(s,1)}\|\tilde{D}\|^{\frac{1}{2}}_{(s,1)}\|\tilde C^{-\frac{1}{2}}AA^*\tilde C^{-\frac{1}{2}}\|^\frac{1}{2}_{-P^\prime}\|\tilde D^{-\frac{1}{2}}B^*B\tilde D^{-\frac{1}{2}}\|^\frac{1}{2}_{-P^\prime}\|A^{-1}XB^{-1}\|_1 \\
     &=\inf_{A,B}\|AA^*\|^\frac{1}{2}_{(-q^\prime, -p^\prime)}\|B^*B\|^\frac{1}{2}_{(-q^\prime, -p^\prime)}\|A^{-1}XB^{-1}\|_1
\end{align} where we suppressed writing the support conditions explicitly. The last equality follows from the aforementioned positive symmetric version of \eqref{equ:proof.intermediate} applied to 
\begin{align}
    \inf_{\tilde C\geq0}\|\tilde{C}\|_{(s,1)}\|\tilde C^{-\frac{1}{2}}AA^*\tilde C^{-\frac{1}{2}}\|_{-P^\prime} = \|AA^*\|_{(-q^\prime,-p^\prime)}
\end{align} since $-p^\prime\leq -q^\prime\leq -P^\prime$ where they are each compatible and s.t. $\frac{1}{-q^\prime}-\frac{1}{-P^\prime}=\frac{1}{q}-\frac{1}{P}=\frac{1}{s}$ and $\frac{1}{-p^\prime}-\frac{1}{-P^\prime}=1$.

\medskip

\noindent To get the statement for $X\geq0$ from \eqref{equ:GenRevHölder}, we simply need to observe that  the above proof holds with $A=B$ due to \cref{lem:Positive.Symmetrie}. Then since $\|A^{-\frac{1}{2}}XA^{-\frac{1}{2}}\|_1=\Tr[A^{-\frac{1}{2}}XA^{-\frac{1}{2}}]=\Tr[XA^{-1}]$ the statement 
\begin{align}
   \|X\|_{(q,p)}=\inf_{\underset{X=(\Pi_A\otimes\1)X(\Pi_A\otimes\1)}{A\geq0}}\|A\|_{(-q^\prime,-p^\prime)}\Tr[XA^{-1}]= \inf_{\underset{X=\Pi_YX\Pi_Y}{Y\geq0}}\|Y^{-1}\|_{(-q^\prime,-p^\prime)}\Tr[YX]
\end{align} follows, where $Y^{-1}=A$ is the generalized Moore-Penrose inverse of $A$. 
\end{proof}

\section{Applications to Quantum Information Theory}\label{sec:Applications}

\noindent In what follows, we make use of standard notations from quantum information theory. For instance, subsystems will be labeled by capital letters $A,Q,P$, etc.~. Channels $\Phi:\cB(\cH_Q)\to \cB(\cH_P)$ will be denoted as $\Phi:Q\to P$, and norms will inherit these notations: for instance, the two-index Schatten (quasi-)norm of a bipartite operator $X_{AB}$ will we written as $$\|X_{AB}\|_{(A:a,B:b)}\equiv \|X\|_{(a,b)}\,.$$

\subsection{Entropic Quantities via Quasi-Norms}\label{sec:Apl:Entropies}

First recall that the sandwiched Rényi relative entropy \cite{Wilde.2014, Book.Tomamichel.2016} between two states $\rho$ and $\sigma$ with $\operatorname{supp}(\rho)\subseteq\operatorname{supp}(\sigma)$ for $\frac{1}{2}\leq\alpha\leq\infty$ is defined as
\begin{align}
    D_\alpha(\rho\|\sigma):=\frac{\alpha}{\alpha-1}\log\|\sigma^{\frac{1-\alpha}{2\alpha}}\rho\sigma^{\frac{1-\alpha}{2\alpha}}\|_\alpha.
\end{align}

\subsubsection{Sandwiched Conditional Entropy}

In \cite{Beigi.2023, Fawzi.2026}, but already going back to expressions from \cite{Devetak.2006}, it was observed and exploited that the optimized conditional sandwiched Rényi entropy can be expressed via a 2-indexed norm. Specifically for any $\alpha>1$ and bipartite quantum state $\rho_{AB}\in\cD(\cH_{A}\otimes\cH_B)$, it is shown that
\begin{align}
H^\uparrow_\alpha(A|B)_\rho:=\sup_{\sigma_B\in\cD(\cH_B)}-D_\alpha(\rho_{AB}\|\1\otimes\sigma_B) = \frac{\alpha}{1-\alpha}\log \|\rho_{BA}\|_{(B:1, A:\alpha)}.
\end{align} 
 This rewriting as a 2-indexed norm has been used to give simple proofs of standard properties of the optimized conditional R\'{e}nyi entropy, including continuity bounds \cite{Beigi.2023} and chain rules \cite{Fawzi.2026}. In the following, we prove that this correspondence continues to hold for $\alpha\in[\frac{1}{2},1)$. 
\begin{corollary}\label{cor:conditional.entropy}
Let $\alpha\in[\frac{1}{2},\infty]$ and $\rho_{AB} \in \cD(\cH_A \otimes \cH_B)$, then it holds that
\begin{align}
  H^\uparrow_\alpha(A|B)_\rho = \frac{\alpha}{1-\alpha}\log\|\rho_{BA}\|_{(B:1,A:\alpha)}.
\end{align}
\end{corollary}
\begin{proof}
The case $\alpha \geq 1$ is already covered in~\cite{Beigi.2023} so we focus on the case $\frac{1}{2}\leq \alpha<1$. By \cref{lem:Positive.Symmetrie},  for a positive operator $\rho\geq0$, we have
\begin{align}
\|\rho_{BA}\|_{(B:1,A:\alpha)}&=\sup_{\sigma_B\geq0}\|\sigma_B\|^{-2}_{-2\alpha^\prime}\|\sigma_B\rho_{BA}\sigma_B\|_\alpha =\sup_{\sigma\geq0, \Tr[\sigma]=1}\|\sigma_B^{-\frac{1}{2\alpha^\prime}}\rho_{BA}\sigma_B^{-\frac{1}{2\alpha^\prime}}\|_\alpha. 
\end{align} where $\alpha'$ denotes the H\"{o}lder conjugate of $\alpha$.
\end{proof}

\noindent We believe that the identification of the optimized conditional sandwiched Rényi entropy as a quasi-norm is of independent interest. As a simple first illustration of this, we recover a well-known continuity bound \cite{Bluhm.2024, Marwah.2022}, in analogy with the regime $\alpha>1$ from~\cite[Theorem 5.2]{Beigi.2023}.

\begin{proposition}[Continuity bound]\label{thm:continuity.bound}
Let $\alpha\in[\frac{1}{2},1)$ and $\rho_{AB},\sigma_{AB}$ be bipartite quantum states  s.t. $\frac{1}{2}\|\rho-\sigma\|_1\leq\epsilon$, then
\begin{align}
  \big|  H^\uparrow_\alpha(A|B)_\rho-H^\uparrow_\alpha(A|B)_\sigma\big|\leq \frac{1}{1-\alpha}\log\big(1+2\epsilon^\alpha d_A^{2(1-\alpha)}\big),
\end{align} where $d_A$ is the dimension of system $A$.
\end{proposition}
\noindent Before we prove this continuity bound we require the following lemma which bounds the $\|\cdot\|_{(1,\alpha)}$-quasi-norm in terms of the trace norm.
\begin{lemma}\label{lem:1.alpha.bound}
Let $\rho_{AB}$ be a bipartite quantum state and $\frac{1}{2}\leq\alpha\leq 1$, then
\begin{align}
   d_A^{\frac{\alpha-1}{\alpha}}\le  \|\rho_{BA}\|_{(B:1,A:\alpha)} \leq d_A^{\frac{1-\alpha}{\alpha}}
\end{align}
\end{lemma}
\begin{proof}
For the first inequality denote with $\{W^{(k)}\}_{k=1}^{d_A^2}$ the unitary Heisenberg-Weyl operators, then using that $d_A^{-2}\sum_{k=1}^{d^2}W_A^{(k)}\sigma_AW^{(k) *}_A=\Tr_A[\sigma_A]\frac{\1_A}{d_A}$ we have, using \cref{tensorproduct} and \cref{thm:QuasiNomrms}
\begin{align}
d_A^{1-\alpha}&= \|\rho_B\|_1^\alpha\,\Big\|\frac{\1}{d_A}\Big\|_\alpha^\alpha = \Big\|\rho_B\otimes\frac{\1}{d_A}\Big\|_{(B:1,A:\alpha)}^\alpha \\&= \Big\|d_A^{-2}\sum_{k=1}^{d^2} W^{(k)}_A\rho_{BA}  W^{(k)*}_A\Big\|^\alpha_{(B:1,A:\alpha)} \\&= d_A^{-2\alpha}\Big\|\sum_{k=1}^{d^2} W^{(k)}_A\rho_{BA}  W^{(k)*}_A\Big\|^\alpha_{(B:1,A:\alpha)} \\&\leq d_A^{-2\alpha}\sum_{k=1}^{d^2}\Big\|W^{(k)}_A\rho_{BA} W^{(k) *}_A\Big\|_{(B:1,A:\alpha)}^\alpha \\ &= d_A^{2(1-\alpha)}\|\rho_{BA}\|^\alpha_{(B:1,A:\alpha)}.
\end{align}
Rearranging this yields the first inequality.
For the second inequality, we observe that by the reverse Hölder inequality
\begin{align}
    \|\rho_{BA}\|_{1}&\geq\sup_{a,b\geq0}\Big\|a_B^{-1}a_B\rho_{BA}b_Bb_B^{-1}\Big\|_{1}\\ &\geq\sup_{a,b\geq0} \Big\|a_B^{-1}\otimes\1_A\Big\|_{2\alpha^\prime}\Big\|b_B^{-1}\otimes\1_A\Big\|_{2\alpha^\prime}\Big\|a_B\rho_{BA}b_B\Big\|_\alpha \\ &=d_A^\frac{1}{\alpha^\prime}\sup_{a,b\geq0}\|a\|^{-1}_{-2\alpha^\prime}\|b\|^{-1}_{-2\alpha^\prime}\Big\|a_B\rho_{BA}b_B\Big\|_\alpha \\ &= d_A^\frac{1}{\alpha^\prime}\|\rho_{BA}\|_{(B:1,A:\alpha)},
\end{align} 
which rearranged yields the desired inequality.
\end{proof}
\noindent A direct consequence of these bounds is the following known bound on the conditional entropy for $\alpha\in[\frac{1}{2},1)$:
\begin{align}
   |H_\alpha^\uparrow(A|B)_\rho|\leq \log d_A.
\end{align}

\begin{proof}[Proof of \cref{thm:continuity.bound}]
Without loss of generality, we may assume $\|\rho-\sigma\|_1=2\epsilon$.
We denote with $\mu,\nu$ the trace normalized positive and negative parts of $\rho-\sigma$, so we have
\begin{align}
    \rho-\sigma=\epsilon(\mu-\nu).
\end{align}
Now we have that
\begin{align}
    \left|\|\rho_{BA}\|^\alpha_{(B:1,A:\alpha)}-\|\sigma_{BA}\|_{(B:1,A:\alpha)}^\alpha\right|&\leq \|\rho_{BA}-\sigma_{BA}\|^\alpha_{(B:1,A:\alpha)}\\
    &= \epsilon^\alpha\|\mu_{BA}-\nu_{BA}\|_{(B:1,A:\alpha)}^\alpha\\
    &\leq \epsilon^\alpha\big(\|\mu_{BA}\|_{(B:1,A:\alpha)}^\alpha+\|\nu_{BA}\|_{(B:1,A:\alpha)}^\alpha \big)\\
    &\leq 2\epsilon^\alpha d_A^{-\frac{\alpha}{\alpha^\prime}},
\end{align} 
where the first inequality is the reversed triangle inequality, which holds since $\|\cdot\|_{(B:1,A:\alpha)}^\alpha$ satisfies the triangle inequality, which we also used in the second inequality. The last bound follows from \Cref{lem:1.alpha.bound}. Next, 
another use of \cref{lem:1.alpha.bound} yields
\begin{align}
    \frac{\|\rho_{BA}\|_{(B:1,A:\alpha)}^\alpha}{\|\sigma_{BA}\|^\alpha_{(B:1,A:\alpha)}}\leq   \frac{\|\sigma_{BA}\|_{(B:1,A:\alpha)}^\alpha+2\epsilon^\alpha d_A^{-\frac{\alpha}{\alpha^\prime}} }{\|\sigma_{BA}\|^\alpha_{(B:1,A:\alpha)}} \leq 1+\frac{2\epsilon^\alpha d_A^{-\frac{\alpha}{\alpha^\prime}}}{\|\sigma_{BA}\|^\alpha_{(B:1,A:\alpha)}}
    \leq 1+2\epsilon^\alpha d_A^{-2\frac{\alpha}{\alpha^\prime}}
\end{align}
Taking logarithms, multiplying by $\frac{1}{1-\alpha}$ and observing the symmetry of the argument in $\rho\leftrightarrow\sigma$ yields the claim:
\begin{align}
    \left|\frac{\alpha}{1-\alpha}\log\|\rho_{BA}\|_{(B:1,A:\alpha)}-\frac{\alpha}{1-\alpha}\log\|\sigma_{BA}\|_{(B:1,A:\alpha)}\right|\leq \frac{1}{1-\alpha}\log\left(1+2\epsilon^\alpha d_A^{2(1-\alpha)}\right).
\end{align}
\end{proof}

\subsubsection{Reversed Conditional Entropy and Sandwiched Umlaut Information}
The conditional entropy optimizes the second register of the relative entropy. We can also consider the reversed quantity, where we swap the roles of $\rho_{AB}$ and $\1_A\otimes\sigma_B$. Such reversals have been studied and found operational meaning in the case of the mutual information \cite{Girardi.2025}. Here we propose to define them in terms of 2-indexed quasi-norms. 

\begin{corollary}\label{cor:reversedconditional.entropy}
Let $0<\alpha<1$ and set $\beta:=\frac{\alpha}{1-\alpha}\geq \alpha$, then it holds that
\begin{align}
    \sup_{\rho\geq0, \|\rho\|_1= 1}-D_\alpha(\1_A\otimes\rho_B\|\sigma_{AB}) &=\frac{\alpha}{1-\alpha}\log\|\sigma_{BA}^{\frac{1-\alpha}{\alpha}}\|_{(B:\beta,A:\alpha)}=\beta\log\|\sigma_{BA}^{\frac{1}{\beta}}\|_{(B:\beta,A:\alpha)}, 
\end{align}
\end{corollary}

\begin{proof}
This follows directly from \cref{equ:Def1.2} and \cref{lem:Positive.Symmetrie}. The first equality is due to 
\begin{align}
    \sup_{\rho\geq0, \|\rho\|_1= 1}-D_\alpha(\1_A\otimes\rho_B\|\sigma_{AB})&=\frac{\alpha}{1-\alpha}\log\sup_{\rho\geq0, \|\rho\|_1=1}\|\rho_B^\frac{1}{2}\sigma_{BA}^{\frac{1-\alpha}{\alpha}}\rho_B^\frac{1}{2}\|_\alpha  \\&= \frac{\alpha}{1-\alpha}\log\sup_{\rho\geq0}\|\rho\|^{-1}_1\|\rho_B^\frac{1}{2}\sigma_{BA}^{\frac{1-\alpha}{\alpha}}\rho_B^\frac{1}{2}\|_\alpha \\&= \frac{\alpha}{1-\alpha}\log\sup_{\rho\geq0}\|\rho\|^{-2}_2\|\rho_B\sigma_{BA}^{\frac{1-\alpha}{\alpha}}\rho_B\|_\alpha \\ &=\frac{\alpha}{1-\alpha}\log\|\sigma_{BA}^{\frac{1-\alpha}{\alpha}}\|_{(B:\beta,A:\alpha)}\\
    &=\beta\log\|\sigma_{BA}^{\frac{1}{\beta}}\|_{(B:\beta,A:\alpha)}
\end{align} 
\end{proof}

\begin{remark}
Note that without the prefactor and the logarithms this argument works for any $\alpha\leq 1$, however, it breaks down when $\alpha>1$ since then there cannot exist an index $\beta\geq\alpha$ such that $\frac{1}{\alpha}-\frac{1}{\beta}=1$. 
This is related to the fact that for $\alpha\geq1$, $D_\alpha(\1_A\otimes\rho_B\|\sigma_{AB})$ will diverge if $\sigma_{AB}$ has non-trivial kernel on $A$, i.e. if $L_\sigma=R_\sigma\neq \1_A$. 
\end{remark}
\noindent In the framework of mutual informations in \cite{Girardi.2025}, the quantum \emph{umlaut information} is defined as
\begin{align}
U(A;B)_\rho:=\inf_{\sigma_B\in\cD(\cH_B)}D(\rho_A\otimes\sigma_B\|\rho_{AB}).
\end{align} In that work the authors consider also a Petz-Rényi based Umlaut information and give it an operational interpretation in the context of asymmetric hypothesis testing. In the following we consider the Sandwiched Umlaut information for $\alpha<1$:
\begin{align}
\tilde{U}_\alpha(A;B)_\rho&:=\inf_{\sigma_B\geq0, \|\sigma\|_1= 1}D_\alpha(\rho_A\otimes\sigma_B\|\rho_{AB}).
\end{align}
In the next Lemma, we show that the latter can also be expressed as a 2-indexed quasi-norm, which allows for simple proofs of important properties.

\begin{lemma}
Let $0<\alpha<1$ and set $\beta:=\frac{\alpha}{1-\alpha}\geq \alpha$, i.e. s.t. $\frac{1}{\alpha}-\frac{1}{\beta}=1$, then it holds that
\begin{align}
 \tilde{U}_\alpha(A;B)_\rho=\frac{\alpha}{\alpha-1}\log\|\rho_A^\frac{1}{2}\rho_{BA}^{\frac{1-\alpha}{\alpha}}\rho_A^\frac{1}{2}\|_{(B:\beta,A:\alpha)}=-\beta\log\|\rho_A^\frac{1}{2}\rho_{BA}^{\frac{1}{\beta}}\rho_A^\frac{1}{2}\|_{(B:\beta,A:\alpha)}.
\end{align}
\end{lemma}
\begin{proof}
This expression follows completely analogously to reverse conditional entropy of \cref{cor:reversedconditional.entropy}.
\end{proof}

\noindent  As for the Petz-Rényi version in \cite{Girardi.2025}, taking the limit $\alpha\to1$ recovers the Umlaut information and so we get a novel expression for it.

\begin{lemma}
Let $\rho_{AB}\in\cD(\cH_A\otimes\cH_B)$ be a full rank state, then it holds that
    \begin{align}
    U(A;B)_\rho&=\lim_{\alpha\uparrow 1}\tilde{U}_\alpha(A;B)_\rho= \lim_{\alpha\uparrow1}\inf_{\sigma_B\in\cD(\cH_B)}D_\alpha(\rho_A\otimes\sigma_B\|\rho_{AB}) = -\frac{d}{d\gamma}\Bigg|_{\gamma=0}\!\!\!\!\log\|\rho_A^\frac{1}{2}\rho_{BA}^\gamma\rho_A^\frac{1}{2}\|_{(\frac{1}{\gamma},\frac{1}{1+\gamma})}.
    \end{align}
\end{lemma}

\begin{proof}
To argue for the switching of the limit and infimum in the first equality we use the Mosonyi-Hiai minimax theorem~\cite[Corollary A.2]{Mosony.2011} which applies since $\lim_{\alpha \uparrow 1}$ can be replaced by $\sup_{\alpha\in{(0,1)}}$ and the well known facts that $D_\alpha(\rho_A\otimes\sigma_B\|\rho_{AB})$ is monotonically increasing in $\alpha$ on $(0,1)$ and convex and continuous in $\sigma_B$ on the compact set $\cD(\cH_B)$. That is
\begin{align}
    U(A;B)_\rho&= \inf_{\sigma_B\in\cD(\cH_B)}\sup_{\alpha\in(0,1)}D_\alpha(\rho_A\otimes\sigma_B\|\rho_{AB}) = \sup_{\alpha\in(0,1)}\inf_{\sigma_B\in\cD(\cH_B)}D_\alpha(\rho_A\otimes\sigma_B\|\rho_{AB}) \\&= \lim_{\alpha\uparrow 1}\tilde{U}_\alpha(A;B)_\rho.
\end{align}
To derive the second equality we see that
\begin{align}
\lim_{\alpha\uparrow1}\inf_{\sigma_B}D_\alpha(\rho_A\otimes\sigma_B\|\rho_{AB}) &= \lim_{\alpha\uparrow1}-\beta\log\|\rho_A^\frac{1}{2}\rho_{BA}^\frac{1}{\beta}\rho_A^\frac{1}{2}\|_{(B:\beta,A:\alpha)} \\
    &=\lim_{\beta\uparrow\infty}-\beta\log\|\rho_A^\frac{1}{2}\rho_{BA}^\frac{1}{\beta}\rho_A^\frac{1}{2}\|_{(B:\beta,A:\frac{\beta}{1+\beta})} \\ &=
    \lim_{\gamma\downarrow0}-\frac{1}{\gamma}\log\|\rho_A^\frac{1}{2}\rho_{BA}^\gamma\rho_A^\frac{1}{2}\|_{(B:\frac{1}{\gamma},A:\frac{1}{1+\gamma})} \\
    &= -\frac{d}{d\gamma}\Bigg|_{\gamma=0}\log\|\rho_A^\frac{1}{2}\rho_{BA}^\gamma\rho_A^\frac{1}{2}\|_{(B:\frac{1}{\gamma},A:\frac{1}{1+\gamma})},
\end{align} where the last equality follows since
\begin{align}
    F(\gamma):=\log\|\rho_A^\frac{1}{2}\rho_{BA}^\gamma\rho_A^\frac{1}{2}\|_{(B:\frac{1}{\gamma},A:\frac{1}{1+\gamma})}
\end{align} satisfies $F(0)= \log\|\1_B\otimes\rho_A\|_{(\infty,1)}=\log(1)=0$ when $\rho_{AB}$ is full rank.
\end{proof}
\begin{remark}
Note that when $\rho_{AB}$ is s.t. $\|\rho_A^\frac{1}{2}\Pi_{\rho_{BA}}\rho_A^\frac{1}{2}\|_{(B:\infty,A:1)}<1$ then $F(0)<0$ in the above proof and the function has a non-zero jump there implying that $U(A;B)_\rho=\infty$. This condition is thus a necessary and sufficient condition for $U(A;B)_\rho=\infty$ (in finite dimensions) and it is easy to see that it is equivalent to the one of \cite[Lemma 2 (2)]{Girardi.2025}: $\exists \psi_B$ s.t. $\supp(\rho_A\otimes\psi_B)\subset\supp(\rho_{AB})$.
\end{remark}

\begin{proposition}
Let $0<\alpha<1$, then $\tilde{U}_\alpha(A;B)_\rho$ satisfies
\begin{enumerate}
     \item Invariance under local unitary channels $\mathcal{U}_A, \mathcal{U_B}$, i.e.
    \begin{align}
        \tilde{U}_\alpha(A;B)_{(\mathcal{U}_A\otimes\mathcal{U}_B)(\rho)}=  \tilde{U}_\alpha(A;B)_\rho
    \end{align}
    \item Simplification on states classical on $B$, i.e. for $\rho_{AB}=\sum_ip_i\rho^{(i)}_A\otimes|i\rangle\langle i|_B$,
    \begin{align}
        \tilde{U}_\alpha(A;B)_\rho= -\log\left(\sum_ip_i\|\rho_A^\frac{1}{2}\rho^{(i)\,\frac{1}{\beta}}_A\rho_A^\frac{1}{2}\|^{\frac{\alpha}{1-\alpha}}_\alpha\right)
    \end{align} 
\end{enumerate}
\end{proposition}

\begin{proof}
1. The invariance under $\mathcal{U}_B$ is clear by local unitary invariance of the $\|\cdot\|_{(\beta,\alpha)}$-quasi norm. For $\mathcal{U}_A$ see that
\begin{align}
    \|\mathcal{U}_A(\rho_A)^\frac{1}{2}\mathcal{U}_A(\rho_{BA})^\frac{1}{\beta}\mathcal{U}_A(\rho_A)^\frac{1}{2}\|_{(B:\beta,A:\alpha)} &= \|U_A\rho_A^\frac{1}{2}U_A^*U_A\rho_{BA}^\frac{1}{\beta}U_A^*U_A\rho_A^\frac{1}{2}U_A^*\|_{(B:\beta,A:\alpha)}\\
    &=\|\rho_A^\frac{1}{2}\rho_{BA}^\frac{1}{\beta}\rho_A^\frac{1}{2}\|_{(B:\beta,A:\alpha)},
\end{align}
where we used again the local isometric invariance and omitted writing identity operators. \\
2. This follows directly from \cref{thm:cq.additivity}. Writing $\tr_B\rho_{AB}=\sum_ip_i\rho_A^{(i)}\equiv \rho_A$, it holds that
\begin{align}
\Big\|\rho_A^\frac{1}{2}(\sum_i|i\rangle\langle i|_B\otimes  p_i\rho^{(i)}_{A})^\frac{1}{\beta}\rho_A^\frac{1}{2}\Big\|_{(B:\beta,A:\alpha)} &= \Big\|\sum_i|i\rangle\langle i|_B\otimes (p_i^\frac{1}{\beta}\rho_A^\frac{1}{2}(\rho^{(i)}_{A})^\frac{1}{\beta}\rho_A^\frac{1}{2})\Big\|_{(B:\beta,A:\alpha)} \\
&=\left(\sum_ip_i\|\rho_A^\frac{1}{2}(\rho^{(i)}_A)^\frac{1}{\beta}\rho_A^\frac{1}{2}\|^  \beta_\alpha\right)^\frac{1}{\beta}.
\end{align}
Applying $-\beta\log$ to this completes the proof.
\end{proof}
\noindent Another important property which we can prove is additivity under tensor product states. However, because that proof requires more advanced machinery that we believe deserves its own introduction, we leave it for future work \cite{Kochanowski.in.prep} to present.
\begin{proposition}[\cite{Kochanowski.in.prep}]
    Additivity under tensor products: For $\rho=\rho_1\otimes\rho_2\in\cD(\cH_A\otimes\cH_B)$ with $A=A_1A_2, B=B_1B_2$, $\rho_1:=\rho_{A_1B_1}$ and $\rho_2:=\rho_{A_2B_2}$, it holds that
    \begin{align}
         \tilde{U}_\alpha(A;B)_{\rho_1\otimes\rho_2}=  \tilde{U}_\alpha(A_1;B_1)_{\rho_1}+ \tilde{U}_\alpha(A_2;B_2)_{\rho_2}.
    \end{align}
\end{proposition}

\subsection{Completely Bounded Quasi-Norms and Co-Quasi-Norms}\label{sec:cb.mixed.norms}

Completely bounded norms \cite{book:Paulsen.2003} find many applications in quantum information theory. Most famously perhaps in the form of the diamond norm \cite[Chapter 3]{book:Watrous.2018}. When the input and output spaces are (multi-index-)Schatten spaces, applications include channel entropy uncertainty relations \cite{Gao.2018, Gao.2023}, completely bounded hypercontractivity \cite{Bardet.2022, Bardet.2024,Beigi.2016} and complete modified logarithmic Sobolev constants \cite{gao2022complete}, completely bounded minimal output entropies \cite{Devetak.2006}, and minimal output conditional entropies under arbitrary environments \cite{Fawzi.2026}.
In this section we extend the notion of completely bounded norms to quasi-normed spaces and consider completely bounded quasi-norms and co-quasi-norms to show among others multiplicativity results in \cref{sec:app:mixed.norm.additivity} and \cref{sec:app:mixed.conorm.additivity} in analogy to \cite[Section 4]{Devetak.2006}.

\medskip 

 We will make use of the following standard definitions: given any two (quasi-)normed spaces $\mathcal{X},\mathcal{Y}$ with associated (quasi-)norms labeled as $\|\cdot \|_{x}$ and $\|\cdot \|_{y}$, and a bounded map $\Phi:\mathcal{X}\to\mathcal{Y}$:
\begin{align}
\|\Phi\|_{x\to y}:=\sup_{X\in\mathcal{X}}\frac{\|\Phi(X)\|_y}{\|X\|_x}\qquad \text{ and }\qquad
\|\Phi\|^+_{x\to y}:=\sup_{\underset{X\geq0}{X\in\mathcal{X}}}\frac{\|\Phi(X)\|_y}{\|X\|_x}\,~\,.
\end{align}
For instance $\|\cdot\|_x:=\|\cdot \|_{(Q:q,P:p)}$ in what follows. 
Note that these functional are, by definition, also quasi-norms on the space of all bounded (and cone of CP) linear maps. We hence sometimes call them `mixed' quasi-norms, referring to the fact that the in- and output Spaces possess different multi-indexed Schatten (quasi-)norms.
Similarly, we define the co-(quasi-)norms of $\Phi$ as 
\begin{align}
\|\Phi\|^+_{\co,x\to y}:=\inf_{\underset{X\geq0}{X\in\mathcal{X}}}\frac{\|\Phi(X)\|_y}{\|X\|_x}\,~\,
\end{align} and refer to them analogously as `mixed' co-quasi-norms.

\subsubsection{Completely Bounded Quasi-Norms of Linear Maps}

Inspired by the completely bounded norm between operator spaces and its rewriting via Operator valued Schatten norms, we define the concept of complete boundedness between Schatten-quasi-normed spaces in the following.

\begin{definition}\label{def:cb.quasi.norm}
Let $q,p\in[\frac{1}{2},\infty]$ and $\Phi:Q\to P$ be a linear map, then
\begin{align}
    &\|\Phi\|_{\cb,(Q:q)\to (P:p)}:=\sup_E\|\id_E\otimes\Phi\|_{(E:1,Q:q)\to (E:1,P:p)}\\
    &\|\Phi\|^+_{\cb,(Q:q)\to (P:p)}:=\sup_E\|\id_E\otimes\Phi\|^+_{(E:1,Q:q)\to (E:1,P:p)}\,,
\end{align}
where the supremum is over any finite dimensional reference system $E$.
\end{definition}

\noindent When $q,p\geq1$ this recovers the usual completely bounded norm  \cite[Lemma 1.7]{Book.Pisier.1998}. We restrict here to values $q,p\geq\frac{1}{2}$, however, in principle one could define a notion of completely bounded quasi-norm analogously by choosing the index on the reference systems to be some $t$, which is compatible in the sense of \cref{def:Definition1} with both $q$ and $p$. This works as long as $\big|\frac{1}{q}-\frac{1}{p}\big|\leq 2$ and is well defined due to the following \cref{lem:cb.norm.via.t}. In the normed setting, i.e. where $q,p,t\geq1$ this is a well known statement going back to Pisier \cite{Book.Pisier.1998}.

\begin{lemma}\label{lem:cb.norm.via.t}
Let $\Phi:Q\to P$ be a linear map, and $q,p,t>0$ such that $\max\big\{\big|\frac{1}{t}-\frac{1}{q}\big|,\big|\frac{1}{t}-\frac{1}{p}\big|\big\}\leq1$, then
\begin{align}
    &\|\Phi\|^{+}_{\cb,(Q:q)\to (P:p)}=\sup_E\|\id_E\otimes\Phi\|^{+}_{(E:t,Q:q)\to (E:t,P:p)}\,,\\
    &\|\Phi\|_{\cb,(Q:q)\to (P:p)}=\sup_E\|\id_E\otimes\Phi\|_{(E:t,Q:q)\to (E:t,P:p)}.
\end{align} 
\end{lemma}

\begin{proof}
This equivalence is a simple consequence of the factorization formulas from \cref{thm:qt.to.pt}. We show the equality by separately proving the two inequalities contained in it. In order not to separately consider the cases $t\leq 1$ and $t\geq1$, we will in the following work with $s\leq t$, such that both are assumed to be such that the formulas of \cref{thm:qt.to.pt} apply. Then later in the proof taking either $s=1$ covers the proof of the case $t\geq1$ in the statement of the lemma while taking $t=1$ in the proof covers the case $t\leq 1$: Let $\frac{1}{r}=\frac{1}{s}-\frac{1}{t}\leq 1$.
Then on the one hand we have, for any input $X_{EQ}$ and any $\epsilon>0$, that there exist approximately optimal $\hat{a}_E,\hat{b}_E,\hat{Z}_{EQ}$ s.t. $X_{EQ}=\hat{a}_E\hat{Z}_{EQ}\hat{b}_E$ and
\begin{align}
    \|X\|_{(E:s,Q:q)}\geq \|\hat{a}\|_{2r}\|\hat{b}\|_{2r}\|\hat{Z}\|_{(E:t,Q:q)}-\epsilon=\|\hat{Z}\|_{(E:t,Q:q)}-\epsilon,
\end{align}
as we may assume $\|\hat{a}\|_{2r}=\|\hat{b}\|_{2r}=1$. At the same time
\begin{align}
    \|(\id_E\otimes\Phi)(X_{EQ})\|_{(E:s,P:p)}\leq \|a\|_{2r}\|b\|_{2r}\|(\id_E\otimes\Phi)(Z_{EQ})\|_{(E:t,P:p)},
\end{align} for any $a,b,Z$ with $X_{EQ}=a_EZ_{EQ}b_E$, since this yields $(\id_E\otimes\Phi)(X_{EQ})=a_E(\id_E\otimes\Phi)(Z_{EQ})b_E$. Together these yield
\begin{align}
    \frac{\|(\id_E\otimes\Phi)(X)\|_{(E:s,P:p)}}{\|X\|_{(E:s,Q:q)}}\leq \frac{\|\hat{a}\|_{2r}\|\hat{b}\|_{2r}\|(\id_E\otimes\Phi)(\hat{Z})\|_{(E:t,P:p)}}{\|\hat{a}\|_{2r}\|\hat{b}\|_{2r}\|\hat{Z}\|_{(E:t,Q:q)}-\epsilon}= \frac{\|(\id_E\otimes\Phi)(\hat{Z})\|_{(E:t,P:p)}}{\|\hat{Z}\|_{(E:t,Q:q)}-\epsilon}.
\end{align} Since this works for any $\epsilon>0$ it follows that
\begin{align}
    \frac{\|(\id_E\otimes\Phi)(X)\|_{(E:s,P:p)}}{\|X\|_{(E:s,Q:q)}}\leq \sup_{Z_{EQ}}\frac{\|(\id_E\otimes\Phi)(Z)\|_{(E:t,P:p)}}{\|Z\|_{(E:t,Q:q)}}
\end{align}
where the last inequality follows since this works for any $\epsilon>0$. 
On the other hand we similarly get that for any $\epsilon>0$ approximate  $\hat{a},\hat{b}$ such that 
\begin{align}
   \|(\id_E\otimes\Phi)(X_{EQ})\|_{(E:t,P:p)}\leq \|\hat{a}\|^{-1}_{2r}\|\hat{b}\|^{-1}_{2r}\|(\id_E\otimes\Phi)(\hat{a}_EX_{EQ}\hat{b}_E)\|_{(E:s,P:p)}+\epsilon, 
\end{align} where again we may assume $\|\hat{a}\|^{-1}_{2r}=\|\hat{b}\|^{-1}_{2r}=1$. Analogous to before it also holds that
\begin{align}
    \|X_{EQ}\|_{(E:t,Q:q)}\geq \|\hat{a}\|^{-1}_{2r}\|\hat{b}\|^{-1}_{2r}\|\hat{a}_EX_{EQ}\hat{b}_E\|_{(E:s,Q:q)}
\end{align} and hence
\begin{align}
  \frac{\|(\id_E\otimes\Phi)(X)\|_{(E:t,P:p)}}{\|X\|_{(E:t,Q:q)}}&\leq \frac{\|\hat{a}\|^{-1}_{2r}\|\hat{b}\|^{-1}_{2r}\|(\id_E\otimes\Phi)((\hat{a}_EX_{EQ}\hat{b}_E)\|_{(E:s,P:p)}+\epsilon}{\|\hat{a}\|^{-1}_{2r}\|\hat{b}\|^{-1}_{2r}\|\hat{a}_EX_{EQ}\hat{b}_E\|_{(E:s,Q:q)}} \\&= \frac{\|(\id\otimes\Phi)(\hat{a}_EX_{EQ}\hat{b}_E)\|_{(E:s,P:p)}+\epsilon}{\|\hat{a}_EX_{EQ}\hat{b}_E\|_{(E:s,Q:q)}}.
\end{align} Since this works for any $\epsilon>0$ this implies that
\begin{align}
  \frac{\|(\id_E\otimes\Phi)(X)\|_{(E:t,P:p)}}{\|X\|_{(E:t,Q:q)}}  \leq\sup_Z\frac{\|(\id\otimes\Phi)(Z)\|_{(E:s,P:p)}}{\|Z\|_{(E:s,Q:q)}}.
\end{align}
Taking suprema over systems $E$, the dimension of the ancilla systems then completes the proof. For the positive case the exact same proof applies, since by \cref{lem:Positive.Symmetrie} and \cref{thm:qt.to.pt} if $X\geq0$ we may take $a=b\geq0$ and hence $Z\geq0$ as well.
\end{proof}

\subsubsection{Completely Bounded Co-Quasi-Norms of CP Maps}

 Instead of considering the usual quasi-norm of a linear map $\Phi:\cX\to\cY$ between two quasi-normed spaces, which are given as maximizations, one may be interested in analogous quantities involving a minimization. This is sometimes called the \textit{minimum modulus} or \textit{co(-quasi-)norm} \cite{Mbekhta.1996, amelin.1973}. Next, we define the \textit{co-quasi-norm} and \textit{completely bounded co-quasi-norm} of a CP maps. This quantity is much less well behaved that the former since it if often just $0$, which is why we will be focusing on the version restricted to positive inputs and CP maps. 
\begin{definition}
Let $\Phi:Q\to P$ be a linear map and let $q,p\in[\frac{1}{2},\infty]$ then we define     
\begin{align}
\|\Phi\|^+_{\cb,\co, (Q:q)\to (P:p)}&:= \inf_{E}\|\id_E\otimes\Phi\|^+_{\co,(E:1,Q:q)\to(E:1,P:p)}.
\end{align}
\end{definition}
\noindent In exactly the same way as for the (completely bounded) quasi-norm, we can show that for the co-quasi-norm the reference system index can be chosen arbitrarily, as long as it is compatible with both $0<q,p$. In particular this means that the above quantity is well defined: 
\begin{lemma}\label{lem:cb.co.norm.via.t}
Let $\Phi:Q\to P$ be a linear map, $q,p,t>0$ such that $\max\big\{\big|\frac{1}{t}-\frac{1}{q}\big|,\big|\frac{1}{t}-\frac{1}{p}\big|\big\}\leq1$, then
\begin{align}
    \|\Phi\|^{+}_{\cb,\co,(Q:q)\to(P: p)}=\inf_E\|\id_E\otimes\Phi\|^{+}_{\co,(E:t,Q:q)\to (E:t,P:p)}
\end{align} 
is independent of the index $t$ placed on the reference system, as long as it is compatible with the in- and output indices. This also allows to define this quantity for arbitrary $q,p>0$ for which such a $t$ exists.
\end{lemma}
\begin{proof}
    The proof follows exactly as the proof of \cref{lem:cb.norm.via.t}, except that we take infima instead of suprema at the end.
\end{proof}

\subsection{Non-Commutative Minkowski Inequality on Quasi-Normed Spaces}\label{sec:app:Minkowski.Inequality}

\noindent Having a notion of completely bounded quasi-norm and co-quasi-norm, we may extend a seminal result concerning the swapping of adjacent systems to the quasi-normed setting. For this, we need to introduce certain 3-indexed Schatten quasi-norms to show that the swap operator for two adjacent systems is a complete contraction (see \cref{thm:QuasiNormMinkowskiInequality} below). We introduce these 3-indexed quasi-norms inspired by both the normed setting \cite{Devetak.2006} and \cref{def:Definition1} focusing only on the required cases and  leave it to future work to present a complete theory of multi-indexed Schatten (quasi-)norms.

\begin{definition}
For $0<q\leq p$ with $\frac{1}{r}:=\frac{1}{q}-\frac{1}{p}\leq 1$ and an operator $X$ on a tripartite system $PQR$,
\begin{align}\label{def:3.indexed.special}
\|X\|_{(P:p,Q:q,R:p)}:=\sup_{a,b\geq0}\|a\|^{-1}_{2r}\|b\|^{-1}_{2r}\|a_PX_{PQR}b_P\|_{(PQ:q,R:p)}.
\end{align}
\end{definition}
\noindent Note that one can easily show, in an extension of the proof of \cref{lem:Positive.Symmetrie}, that when $X_{PQR}\geq0$ it suffices to optimize here also only over $a_P=b_P$. See \cref{prop:3index.sym.pos} in the appendix for this. 

Next, we consider the swap map between two systems $\SWAP:QP\to PQ$, which is defined as
\begin{align}
\SWAP_{}(X_Q\otimes Y_P)=Y_{P}\otimes X_{Q}\,.
\end{align}

\begin{theorem}[A Minkowski Inequality for quasi-normed spaces]\label{thm:QuasiNormMinkowskiInequality}
For any three systems $E,Q,P$ and indices $0<q\leq p$ s.t. $\big|\frac{1}{q}-\frac{1}{p}\big|\leq1$, 
\begin{align}
\|\id_E\otimes \SWAP_{QP}\|_{(E:p,Q:q,P:p)\to (EP:p,Q:q)}\le 1 \text{ and } \|\id_E\otimes \SWAP_{QP}\|^+_{\co,{(EP:p,Q:q)\to (E:p,Q:q,P:p)}}\ge 1\,.
\end{align}
\end{theorem}

 \noindent To prove  \cref{thm:QuasiNormMinkowskiInequality}, we first establish the following factorization result.
\begin{lemma}\label{lem:3.index.formulas}
Let $0<q\leq p$ be s.t. $\frac{1}{r}:=\frac{1}{q}-\frac{1}{p}\leq1$, then for any $X$ on the tripartite system $PQR$ it holds that
\begin{align}\label{equ:pqp.to.ppp}
\|X\|_{(P:p,Q:q,R:p)}&=\inf_{A_{PQR},B_{PQR}}\|AA^*\|_{(P:\infty,Q:r,R:\infty)}^\frac{1}{2}\|B^*B\|_{(P:\infty,Q:r,R:\infty)}^\frac{1}{2}\|A^{-1}XB^{-1}\|_{p} \\ 
    &= \inf_{A_{PQ},B_{PQ}}\|AA^*\|_{(P:\infty,Q:r)}^\frac{1}{2}\|B^*B\|_{(P:\infty,Q:r)}^\frac{1}{2}\|A_{PQ}^{-1}X_{PQR}B_{PQ}^{-1}\|_{p} \nonumber
\end{align} and for any $X$ on the bipartite system $PQ$:
\begin{align}\label{equ:pq.to.pp}
\|X\|_{(P:p,Q:q)}=\inf_{A,B}\|A_{PQ}A_{PQ}^*\|_{(P:\infty,Q:r)}^\frac{1}{2}\|B_{PQ}^*B_{PQ}\|_{(P:\infty,Q:r)}^\frac{1}{2}\|A^{-1}XB^{-1}\|_{p}.
\end{align}
\end{lemma}

\begin{proof}[Proof of \cref{lem:3.index.formulas}]
Note that \eqref{equ:pq.to.pp} follows directly from \eqref{equ:pqp.to.ppp} by taking the last system $R$ to be 1-dimensional.
 We first focus on the first equation in \eqref{equ:pqp.to.ppp}: 
\begin{align}
\|X\|_{(P:p,Q:q,R:p)}&\overset{(1)}{=}\sup_{a,b\geq0}\inf_{\underset{}{c,d\geq0}}\|a\|^{-1}_{2r}\|b\|^{-1}_{2r}\|c\|_{2r}\|d\|_{2r}\|c^{-1}_{PQ}a_PX_{PQR}b_Pd^{-1}_{PQ}\|_{p} \\  
&\overset{(2)}{=}\sup_{a,b\geq0}\inf_{c,d\geq0}\sup_{e,f\geq0}\inf_{g,h\geq0}\|a\|_{2r}^{-1}\|b\|_{2r}^{-1}\|c\|_{2r}\|d\|_{2r}\|e\|_{2r}^{-1}\|f\|_{2r}^{-1}\|g\|_{2r}\|h\|_{2r}\\ &\qquad\qquad\qquad\qquad\qquad\qquad\qquad\times\|g^{-1}_{PQR}e_{PQR}c^{-1}_{PQ}a_PX_{PQR}b_Pd_{PQ}^{-1}f_{PQR}h_{PQR}^{-1}\|_p \\ 
&\overset{(3)}{=}\sup_{a,b\geq0}\inf_{c,d\geq0}\sup_{e,f\geq0}\inf_{\tilde{A},\tilde{B}\geq0}\|a\|_{2r}^{-1}\|b\|_{2r}^{-1}\|c\|_{2r}\|d\|_{2r}\|e\|_{2r}^{-1}\|f\|_{2r}^{-1}\|\tilde{A}^{-\frac{1}{2}}_{PQR}X_{PQR}\tilde{B}^{-\frac{1}{2}}_{PQR}\|_p\\ &\qquad \qquad \times \|e_{PQR}c_{PQ}^{-1}a_P\tilde{A}_{PQR}a_Pc_{PQ}^{-1}e_{PQR}\|^\frac{1}{2}_{r}\|f_{PQR}d_{PQ}^{-1}b_P\tilde{B}_{PQR}b_Pd_{PQ}^{-1}f_{PQR}\|^\frac{1}{2}_{r}.
\end{align}
Here (1) follows from definition \cref{def:3.indexed.special} and \cref{lem:SupportConditions} relating the $\|\cdot\|_{(q,p)}$ to the $\|\cdot\|_{p}$ quasi-norm; (2) holds by \cref{prop:SpVariationalFormulas} relating the $p$-quasi-norm to the $q$-one via $\sup_{e,f}$ and back to the $p$-one via $\inf_{g,h}$; in (3) we defined $\tilde{A}:=AA^*$, where $A=a_P^{-1}c_{PQ}e_{PQR}^{-1}g_{PQR}$ and likewise $\tilde{B}:=B^*B$, where $B=h_{PQR}f_{PQR}^{-1}d_{PQ}b_P^{-1}$ and used local unitary invariance. 

\medskip

\noindent Next, by \cref{lem:SupportConditions} together with \eqref{equ:full.to.marginal.support}, the supports of $\tilde{A}, \tilde{B}$ are well defined and equal to the image, respectively, support of $X$.
Now by exactly the same argument as in the proof of \eqref{equ:InvDef1.2} of \cref{lem:inverseDef1}, up to the additional appearance of the operators $a_P,b_P,c_{PQ}^{-1},d_{PQ}^{-1}$, which, however, do not play a role the concavity of the displayed functions in $e,f$, nor their continuity, we can apply Sion's minimax theorem to switch the infimum over $\tilde{A},\tilde{B}$ with the supremum over $e,f$. The former also naturally commutes with the infimum over $c,d$, so that
\begin{align}
   \|X\|_{(P:p,Q:q,R:p)}&= \sup_{a,b\geq0}\inf_{\tilde{A},\tilde{B}\geq0}\inf_{c,d\geq0}\sup_{e,f\geq0}\|a\|_{2r}^{-1}\|b\|_{2r}^{-1}\|c\|_{2r}\|d\|_{2r}\|e\|_{2r}^{-1}\|f\|_{2r}^{-1}\|\tilde{A}^{-\frac{1}{2}}X\tilde{B}^{-\frac{1}{2}}\|_p\\ &\quad \qquad \qquad \times \|e_{PQR}c_{PQ}^{-1}a_P\tilde{A}_{PQR}a_Pc_{PQ}^{-1}e_{PQR}\|^\frac{1}{2}_{r}\|f_{PQR}d_{PQ}^{-1}b_P\tilde{B}_{PQR}b_Pd_{PQ}^{-1}f_{PQR}\|^\frac{1}{2}_{r}\\ 
   &\overset{(1)}{=} \!\!\!\!\!\!\sup_{\underset{\|a\|_1,\|b\|_1\leq 1}{a,b\geq0}}\!\!\!\inf_{\tilde{A},\tilde{B}}\!\|a_P^\frac{1}{2r}\tilde{A}_{PQR}a_P^\frac{1}{2r}\|^{\frac{1}{2}}_{(PQ:r,R:\infty)}\|b_P^\frac{1}{2r}\tilde{B}_{PQR}b_P^\frac{1}{2r}\|^{\frac{1}{2}}_{(PQ:r,R:\infty)} \|\tilde{A}^{-\frac{1}{2}}X\tilde{B}^{-\frac{1}{2}}\|_p  \\
   &= \bigg\{\!\sup_{\underset{\|a\|_1,\|b\|_1\leq 1}{a,b\geq0}}\!\!\inf_{\tilde{A},\tilde{B}}\frac{p}{2(1+p)}\bigg(\|a_P^\frac{1}{2r}\tilde{A}_{PQR}a_P^\frac{1}{2r}\|_{(PQ:r,R:\infty)}+\|b_P^\frac{1}{2r}\tilde{B}_{PQR}b_P^\frac{1}{2r}\|_{(PQ:r,R:\infty)}\bigg)\\ &\qquad\qquad\qquad\qquad\qquad\qquad\qquad\qquad\qquad\qquad\qquad+\frac{1}{(1+p)}\|\tilde{A}^{-\frac{1}{2}}X\tilde{B}^{-\frac{1}{2}}\|^p_p \bigg\}^{\frac{1+p}{p}},
\end{align} 
where (1) follows from successive uses of \cref{prop:SpVariationalFormulas} and \cref{lem:Positive.Symmetrie},
whereas the last equality follows from \cref{lem:Inf.Prod.to.Sum}. Now, the desired concavity, upper semi-continuity in $a,b$ and convexity and lower-semi-continuity in $\tilde{A},\tilde{B}$ of the first two expression follows directly from \cref{lem.app.op.concavity} and for the last summand from \cite[Theorem 1.1 (2)]{Zhang.2020}. Lower semicontinuity in $\tilde{A}, \tilde{B}\geq0$ with $\Pi_{\tilde{A}}=L_X, \Pi_{\tilde{B}}=R_X$ follows with the help of \cref{lem:SupportConditions}. These conditions are now sufficient to apply Sion's minimax \cref{thm:Sion}. This yields 
\begin{align}
    \|X\|_{(P:p,Q:q,R:p)} &=\!\bigg\{\inf_{\tilde{A},\tilde{B}}\!\!\!\!\sup_{\underset{\|a\|_1,\|b\|_1\leq 1}{a,b\geq0}}\!\!\!\frac{p}{2(1+p)}\bigg(\|a_P^\frac{1}{2r}\tilde{A}_{PQR}a_P^\frac{1}{2r}\|_{(PQ:r,R:\infty)}+\|b_P^\frac{1}{2r}\tilde{B}_{PQR}b_P^\frac{1}{2r}\|_{(PQ:r,R:\infty)}\bigg)\\ &\qquad\qquad\qquad\qquad\qquad\qquad\qquad\qquad\qquad\qquad\qquad+\frac{1}{(1+p)}\|\tilde{A}^{-\frac{1}{2}}X\tilde{B}^{-\frac{1}{2}}\|^p_p \bigg\}^{{\frac{1+p}{p}}}\\
    &\overset{(1)}{=}\!\!\bigg\{\!\inf_{\tilde{A},\tilde{B}}\!\!\frac{p}{2(1+p)}\!\bigg(\!\|\tilde{A}\|_{(P:\infty,Q:r,R:\infty)}\!\!+\!\|\tilde{B}\|_{(P:\infty,Q:r,R:\infty)}\!\!\bigg)\!\!+\!\frac{1}{(1+p)}\|\tilde{A}^{-\frac{1}{2}}X\tilde{B}^{-\frac{1}{2}}\|^p_p \bigg\}^{{\frac{1+p}{p}}}\\ &=\inf_{A,B}\|AA^*\|^\frac{1}{2}_{(P:\infty,Q:r,R:\infty)}\|B^*B\|^\frac{1}{2}_{(P:\infty,Q:r,R:\infty)}\|A^{-1}X B^{-1}\|_p.
\end{align} 
where (1) follows from \eqref{def:3.indexed.special} and the remark below it, and the last line by the same application of \cref{lem:Inf.Prod.to.Sum} as above. This finishes the proof of the first part of \cref{equ:pqp.to.ppp}. The derivation of \eqref{equ:pq.to.pp} follows exactly as above, but without the need for the optimizations over $a,b,c,d$ and hence the use of the second Sions minimax theorem. We will hence not explicitly spell it out.
\end{proof}

\noindent Having established this lemma we can now use it to prove the Minkowski inequality in the quasi-normed setting. We will use the former to relate to the normed setting and then apply the Minkowski inequality there \cite[Corollary 1.10]{Book.Pisier.1998}. 

\begin{proof}[Proof of \cref{thm:QuasiNormMinkowskiInequality}]
Fix indices $0<q\leq p$, s.t. $\frac{1}{r}=\frac{1}{q}-\frac{1}{p}\leq 1$. By \cref{lem:cb.norm.via.t} and \cref{lem:cb.co.norm.via.t} it suffices to prove that
\begin{align}
    \|X_{EPQ}\|_{(EP:p,Q:q)}\leq \|X_{EQP}\|_{(E:p,Q:q,P:p)}.
\end{align} This is done by using \cref{lem:3.index.formulas} twice as follows.
\begin{align}
    \|X_{EPQ}\|_{(EP:p,Q:q)} &= \inf_{A_{EPQ},B_{EPQ}}\|A_{EPQ}A_{EPQ}^*\|^\frac{1}{2}_{(EP:\infty,Q:r)}\|B_{EPQ}^*B_{EPQ}\|^\frac{1}{2}_{(EP:\infty,Q:r)}\\ &\qquad\qquad\qquad\times\|A_{EPQ}^{-1}X_{EPQ}B_{EPQ}^{-1}\|_{(EP:p,Q:p)} \\
    &\leq \inf_{A_{EQP},B_{EQP}}\|A_{EQP}A_{EQP}^*\|^\frac{1}{2}_{(E:\infty,Q:r,P:\infty)}\|B_{EQP}^*B_{EQP}\|^\frac{1}{2}_{(E:\infty,Q:r,P:\infty)}\\ &\qquad\qquad\qquad\times\|A_{EQP}^{-1}X_{EQP}B_{EQP}^{-1}\|_{(E:p,Q:p,P:p)} \\ 
    &= \|X_{EQP}\|_{(E:p,Q:q,P:p)}.
\end{align}
Here the inequality is the complete contractiveness of the $\SWAP$ from 
\cite[Corollary 1.10]{Book.Pisier.1998}, i.e. the noncommutative Minkowski inequality for Pisier norms. It applies since $\infty\geq r\geq 1$ by assumption. Note that the $p$ quasi-norm is unchanged since the $\SWAP$ map is a unitary map. Note that $X_{EPQ}\geq0$ iff $X_{EQP}\geq0$, i.e. the $\SWAP$ is a completely positive map.
\end{proof}

\subsection{Multiplicativity of Completely Bounded (Co-)Quasi-Norms}\label{sec:app:quasinorm.additivity}

Having established this important inequality we can now apply this to derive several multiplicativity results for mixed norms and co-norms of linear maps. These, among others, will show that completely bounded minimal and maximal output Rényi entropies for any index $\frac{1}{2}\leq p\leq \infty$ are additive and that the notion of completely bounded reverse hypercontractivity is indeed a well defined and useful one, in the sense that the corresponding co-quasi-norm is supermultiplicative under tensor products of channels. We start by considering the mixed quasi-normed setting. 

\subsubsection{Multiplicativity of Completely Bounded Quasi-Norms for linear maps}\label{sec:app:mixed.norm.additivity}

\begin{theorem}\label{thm:mixed.up.cb.additivity}
Let $\Phi:Q_1\to P_1,\Psi:Q_2\to P_2$ be (not necessarily CP) linear maps, then for $0<q\leq p$ with $\frac{1}{q}-\frac{1}{p}\leq1$ it holds that
\begin{align}\label{eq:multiplicativitynorms}
    \|\Phi\otimes\Psi\|_{\cb,(Q_1Q_2:q)\to (P_1P_2:p)} \leq \|\Phi\|_{\cb,(Q_1:q)\to (P_1:p)}\cdot \|\Psi\|_{\cb,(Q_2:q)\to (P_2:p)}
\end{align}
\end{theorem}
\noindent While equality holds in Equation \eqref{eq:multiplicativitynorms}, we only establish submultiplicativity (which is the relevant inequality from a quantum information perspective) in this paper. We leave the proof of achievability to upcoming work \cite{Kochanowski.in.prep} as it requires more advanced machinery. 

\begin{proof}[Proof of \cref{thm:mixed.up.cb.additivity}]
We prove the upper bound using \cref{thm:QuasiNormMinkowskiInequality}, as in \cite[Theorem 11]{Devetak.2006}. Fix some size of the system $E$, then by straightforward submultiplicativity of the ($\cb$) quasi-norm under compositions we get
\begin{align}
    \|\id_E&\otimes\Phi\otimes\Psi\|_{(E:p,Q_1Q_2:q)\to (E:p,P_1P_2:p)}  \\ &\leq \|\id_E\otimes\Phi\otimes\id_{Q_2}\|_{(E:p,Q_1Q_2:q)\to (EP_1:p,Q_2:q)}\|\id_E\otimes\id_{P_1}\otimes\Psi\|_{(EP_1:p,Q_2:q)\to (E:p,P_1P_2:p)} \\ 
    &\overset{(1)}{\leq} \|\id_E\otimes\id_{Q_2}\otimes\Phi\|_{(E:p,Q_1Q_2:q)\to (E:p,Q_2:q,P_1:p)}\|\id_E\otimes\id_{P_1}\otimes\Psi\|_{(EP_1:p,Q_2:q)\to (EP_1:p,P_2:p)} \\ 
    &\overset{(2)}{\leq} \|\id_E\otimes\id_{Q_2}\otimes\Phi\|_{(E:p,Q_1Q_2:q)\to (E:p,Q_2:q,P_1:p)}\|\Psi\|_{\cb,(Q_2:q)\to (P_2:p)}  \\
    &\overset{(3)}{\leq}\|\Phi\|_{\cb,(Q_1:q)\to (P_1:p)}\|\Psi\|_{\cb,(Q_2:q)\to (P_2:p)}.
\end{align} Taking now $\sup_E$ yields the statement of the theorem thanks to \cref{lem:cb.norm.via.t}.
Here in $(1)$ we used \cref{thm:QuasiNormMinkowskiInequality} and in $(2)$ \cref{lem:cb.norm.via.t} by combining the systems $EP_1$. (3) requires the slightly more involved claim:
\begin{align}
\|\id_E\otimes\id_{Q_2}\otimes\Phi\|_{(E:p,Q_1Q_2:q)\to (E:p,Q_2:q,P_1:p)}\leq \|\Phi\|_{\cb,(Q_1:q)\to (P_1:p)}.
\end{align} 
To show the above bound, we prove that the mixed quasi-norm is unchanged when changing the index on the $E$ system from $p$ to $q$ by following the proof strategy of \cref{lem:cb.norm.via.t} but using \cref{def:3.indexed.special}. We will then conclude with the same argument as in the inequality $(2)$ just above: first, by \cref{equ:Def1.2} we have 
\begin{align}
\|X_{EQ_2Q_1}\|_{(E:p,Q_1Q_2:q)}&=\sup_{a,b\geq0}\|a\|_{2r}^{-1}\|b\|_{2r}^{-1}\|a_EX_{EQ_2Q_1}b_E\|_{(E:q,Q_1Q_2:q)} \\ &\geq \|\hat{a}_EX_{EQ_2Q_1}\hat{b}_E\|_{(E:q,Q_1Q_2:q)}
\end{align} for any $\hat{a}, \hat{b}$ on $E$ s.t. $\|\hat{a}\|_{2r}=\|\hat{b}\|_{2r}=1$. Similarly, by \cref{def:3.indexed.special} for any $\epsilon>0$ we have for any prior fixed $X\equiv X_{EQ_2Q_1}$:
\begin{align}
    \|(\id_E\otimes\id_{Q_2}\otimes\Phi)(X)\|_{(E:p,Q_2:q,P_1:p)}&\!=\!\sup_{a,b\geq0}\|a\|_{2r}^{-1}\|b\|_{2r}^{-1}\|a_E(\id_E\otimes\id_{Q_2}\otimes\Phi)(X_{EQ_2Q_1})b_E\|_{(EQ_2:q,P_1:p)} \\ &\!=\! \sup_{a,b\geq0}\|a\|_{2r}^{-1}\|b\|_{2r}^{-1}\|(\id_E\otimes\id_{Q_2}\otimes\Phi)(a_EX_{EQ_2Q_1}b_E)\|_{(EQ_2:q,P_1:p)} \\
    &\!\leq\!\|(\id_E\otimes\id_{Q_2}\otimes\Phi)(\hat{a}_EX_{EQ_2Q_1}\hat{b}_E)\|_{(EQ_2:q,P_1:p)}+\epsilon
\end{align} 
for some approximately optimal $\hat{a},\hat{b}$ which we may normalize to satisfy $\|\hat{a}\|_{2r}=\|\hat{b}\|_{2r}=1$.
Putting these together we get
\begin{align}
    \|\id_E\otimes\id\otimes\Phi\|_{(E:p,Q_1Q_2:q)\to (E:p,Q_2:q,P_1:p)} &=\sup_X\frac{\|\id_E\otimes\id_{Q_2}\otimes\Phi(X)\|_{(E:p,Q_2:q,P_1:p)}}{\|X\|_{(E:p,Q_1Q_2:q)}} \\ &\leq\sup_X \frac{\|(\id_E\otimes\id_{Q_2}\otimes\Phi)(\hat{a}_EX_{EQ_2Q_1}\hat{b}_E)\|_{(EQ_2:q,P_1:p)}+\epsilon}{\|\hat{a}_EX_{EQ_2Q_1}\hat{b}_E\|_{(E:q,Q_1Q_2:q)}} \\
    &\leq \sup_Z\frac{\|(\id_E\otimes\id_{Q_2}\otimes\Phi)(Z_{EQ_2Q_1})\|_{(EQ_2:q,P_1:p)}+\epsilon}{\|Z_{EQ_2Q_1}\|_{(E:q,Q_1Q_2:q)}}
\end{align}
Since this holds for any $\epsilon>0$ we conclude that
\begin{align}
  \|\id_E\otimes\id_{Q_2}\otimes\Phi\|_{(E:p,Q_1Q_2:q)\to (E:p,Q_2:q,P_1:p)} \leq   \|\id_E\otimes\id_{Q_2}\otimes\Phi\|_{(E:q,Q_1Q_2:q)\to (EQ_2:q,P_1:p)}.
\end{align}
and combining the systems $EQ_2$, as they have the same index $q$ and applying \cref{lem:cb.norm.via.t} we have proved the claim and thus the theorem.
\end{proof}

\noindent Next we present two results related to showing stability of certain mixed quasi-norms under adjoining identity channels and along the way prove positivity of optimizers for CP maps. 

\begin{theorem}\label{thm:quasi.chain.rule}
Let $t\geq 1$ and $q,p\geq\frac{1}{2}$ s.t. $\max\{|\frac{1}{t}-\frac{1}{q}|,|\frac{1}{t}-\frac{1}{p}|\}\leq 1$. Then for any CP map $\Phi:Q\to P$ it holds that
\begin{align}
    \|\Phi\otimes\id_T\|_{(q,t)\to (p,t)}=\|\Phi\|^+_{q\to p}
\end{align}
\end{theorem}
This chain rule will become important in the special case where $t=1$.

\begin{remark}
This chain rule type equality is a generalization of the normed version \cite[Lemma 5]{Devetak.2006} to $q,p\geq \frac{1}{2}$. It is optimal in the sense that the above given range of $q,p$ for which this holds is maximal, because for $t<1$ this will generally not hold: taking $p=q=t<1$ with $\Phi$ the completely depolarizing channel suffices to see this as we demonstrate below in \cref{counter:1}.  

Moreover, this theorem is significant because it extends an influential result of \cite{Watrous.2004} to Schatten-quasi-norms.  Namely as a consequence, when setting the system $T$ to be trivial, it follows directly that for any $\frac{1}{2}\leq q,p$ the optimizers of these $q\to p$ quasi-norms of CP maps can be taken to be positive. 
Further from the proof it becomes evident that this also holds for any $0<q,p$, which is the following \cref{cor:postivie.mixed.optimizers} and answers an open question posed in a prior version of this manuscript. 
\end{remark}

\begin{corollary}\label{cor:postivie.mixed.optimizers}
Let $0<q,p$, then for a CP map $\Phi:Q\to P$ it holds that
\begin{align}
    \|\Phi\|_{q\to p}=\|\Phi\|_{q\to p}^+.
\end{align}
\end{corollary}

Before proving the corollary we prove the theorem.
Its proof is more or less the one of \cite[Lemma 5]{Devetak.2006} or \cite[Theorem 4.1]{Fawzi.2026}, using \cref{thm:qt.to.pt} instead of Pisier's theorem. For completeness we present it here.
\begin{proof}[Proof of \cref{thm:quasi.chain.rule}]
Since clearly $\|\Phi\otimes\id_T\|_{(q,t)\to (p,t)}\geq\|\Phi\|^+_{q\to p}$ it suffices to prove the reverse inequality.
Denote with $\Phi(\cdot)=\sum^n_{i=1}K_i\cdot K_i^*$ the Kraus decomposition of the CP map $\Phi$ and take a decomposition $X_{QT}=a_QY_{QT}b_Q$. Then we can write 
\begin{align}
    (\Phi\otimes\id_T)(X_{QT})=\sum^n_{i=1}(K_iA\otimes\1)Y(\1\otimes BK_i^*) = V_A(\1_N\otimes Y_{QT})V^*_B
\end{align} where $\1\in\cB(\cH_T)$ and $V_A=(K_1A\otimes\1, \cdots, K_nA\otimes\1)$ is a block-row vector with blocks $K_iA\otimes\1 $ and similarly $V^*_B$ is a block-colummn vector with blocks $BK_i^*\otimes\1$. By padding with $0$-blocks these operators $V_A,V_B$ can be embedded into the spaces $\cB(\mathbb{C}^n\otimes\cH_Q\otimes\cH_T,\cH_N\otimes\cH_P\otimes\cH_T)$ and similarly $V^*_A,V^*_B$ into the space $\cB(\mathbb{C}^n\otimes\cH_P\otimes\cH_T,\cH_N\otimes\cH_Q\otimes\cH_T)$ by padding with $0$-columns. Call these extended operators, respectively, $V^\prime_A\otimes\1_T,V^\prime_B\otimes\1_T,V_A^{\prime\,*}\otimes\1_T,V_B^{\prime\,*}\otimes\1_T$, then one easily checks that $V^\prime_AV_A^{\prime\,*}\otimes\1_T=|1\rangle\langle 1|\otimes V_AV_A^*=|1\rangle\langle 1|\otimes\Phi(AA^*)$, likewise for $B$ and that $|1\rangle\langle 1|\otimes(\Phi\otimes\id_T)(X_{QT})=(V^\prime_A\otimes\1_T)(\1_N\otimes Y_{QT})(V^{\prime\,*}_B\otimes\1_T)$. Then by \cref{thm:qt.to.pt} (Case 1 or 6) it follows with $t\geq 1$ that
\begin{align}
    \|(\Phi\otimes\id_T)(X_{QT})\|_{(P:p,T:t)} &\overset{(1)}{=}  \||1\rangle\langle1|_N\otimes(\Phi\otimes\id_T)(X_{QT})\|_{(NP:p,T:t)} \\
    &=\|(V^\prime_A\otimes\1_T)(\1_N\otimes Y_{QT})(V^{\prime\,*}_B\otimes\1_T)\|_{(NP:p,T:t)} \\
    &\overset{(2)}{\leq} \|V_A^\prime V_A^{\prime\,*}\|^\frac{1}{2}_{p}\|V_B^\prime V_B^{\prime\,*}\|^\frac{1}{2}_{p}\|(\1_N\otimes Y_{QT})\|_{(N:\infty,Q:\infty,T:t)} \\
    &\overset{(3)}{=} \|\Phi(AA^*)\|^\frac{1}{2}_{p}\|\Phi(B^*B)\|^\frac{1}{2}_{p}\|Y_{QT}\|_{(Q:\infty,T:t)} \\
    &\leq \|\Phi\|^+_{q\to p}\|AA^*\|_q^\frac{1}{2}\|B^*B\|_q^\frac{1}{2}\|Y_{QT}\|_{(Q:\infty,T:t)}.
\end{align} Now minimizing over all such partitions yields with Case 1 of \cref{thm:qt.to.pt} the claim, as $t\geq 1$ since it needs to be compatible with $\infty$. $(2)$ follows also via Case 1 of \cref{thm:qt.to.pt}. In $(3)$ we used \cite[Corollary 1.3]{Book.Pisier.1998}.
In $(1)$ we used that $\||1\rangle\langle1|_N\otimes Z_{PT}\|_{(NP:p,T:t)}=\|Z_{QT}\|_{(P:p,T:t)}$. For $p\geq 1$ this holds by \cite[Corollary 1.3]{Book.Pisier.1998} and for $p<1$ since
\begin{align}
    \||1\rangle\langle1|_N\otimes Z_{PT}\|_{(NP:p,T:t)} &= \inf_{a_{NP},b_{NP}\geq0}\|a\|_{2r}\|b\|_{2r}\|a_{NP}^{-1}(|1\rangle\langle1|_N\otimes Z_{PT})b_{NP}^{-1}\|_t \\
    &\overset{(4)}{=}\inf_{a^\prime_P,b^\prime_P\geq0}\|a^\prime\|_{2r}\|b^\prime\|_{2r}\||1\rangle\langle1|_N\otimes a^{\prime\,-1}_PZ_{PT}b^{\prime\,-1}\|_t \\
    &= \|Z_{PT}\|_{(P:p,T:t)},
\end{align} where $(4)$ follows, as by \cref{lem:SupportConditions} we can have that $\Pi_a=L_{|1\rangle\langle1|_N\otimes Z_{PT}}=|1\rangle\langle1|_N\otimes L_{Z_{PT}}$ and likewise $\Pi_b=R_{|1\rangle\langle1|_N\otimes Z_{PT}}=|1\rangle\langle1|_N\otimes R_{Z_{PT}}$. This implies, necessarily that $a_{NP}=|1\rangle\langle1|_N\otimes a^\prime_P$ which is verified since in particular $a_{NP}=|1\rangle\langle1|_Na_{NP}|1\rangle\langle1|_N=|1\rangle\langle1|_N\otimes a^\prime_P$. Analogously for $b_{NP}=|1\rangle\langle1|_N\otimes b^\prime_P$.
\end{proof}

\begin{proof}[Proof of \cref{cor:postivie.mixed.optimizers}]
The proof of this follows exactly as the proof of \cref{thm:quasi.chain.rule}, just without having the $T$ system present. Then all uses of \cref{thm:qt.to.pt} are replaced by \cref{prop:SpVariationalFormulas}, which drops the requirement $t\geq1$ and this works for any $q,p>0$. For completeness we have
\begin{align}
    \|\Phi(X_{Q})\|_{p} &\overset{(1)}{=}  \||1\rangle\langle1|_N\otimes\Phi(X_{Q})\|_{p} \\
    &=\|V^\prime_A(\1_N\otimes Y_{Q})V^{\prime\,*}_B\|_{p} \\
    &\overset{(2)}{\leq} \|V_A^\prime V_A^{\prime\,*}\|^\frac{1}{2}_{p}\|V_B^\prime V_B^{\prime\,*}\|^\frac{1}{2}_{p}\|\1_N\otimes Y_{Q}\|_{\infty} \\
    &\overset{(3)}{=} \|\Phi(AA^*)\|^\frac{1}{2}_{p}\|\Phi(B^*B)\|^\frac{1}{2}_{p}\|Y_{Q}\|_{\infty} \\
    &\leq \|\Phi\|^+_{q\to p}\|AA^*\|_q^\frac{1}{2}\|B^*B\|_q^\frac{1}{2}\|Y_{Q}\|_{\infty}.
\end{align}
\end{proof}

\noindent The previous \cref{thm:quasi.chain.rule} directly implies stability of adjoining identity maps also to the left within mixed norms when $q\geq 1\geq p$.

\begin{lemma}\label{thm:q.downto.p.cb.equal.noncb}
Let $\Phi:Q\to P$ be a CP map, then for $q\geq 1\geq p$ with $\frac{1}{p}-\frac{1}{q}\leq 1$ it holds that
\begin{align}
    \|\Phi\|_{\cb,(Q:q)\to (P:p)}= \|\Phi\|^+_{(Q:q)\to (P:p)}.
\end{align}
\end{lemma}
\begin{remark}
This result is an extension of \cite[Theorem 13]{Devetak.2006} where the statement was proved for $q\geq p \geq 1$. It is optimal in the sense that for $q<1$ it cannot hold as we demonstrate in \cref{counter:1}.
\end{remark}

\begin{proof}
Clearly it holds that $ \|\Phi\|_{\cb,(Q:q)\to (P:p)}\geq \|\Phi\|^+_{(Q:q)\to (P:p)}$ so it suffices to prove the reverse inequality.
Now due to \cref{thm:QuasiNormMinkowskiInequality} and $q\geq 1\geq p$ it follows that
\begin{align}
    \|\id_E\otimes\Phi\|_{(E:1,Q:q)\to (E:1,P:p)} &= \sup_{X_{EQ}} \frac{\|\id_E\otimes\Phi(X_{EQ})\|_{(E:1,P:p)}}{\| X_{EQ} \|_{(E:1,Q:q)}} \\ &\leq  \sup_{X_{EQ}} \frac{\|(\Phi \otimes \id_E)(X_{QE})\|_{(P:p,E:1)}}{\| X_{QE} \|_{(Q:q,E:1)}} \\& = \|\Phi\otimes\id_E\|_{(Q:q,E:1)\to (P:p,E:1)}\\
    &= \|\Phi\|^+_{(Q:q)\to (P:p)},
\end{align} where in the last equality, we used \cref{thm:quasi.chain.rule}.
\end{proof}

\noindent The following counterexample shows that the two results presented here are optimal w.r.t their parameter regime in the sense that when taking $t<1$ in \cref{thm:quasi.chain.rule} or $q<1$ in \cref{thm:q.downto.p.cb.equal.noncb} their statements cannot be true. This is a consequence of the non-convex geometry underlying the appearing quasi-norms.
\begin{counterexample}[Counterexample to \cref{thm:quasi.chain.rule} with $t<1$ and \cref{thm:q.downto.p.cb.equal.noncb} with $q<1$]\label{counter:1}
Let $1>q\geq p$ be compatible and consider $\Phi: \cB(\mathbb{C}^{d})\to \cB(\mathbb{C}^{d}),\, \rho\mapsto \Tr[\rho]\frac{\1}{d}$ to be the replacer channel with the maximally mixed state with $d>1$. This map is clearly completely positive, trace preserving, and unital. Then on the one hand we can simply lower bound the completely bounded mixed norm by an ancilla system of the same dimension $d$ to get
\begin{align}
    \|\Phi\|^+_{\cb,q\to p}\geq\sup_{\rho_{12}\in\cD(\cH_1\otimes\cH_2)}\frac{\|\tr_2[\rho_{12}]\otimes\frac{\1_2}{d}\|_{(q,p)}}{\|\rho_{12}\|_{q}}= d^{\frac{1}{p}-1}\cdot \sup_{\rho_{12}}\frac{\|\tr_2[\rho_{12}]\|_q}{\|\rho_{12}\|_q}\geq d^{\frac{1}{p}-1}d^{\frac{1}{q}-1},
\end{align} where the last lower bound follows from lower bounding the supremum over all bipartite operators $\rho_{12}$ with the maximally entangled state $\rho_{12}=\Omega_{12}=|\Omega\rangle\langle\Omega|$, s.t. $\|\Omega_{12}\|_q=1$ and $\tr_2[\Omega_{12}]=\frac{\1_1}{d}$. On the other hand, we may upper bound the non cb norm as
\begin{align}
    \|\Phi\|_{q\to p}&=\sup_\rho\frac{\|\frac{\1}{d}\Tr[\rho]\|_p} {\|\rho\|_q} =d^{\frac{1}{p}-1}\sup_\rho \frac{|\Tr[\rho]|}{\|\rho\|_q} \leq d^{\frac{1}{p}-1}\sup_\rho \frac{\|\rho\|_1}{\|\rho\|_q}\leq d^{\frac{1}{p}-1},
\end{align} where the last inequality follows from the well known fact that $\|\rho\|_1\leq\|\rho\|_q$ since $q \leq 1$.  Note that this entire bound is achievable. 
Hence we have shown that for this particular choice of channel
\begin{align}
    \|\Phi\|^+_{\cb,q\to p}\geq d^{\frac{1}{p}-1}d^{\frac{1}{q}-1} > d^{\frac{1}{p}-1} = \|\Phi\|^+_{q\to p},
\end{align} for $q<1$. Additionally, when $q=p=t<1$ it holds true that
\begin{align}
    \|\Phi\|^+_{\cb,t\to t}=\sup_n\|\id_n\otimes\Phi\|^+_{(t,t)\to (t,t)} = \sup_n\|\Phi\otimes\id_n\|^+_{(t,t)\to (t,t)}>\|\Phi\|_{t\to t} \geq \|\Phi\|^+_{t\to t}.
\end{align} Hence there must exists some $N\in\mathbb{N}$ s.t. $\|\Phi\otimes\id_N\|^+_{(t,t)\to (t,t)}>\|\Phi\|_{t\to t}$
which constitutes a counterexample to \cref{thm:quasi.chain.rule} with $t<1$.
\end{counterexample}

\noindent Next we present a multiplicativity statement which corresponds to the additivity of maximal output Rényi entropies with indices $p\in[\frac{1}{2},1)$. As far as we are aware, this is the first additivity result for maximum output entropies. An important distinction to the minimum output entropy case is that it holds already for the maximum output Rényi entropy and does not require completely bounded Rényi entropies. We recall the definition of the $p$-Rényi-entropy of a state $\omega$ on system $A$: 
\[H_p(A)_{\omega}:=\frac{1}{1-p}\log\Tr[\omega_A^p]\,.\]

\begin{theorem}[Multiplicativity of completely bounded $1\to p$ quasi-norms for $p<1$]\label{thm:cb.1.to.p.multiplicativity}
Let $\Phi:Q_1\to A_1\,,\Psi:Q_2\to A_2$ be completely bounded CP maps and $1\geq p\geq\frac{1}{2}$, then it holds that
\begin{align}
    \|\Phi\otimes\Psi\|^+_{\cb,(Q_1Q_2:1)\to (A_1A_2:p)}= \|\Phi\|^+_{(Q_1:1)\to (A_1:p)}\|\Psi\|^+_{(Q_2:1)\to (A_2:p)}
\end{align} which is equivalent to
\begin{align}
\sup_{E}\sup_{\rho_{EQ_1Q_2}}H_p^\uparrow(A_1A_2|E)_{(\id_E\otimes\Phi\otimes\Psi)(\rho_{EQ_1Q_2})} =
\sup_{\rho_{Q_1}}H_p(A_1)_{\Phi(\rho_{Q_1})}+\sup_{\rho_{Q_2}}H_p(A_2)_{\Psi(\rho_{Q_2})}\,.
\end{align} 
\end{theorem}
\begin{remark}
Since conditioning on environments on the LHS can only decrease the conditional entropy, this theorem implies in particular that $$\sup_{\rho_{Q_1Q_2}}H_p(A_1A_2)_{(\Phi\otimes\Psi)(\rho_{Q_1Q_2})} = \sup_{\rho_{Q_1}}H_p(A_1)_{\Phi(\rho_{Q_1})}+\sup_{\rho_{Q_2}}H_p(A_2)_{\Psi(\rho_{Q_2})}.$$
It is an open question whether this multiplicativity extends to arbitrary input norm-indices $q\geq 1$. For $q<1$ this cannot be true due to the optimality of \cref{thm:q.downto.p.cb.equal.noncb}, see also \cref{counter:1}. 
\end{remark}
\begin{proof}[Proof of \cref{thm:cb.1.to.p.multiplicativity}]
Firstly by \cref{thm:q.downto.p.cb.equal.noncb}, since $q=1\geq p$ we have
\begin{align}
    \|\Phi\otimes\Psi\|^+_{\cb,(Q_1Q_2:1)\to (A_1A_2:p)} = \|\Phi\otimes\Psi\|^+_{(Q_1Q_2:1)\to (A_1A_2:p)},
\end{align} for which it holds that
\begin{align}
    \|\Phi\otimes\Psi\|^+_{(Q_1:1,Q_2:1)\to (A_1:p,A_2:p)} &= \|(\id_{A_1}\otimes\Psi)\circ (\Phi\otimes\id_{Q_2})\|^+_{(Q_1:1,Q_2:1)\to (A_1:p,A_2:p)} \\&\leq \|\Phi\otimes\id_{Q_2}\|^+_{(Q_1:1,Q_2:1)\to (A_1:p,Q_2:1)}\|\id_{A_1}\otimes\Psi\|^+_{(A_1:p,Q_2:1)\to (A_1:p,A_2:p)} \\ &\leq \|\Phi\|^+_{(Q_1:1)\to (A_1:p)}\, \|\Psi\|^+_{\cb,(Q_2:1)\to (A_2:p)} \\
    &=\|\Phi\|^+_{(Q_1:1)\to (A_1:p)}\, \|\Psi\|^+_{(Q_2:1)\to (A_2:p)},
\end{align} 
where in the second inequality we used \cref{thm:quasi.chain.rule} and the definition of the \(\cb\) norm with an index $p$ on the environment system. The last equality follows from \cref{thm:q.downto.p.cb.equal.noncb}. 
Achievability in this case is quite simple since for any $X,Y\geq0$ it holds that
\begin{align}
  \|\Phi\otimes\Psi\|^+_{\cb,(Q_1Q_2:1)\to (A_1A_2:p)} \geq \frac{\|(\Phi\otimes\Psi)(X_{Q_1}\otimes Y_{Q_2})\|_{(A_1:p,A_2:p)}}{\|X\otimes Y\|_{(Q_1:1,Q_2:1)}}= \frac{\|\Phi(X_{Q_1})\|_p}{\|X_{Q_1}\|_1}\frac{\|\Psi(Y_{Q_2})\|_p}{\|Y_{Q_2}\|_1}
\end{align} and maximizing the RHS yields the achievability.
To derive the alternate expression see that for a CP map $\Phi:Q\to A$ and $p<1$
\begin{align}    \sup_{\rho\in\cD(\cH_Q)}H_p(A)_{\Phi(\rho)}&=\sup_{\rho\in\cD(\cH_Q)}\frac{p}{1-p}\log\|\Phi(\rho)\|_p = \frac{p}{1-p}\log\sup_{\rho\in\cD(\cH_Q)}\|\Phi(\rho)\|_p = \frac{p}{1-p}\log\|\Phi\|^+_{1\to p}.
\end{align}
\end{proof}

\subsubsection{Multiplicativity of Completely Bounded Co-Quasi-Norms for CP maps}\label{sec:app:mixed.conorm.additivity}

Analogously to the mixed norms, for the co-norms we prove the following reverse multiplicativity. 
\begin{theorem*}[\cref{thm:cb.reverse hypercontractive.bound}]
Let $q\geq p>0$ be s.t. $\frac{1}{p}-\frac{1}{q}\leq 1$ and $\Phi:Q_1\to P_1,\,\Psi:Q_2\to P_2$ be CP maps, then it holds that
\begin{align}\label{eqcocornmulti}
    \|\Phi\otimes\Psi\|^+_{\cb,\co,(Q_1Q_2:q)\to (P_1P_2:p)}\geq \|\Phi\|^+_{\cb,\co,(Q_1:q)\to (P_1:p)}\cdot \|\Psi\|^+_{\cb,\co, (Q_2:q)\to (P_2:p)}.
\end{align}
\end{theorem*}
\noindent As in the case of \cref{thm:mixed.up.cb.additivity}, although  inequality \eqref{eqcocornmulti} is achievable, we leave the proof of this fact to future work \cite{Kochanowski.in.prep}, as it requires more tools which we believe deserve their own exposition. 

This \cref{thm:cb.reverse hypercontractive.bound} is important and should be considered as one of the main results of this section and the paper for the following two reasons. On the one hand it is necessary for defining a concept of \emph{complete reverse hypercontractivity} which builds on the (super)multiplicativity of the above completely bounded $q\to p$-conorm (for $1> q\geq p\geq\frac{1}{2}$) under tensor products of channels.
\begin{remark}\label{rem:cb.reverse.hypercontractivity}
A quantum channel is called $q\to p$ reverse hypercontractive \cite{Beigi.2020} if for any operator $X>0$ and any $1>q$ it holds that $\|\Phi(X)\|_p\geq \|X\|_q$, i.e. if
\begin{align}
\|\Phi\|^+_{\co,q\to p}\geq1.
\end{align} 
In the context of quantum Markov semigroups (QMS) this concept becomes interesting whenever $1>q> p\geq\frac{1}{2}$ \cite[Theorem 11]{Beigi.2020}. There it relates the logarithmic Sobolev inequality, related to the mixing time to the first time a QMS becomes $q\to p$ reverse hypercontractive. The problem with this concept, however, is when considering tensor-products of QMS $\Phi\otimes\Psi$, as
\begin{align}
    \|\Phi\otimes\Psi\|^+_{\co,q\to p}\leq \|\Phi\|^+_{\co,q\to p}\|\Psi\|^+_{\co,q\to p},
\end{align}
with the inequality in general assumed to be strict. 
Hence even if both $\Phi, \Psi$ are $q\to p$-reverse hypercontractive, then their tensorproduct channel (generated by a suitable sum of their generators) may not be. This problem is solved by introducing the notion of complete reverse hypercontractivity. We say, analogously to above, that a channel $\Phi$ satisfies \emph{complete }$q\to p$\emph{ reverse hypercontractivity} if for $1>q\geq p\geq\frac{1}{2}$,
\begin{align}
    \|\Phi\|^+_{\cb,\co,q\to p}\geq1.
\end{align}
Then by \cref{thm:cb.reverse hypercontractive.bound} if for two channels $\Phi,\Psi$ they each are completely $q\to p$ reverse hypercontractive, then so is their tensorproduct, i.e.
\begin{align}
    \|\Phi\otimes\Psi\|^+_{\cb,\co,q\to p}\geq \|\Phi\|^+_{\cb,\co,q\to p}\|\Psi\|^+_{\cb,\co,q\to p}\geq1
\end{align}
This notion is thus completely analogous to the notion of complete hypercontractivity \cite{Beigi.2016, Bardet.2024} and extends reverse hypercontractivity \cite{Beigi.2020, Cubitt.2015}. 
We note that the supermultiplicativity property we demonstrate here is not sufficient, but necessary, to make it a sensible property of a QMS connected to its mixing time and we leave the details of proving a connection between this notion of complete reverse hypercontractivity of a QMS and it actually satisfying a complete modified logarithmic Sobolev inequality to future work.

\end{remark}

On the other hand it justifies our definition of completely bounded co-quasi-norm, since for $q=1$, it reproduces naturally a minimal output entropy for $p\in[\frac{1}{2},1)$ for CP maps.

\noindent Before we give the proof of this statement, we write it in terms of additivity of the completely bounded minimal output Rényi-$p$-entropy for $\frac{1}{2}\leq p<1$.
\begin{corollary*}[\cref{cor:Additivity.mib.cb.entropy}]
Let $\Phi:Q_1\to A_1\,,\Psi:Q_2\to A_2$ be CP maps and $p\geq\frac{1}{2}$, then 
\begin{align}
\inf_{E,\rho_{EQ_1Q_2}}\!\!\!\!\!H^\uparrow_p(A_1A_2|E)_{(\id_E\otimes\Phi\otimes\Psi)(\rho_{EQ_1Q_2})}\!=\! \!\!\inf_{E,\rho_{EQ_1}}\!\!\!H^\uparrow_p(A_1|E)_{(\id_E\otimes\Phi)(\rho_{EQ_1})}\!\!+\!\!\!\!\inf_{E,\rho_{EQ_2}}\!H^\uparrow_p(A_2|E)_{(\id_E\otimes\Psi)(\rho_{EQ_2})}\,.
\end{align}
\end{corollary*} 
\begin{remark}
This result complements the additivity result of \cite[Theorem 11]{Devetak.2006} of the minimum output $p$-entropy for $p > 1$ by establishing the same statement for $p < 1$.
\end{remark}
\begin{proof}[Proof of \cref{cor:Additivity.mib.cb.entropy}]
The $\geq$ inequality in the corollary directly follows from the theorem with $q=1>p$, when observing that for $p<1$ and any CP map $\Phi$ 
\begin{align}
   \inf_E\inf_{\rho_{EQ}\in\cD(\cH_{E}\otimes\cH_Q)} H_p^\uparrow(A|E)_{(\id\otimes\Phi)(\rho_{EQ})}
   &=\inf_{E,\rho_{EQ}}\frac{p}{1-p}\log\|(\id_E\otimes\Phi)(\rho_{EQ})\|_{(E:1,A:p)} \\
   &=\frac{p}{1-p}\log\inf_{E,\rho_{EQ}\in\cD(\cH_{E}\otimes\cH_Q)}\|(\id_E\otimes\Phi)(\rho_{EQ})\|_{(E:1,A:p)} \\
   &= \frac{p}{1-p}\log\inf_E\|\id_E\otimes\Phi\|^+_{\co,(E:1,Q:1)\to (E:1,A:p)} \\ 
   &=\frac{p}{1-p}\log\|\Phi\|^+_{\cb,\co,(Q:1)\to (A:p)}.
\end{align}
Equality in that bound follows directly from additivity of the optimized sandwiched conditional entropy under tensor-product states \cite[Corollary 5.9]{Book.Tomamichel.2016}.
\end{proof}

\begin{proof}[Proof of \cref{thm:cb.reverse hypercontractive.bound}]
Let $q\geq p>0$ be s.t. $\frac{1}{r}=\frac{1}{p}-\frac{1}{q}\leq1$, then
  \begin{align}
      \|&\id_E\otimes\Phi\otimes\Psi\|^+_{\co,(E:q,Q_1Q_2:q)\to (E:q,P_1P_2:p)}\\
      &\geq \|\id_E\otimes\Phi\otimes\id_{Q_2}\|^+_{\co,(E:q,Q_1Q_2:q)\to (E:q,P_1:p,Q_2:q)}\cdot \|\id_E\otimes\id_{P_1}\otimes\Psi\|^+_{\co,(E:q,P_1:p,Q_2:q)\to (E:q,P_1P_2:p)} \\
      &\overset{(1)}{\geq} \|\id_E\otimes\id_{Q_2}\otimes\Phi\|^+_{\co,(EQ_2:q,Q_1:q)\to (EQ_2:q,P_1:p)}\cdot \|\id_E\otimes\id_{P_1}\otimes\Psi\|^+_{\co,(E:q,P_1:p,Q_2:q)\to (E:q,P_1P_2:p)}
      \\ &\overset{(2)}{\geq} \|\Phi\|^+_{\cb,\co,(Q_1:q)\to (P_1:p)}\cdot \|\Psi\|^+_{\cb,\co,(Q_2:q)\to (P_2:p)},
  \end{align} 
  where $(1)$ follows from \cref{thm:QuasiNormMinkowskiInequality} and in $(2)$ for $\Phi$ we used \cref{lem:cb.co.norm.via.t} on the combined system $EQ_2$. For $\Psi$ we use the following claim
\begin{align}
    \|\id_E\otimes\id_{P_1}\otimes\Psi\|^+_{\co,(E:q,P_1:p,Q_2:q)\to (E:q,P_1P_2:p)} \geq \|\Psi\|^+_{\cb,\co,(Q_2:q)\to (P_2:p)},
\end{align} which we prove analogously to \cref{lem:cb.co.norm.via.t}, however with \cref{def:3.indexed.special}: we have, for any $\epsilon>0$ and $X_{EP_1Q_2}\geq0$ that there exist $\hat{a}$ on $E$ s.t. $\|\hat{a}\|_{2r}=1$ and
\begin{align}
    \|X_{EP_1Q_2}\|_{(E:q,P_1:p,Q_2:q)}=\sup_{a\geq0}\|a\|_{2r}^{-2}\|a_EX_{EP_1Q_2}a_E\|_{(EP_1:p,Q_2:q)}\leq \|\hat{a}_EX_{EP_1Q_2}\hat{a}_E\|_{(EP_1:p,Q_2:q)}+\epsilon
\end{align}
Analogously by \eqref{equ:Def1.2} we have
\begin{align}
    \|(\id_E\otimes\id_{P_1}\otimes\Psi)(X)\|_{(E:q,P_1P_2:p)}&=\sup_{a,b\geq0}\|a\|_{2r}^{-2}\|a_E(\id_E\otimes\id_{P_1}\otimes\Psi)(X_{EP_1Q_2})a_E\|_{p} \\ &= \sup_{a\geq0}\|a\|_{2r}^{-2}\|(\id_E\otimes\id_{P_1}\otimes\Psi)(a_EX_{EP_1Q_2}a_E)\|_{p}  \\
    &\geq \|(\id_E\otimes\id_{P_1}\otimes\Psi)(\hat{a}_EX_{EP_1Q_2}\hat{a}_E)\|_{p}.
\end{align}
Combining these it follows that
\begin{align}
    \inf_{X\geq0}\frac{\|(\id_E\otimes\id_{P_1}\otimes\Psi)(X)\|_{(E:q,P_1P_2:p)}}{\|X_{EP_1Q_2}\|_{(E:q,P_1:p,Q_2:q)}} &\geq \inf_{X\geq0}\frac{\|(\id_E\otimes\id_{P_1}\otimes\Psi)(\hat{a}_EX_{EP_1Q_2}\hat{a}_E)\|_{p}}{\|\hat{a}_EX_{EP_1Q_2}\hat{a}_E\|_{(EP_1:p,Q_2:q)}+\epsilon}\\ &\geq\inf_{Z\geq0}\frac{\|(\id_E\otimes\id_{P_1}\otimes\Psi)(Z_{EP_1Q_2})\|_{p}}{\|Z_{EP_1Q_2}\|_{(EP_1:p,Q_2:q)}+\epsilon}.
\end{align} Since this works for any $\epsilon>0$ it follows that
\begin{align}
  \|\id_E\otimes\id_{P_1}\otimes\Psi\|^+_{\co,(E:q,P_1:p,Q_2:q)\to (E:q,P_1P_2:p)}&\geq \|\id_E\otimes\id_{P_1}\otimes\Psi\|^+_{\co,(EP_1:p,Q_2:q)\to (EP_1:p,P_2:p)} \\&\geq \|\Psi\|^+_{\cb,\co,(Q_2:q)\to (P_2:p)},  
\end{align} where the last inequality holds by \cref{lem:cb.co.norm.via.t}.
\end{proof}

\noindent As in the previous section, we also analyze stability of the mixed-co-quasi-norms under adjoining an identity map. 
\begin{lemma}
\label{lem:adjoining.id.conorm}
Let $\Phi:Q\to P$ be a CP map and let $q,p\geq\frac{1}{2}$ be compatible indices, then it holds that
\begin{align}
    \|\Phi\otimes\id_T\|^+_{\co,(Q:q,T:1)\to (P:p,T:1)} = \|\Phi\|^+_{\co,(Q:q)\to (P:p)}.
\end{align}
\end{lemma}

\begin{proof}
By inputting operators which form a tensor product across the partition $Q$ and $T$ it follows that the LHS is at most as large as the RHS.
The core argument of this proof will be the reverse inequality and will follow from two applications of \cref{cor:q.1.PartialTrace}, which implies that
\begin{align}
    \|(\Phi\otimes\id_T)(\rho_{QT})\|_{(P:p,T:1)} &= \|\tr_T(\Phi\otimes\id_T)(\rho_{QT})\|_{(P:p)} \\ &= \|\Phi(\tr_T[\rho_{QT}])\|_{(P:p)} \\ &\geq \|\Phi\|^+_{\co,(Q:q)\to (P:p)}\|\tr_T[\rho_{QT}]\|_{(Q:q)} \\ &=\|\Phi\|^+_{\co,(Q:q)\to (P:p)}\|\rho_{QT}\|_{(Q:q,T:1)}.
\end{align} Rearranging and minimizing over all $\rho_{QT}$ proves the lemma.
\end{proof}

\noindent The previous Lemma directly implies stability of adjoining identity maps within mixed conorms.

\begin{lemma}\label{thm:q.downto.p.cb.equal.noncb.co}
Let $\frac{1}{2}\leq q\leq 1\leq p$ and $\Phi:Q\to P$ a CP map, then
\begin{align}
    \|\Phi\|^+_{\cb,\co,(Q:q)\to (P:p)}= \|\Phi\|^+_{\co,(Q:q)\to (P:p)}.
\end{align}
\end{lemma}
\begin{proof}
First, trivially it holds that $\|\Phi\|^+_{\cb,\co,(Q:q)\to (P:p)}\leq \|\Phi\|^+_{\co,(Q:q)\to (P:p)}$, so it suffices to consider the reverse inequality.
Note that by assumption, both $q,p$ are compatible indices to $1$, so consider
\begin{align}
    \|\id_E\otimes\Phi\|^+_{\co,(E:1,Q:q)\to (E:1,P:p)}&\geq \|\Phi\otimes\id_E\|^+_{\co,(Q:q,E:1)\to (P:p,E:1)}=\|\Phi\|^+_{\co,(Q:q)\to (P:p)},
\end{align}  
where in the last identity we used \cref{lem:adjoining.id.conorm}. Minimizing over systems $E$ of arbitrary dimension and noting that the RHS is independent of it yields the result of the theorem.
\end{proof}

\noindent For completeness we also show additivity of the maximal output entropy for $p>1$ complementing the $p<1$ case in \cref{thm:cb.1.to.p.multiplicativity}.
\begin{theorem}\label{thm:cb.co.1.to.p.additivity}
Let $\Phi:Q_1\to A_1,\Psi:Q_2\to A_2$ be CP maps and $p\geq1$, then it holds that
\begin{align}
    \|\Phi\otimes\Psi\|^+_{\cb,\co,(Q_1Q_2:1)\to (A_1A_2:p)}=\|\Phi\|^+_{\co,(Q_1:1)\to (A_1:p)}\cdot\|\Psi\|^+_{\co,(Q_2:1)\to (A_2:p)},
\end{align} which is equivalent to
\begin{align}
\sup_{E,\rho_{EQ_1Q_2}}H_p^\uparrow(A_1A_2|E)_{(\id_E\otimes\Phi\otimes\Psi)(\rho_{EQ_1Q_2})} = \sup_{\rho_{Q_1}}H_p(A_1)_{\Phi(\rho_{Q_1})}+\sup_{\rho_{Q_2}}H_p(A_2)_{\Psi(\rho_{Q_2})}.
\end{align} 
\end{theorem}
\begin{remark}
As in \cref{thm:cb.1.to.p.multiplicativity} this also holds without the conditioning system on the LHS, as conditioning can only decrease the conditional entropy. Thus this complements the result of \cref{thm:cb.1.to.p.multiplicativity} and both together imply additivity of the maximum output Rényi entropy for any index $\frac{1}{2}\leq p\leq \infty$ (\cref{cor:max.entropy.additivity}):
\begin{align}
    \sup_{\rho_{Q_1Q_2}}H_p(A_1A_2)_{(\Phi\otimes\Psi)(\rho_{Q_1Q_2})} = \sup_{\rho_{Q_1}}H_p(A_1)_{\Phi(\rho_{Q_1})}+\sup_{\rho_{Q_2}}H_p(A_2)_{\Psi(\rho_{Q_2})}.
\end{align}
\end{remark}

\begin{proof}
Using \cref{thm:q.downto.p.cb.equal.noncb.co} twice, the standard supermultiplicativity of conorms, and \cref{lem:adjoining.id.conorm} on the combined systems $Q_1Q_2, A_1A_2$, we have
\begin{align}
  \|\Phi\otimes\Psi\|^+_{\cb,\co,(Q_1Q_2:1)\to (A_1A_2:p)} &=  \|\Phi\otimes\Psi\|^+_{\co,(Q_1Q_2:1)\to (A_1A_2:p)} \\ 
  &\geq \|\Phi\otimes\id\|^+_{\co,(Q_1Q_2:1)\to (A_1:p,Q_2:1)}\cdot \|\id\otimes\Psi\|^+_{\co, (A_1:p,Q_2:1)\to (A_1A_2:p)} \\ 
  &\geq\|\Phi\|^+_{\co,(Q_1:1)\to (A_1:p)}\cdot \|\Psi\|^+_{\cb,\co,(Q_2:1)\to (A_2:p)}\\&=\|\Phi\|^+_{\co,(Q_1:1)\to (A_1:p)}\cdot \|\Psi\|^+_{\co,(Q_2:1)\to (A_2:p)}.
\end{align} 
Achievability in this case is quite simple since for any $X,Y\geq0, X,Y\neq0$ it holds
\begin{align}
  \|\Phi\otimes\Psi\|^+_{\cb,\co,(Q_1Q_2:1)\to (A_1A_2:p)} \leq \frac{\|(\Phi\otimes\Psi)(X_{Q_1}\otimes Y_{Q_2})\|_{(A_1A_2:p)}}{\|X_{Q_1}\otimes Y_{Q_2}\|_{(Q_1Q_2:1)}}= \frac{\|\Phi(X_{Q_1})\|_p}{\|X_{Q_1}\|_1}\frac{\|\Psi(Y_{Q_2})\|_p}{\|Y_{Q_2}\|_1}\,,
\end{align} and minimizing the RHS yields achievability.
To derive the entropic identity, see that for a CP map $\Phi:Q\to A$ and $p>1$ we have
\begin{align}
\sup_{\rho\in\cD(\cH_Q)}H_p(A)_{\Phi(\rho)}&=\sup_{\rho\in\cD(\cH_Q)}\frac{p}{1-p}\log\|\Phi(\rho)\|_p = \frac{p}{1-p}\log\inf_{\rho\in\cD(\cH_Q)}\|\Phi(\rho)\|_p = \frac{p}{1-p}\log\|\Phi\|^+_{\co,1\to p}.
\end{align}
\end{proof}

\section{Summary and Outlook}

In this work we presented a theory of 2-indexed Schatten-quasi-norms which manifestly generalize Pisier's 2-indexed norms. Our construction in \cref{def:Definition1} is entirely based on factorization statements given via variational formulas instead of operator space theory. This is what allows our construction to work and overcome the obstacles met in trying to approach this problem from a purely operator space theoretic approach. 
As a main result we give the first construction of such quasi-norms which satisfy natural properties one would expect from them:
namely under the necessary and sufficient constraint $\big|\frac{1}{q}-\frac{1}{p}\big|\leq 1$ on the two indices $0<q,p$ not being too far apart we prove in \cref{thm:cq.additivity} that these quasi-norms simplify on block operators to $\ell_q[\cS_p]$-quasi-norms, which implies among others that they are generalizations of the commutative $\ell_q[\ell_p]$-quasi-norms. In \cref{thm:qt.to.pt}, we derive `relational consistency', which in some sense should be viewed as analogues of Pisier's seminal theorem. In \cref{lem:GenRevHölder}, we derive a reverse Hölder's inequality for when $1\geq q,p>0$.

Further in \cref{sec:Apl:Entropies} we established some elementary connections between multiple information measures based on an optimized sandwiched Rényi relative entropy and the 2-indexed Schatten quasi-norms. These results both extend known connections (e.g.~optimized sandwiched Rényi conditional entropy) and create new ones (sandwiched Umlaut information). We leveraged this link to prove properties of these quantities in a simple manner. 

As one of our main applications we used these 2-indexed quasi-norms to define a notion of completely bounded quasi-norm and co-quasi-norm and demonstrated their usefulness to quantum information theory via several multiplicativity/additivity results. To do so, we proved in \cref{thm:QuasiNormMinkowskiInequality} a non-commutative Minkowski inequality for arbitrary (compatible) indices $0<q,p\leq\infty$ generalizing Pisier's seminal result. We then proved several multiplicativity results generalizing effectively all of the results of \cite{Devetak.2006} to indices $0<q,p$. As a consequence we prove additivity of the completely bounded minimal output $p$-Rényi entropy for indices $\frac{1}{2}\leq p\leq 1$ and additivity of the maximal Rényi output entropy for any index $\frac{1}{2}\leq p$.

In \cref{sec:app:mixed.norm.additivity} we discussed the achievability of quasi-norms for CP maps on positive operators extending the analogous result of \cite{Watrous.2004} (and references therein) to any $0<q,p$. In our notations this is $ \|\Phi\|_{q\to p}= \|\Phi\|^+_{q\to p}$.

Lastly we present a conjecture concerning our 2-indexed quasi-norms forming an interpolation scale. For an introduction to (complex) interpolation see \cite{book:interpolation.BL.76}.
A central step in the construction of Pisier norms is complex interpolation \cite{Book.Pisier.1998, Beigi.2023}, which as a consequence means that they form an interpolation scale, i.e. for $\theta\in[0,1]$ and indices $1\leq q_0,p_0,q_1,p_1\leq\infty$ it holds 
\begin{align}\label{equ:Pisiernorm.interpolation}
    \cS_q[\cH_1,\cS_p(\cH_2)] = \left[\cS_{q_0}[\cH_1,\cS_{p_0}(\cH_2)],\cS_{q_1}[\cH_1,\cS_{p_1}(\cH_2)]\right]_\theta,
\end{align} where $\frac{1}{q}=\frac{1-\theta}{q_0}+\frac{\theta}{q_1}$ and $\frac{1}{p}=\frac{1-\theta}{p_0}+\frac{\theta}{p_1}$, see \cite[Definition 4.1, Proposition 4.3]{Beigi.2023}.
Likewise Schatten-quasi-norms form an interpolation scale for all positive indices $0<q\leq\infty$ \cite{Gu.2019, Xu.1990}. So it seems natural to assume the same to hold true for 2-indexed Schatten-quasi-norms defined in this work, effectively extending \eqref{equ:Pisiernorm.interpolation}. We conjecture
\begin{conjecture}
The quasi-Banach spaces defined in \cref{thm:QuasiNomrms} are complex interpolations spaces in the following sense. Let $\theta,r\in[0,1]$ and $1\leq q_1,p_1\leq\infty$  then completely isometrically
        \begin{align}
            \cS_q[\cH_1,\cS_p(\cH_2)] = \left[\cS_{r}(\cH_1\otimes\cH_2),\cS_{q_1}[\cH_1,\cS_{p_1}(\cH_2)]\right]_\theta
        \end{align} for $\frac{1}{q}=\frac{1-\theta}{r}+\frac{\theta}{q_1}$ and $\frac{1}{p}=\frac{1-\theta}{r}+\frac{\theta}{p_1}$.
    \end{conjecture}
\noindent Note that this conjecture requires only the well-established Pisier norms and regular Schatten-$r$-quasi-norms. Observe also that the condition for it to be an interpolation space gives naturally rise to the condition $\big|\frac{1}{q}-\frac{1}{p}\big|\leq1$ since
\begin{align}
    \left|\frac{1}{q}-\frac{1}{p}\right|=\left|\frac{1-\theta}{r}+\frac{\theta}{q_1}-\frac{1-\theta}{r}-\frac{\theta}{p_1}\right|= \theta \left|\frac{1}{q_1}-\frac{1}{p_1}\right|\leq 1.
\end{align}

\noindent Finally, for simplicity we restricted ourselves to the setting of finite dimensional Hilbert spaces in this work. We believe, however, that with essentially no modifications it extends also to separable ones and extensions to type II and III von Neumann algebras seem also quite likely. The technical difficulty in proving all our results in separable Hilbert spaces is that we have to be more careful when invoking Sion's minimax \cref{thm:Sion}. It requires one of the sets over which we optimize to be compact. In this work we could always choose this set to be the set of subnormalized quantum states, which for separable Hilbert spaces, however, is only weak-$^*$ compact by Banach-Alaoglu, implying that one effectively only needs to prove weak-$^*$ upper semicontinuity in \cref{lem:continuity} to conclude all our results for infinite dimensional separable Hilbert spaces. \newline

\noindent \textbf{Acknowledgments}
JK and OF would like to thank Mikael de la Salle for helpful discussions regarding quasi norms and operator space theory and Lewis Wooltorton for helpful discussions related to our additivity results. JK would like to thank Amir Arqand for asking a question which lead to \cref{thm:quasi.chain.rule} and the resolution of an open question posed in an earlier version of the manuscript. JK acknowledges support from the Program QuanTEdu-France n° ANR-22-CMAS-0001 France 2030. OF acknowledges support from the European Research Council (ERC Grant AlgoQIP, Agreement No. 851716). CR is supported by France 2030 under the French National Research Agency award number ``ANR-22-EXES-0013''. \newline

\noindent While finalizing this manuscript we were made aware of the independent and concurrent work \cite{Li.2026}. For a brief discussion of how their main result is a special case of our see \cref{ref:concurrent.work}.

\bibliographystyle{abbrv}
\bibliography{lib}

\hfill
\newpage

\appendix
\section{Proofs Omitted from the Main Text}

\subsection{Variational Formulas for Two Simple Cases}

While we ultimately care about the non-commutative 2-indexed norm case we will for intuition building present the much simpler commutative case and the single-indexed non commutative case here. They will both be rather simple to derive from 'just' the generalized Hölder inequality in the respective setting. However, both will have some striking similarity which we capture in our \cref{def:Definition1}. 

\subsubsection{Variational Formulas for $\ell_q[\ell_p]$ Quasi-Norms}\label{app:lplqVariationalFormulas}

Consider the rather simple and well known case of 2-indexed $\ell_q[\ell_p]$ quasi-norms. For any indices $0<q,p<\infty$ and any vector $v\in\ell_q[\ell_p]\subset \ell_\infty\otimes\ell_\infty$ one defines 
\begin{align}
\|v\|_{\ell_q[\ell_p]}:=\left(\sum_i\left(\sum_j|v_{ij}|^p\right)^\frac{q}{p} \right)^\frac{1}{q}.
\end{align}

While in this setting it may at first be quite unnatural and unhelpful to rewrite them as optimization problems we will do so to highlight the similarity to the non-commutative setting later on, which is the main part of this work. In fact, using standard tools it can be shown that these quasi-norms are also given by variational formulas, completely analogous to the Pisier norms, however, manifestly for all indices $0<q,p\leq \infty$.

The goal of this subsection is to illustrate and prove this.
\begin{theorem*}[Multi-indexed $\ell_q$ spaces \cref{thm:classical-2-indexed-spaces}]
Let $0<p,q\leq\infty$ and $v\in\ell_q[\ell_p]$, then
\begin{align}
    \|v\|_{\ell_q[\ell_p]} = \begin{cases}
        \inf_{\underset{v_i\neq0\implies a_i,b_i>0}{a,b\geq0}}\|a\|_{2r}\|b\|_{2r}\|(a^{-1}\otimes\mathbf{1})v(b^{-1}\otimes\mathbf{1})\|_{\ell_p} \quad &\text{ if } q\leq p \\
        \sup_{a,b\geq0}\|a\|^{-1}_{2r}\|b\|^{-1}_{2r}\|(a\otimes\mathbf{1})v(b\otimes\mathbf{1})\|_{\ell_p}  &\text{ if } q\geq p,
    \end{cases}
\end{align} where $\frac{1}{r}=\big|\frac{1}{q}-\frac{1}{p}\big|$.
\end{theorem*}

\begin{remark}
These expression should remind directly of the factorization expression for the Operator valued Schatten norms on $\cS_q[\cS_p]$ of \cite{Devetak.2006, Beigi.2023} with the difference, apart from the operators being replaced by vectors, that these manifestly are given for any $0<q,p$. This should illustrate that factorization or variational expressions of this type are quite natural and will serve as our starting point for Operator valued Schatten quasi-norms.
\end{remark}

To prove this we will be using the generalized Hölder's inequality \cref{lem:GeneralizedHölder} for vectors as the core tool of this subsection.


\begin{corollary}
Let $0<p,q,r\leq \infty$, then 
\begin{align}   
    \|v\|_{\ell_q}&=\inf_{\underset{v_i\neq 0 \implies a_i,b_i>0}{a,b\in\ell_{2r}}}\|a\|_{\ell_{2r}}\|b\|_{\ell_{2r}}\|a^{-1}vb^{-1}\|_{\ell_p} \quad &\text{ for } q\leq p \text{ and } \frac{1}{r}=\frac{1}{q}-\frac{1}{p}, \\
    \|v\|_{\ell_q}&=\sup_{a,b\in\ell_{2r}}\|a\|^{-1}_{\ell_{2r}}\|b\|^{-1}_{\ell_{2r}}\|a^{}vb^{}\|_{\ell_p} \quad &\text{ for } q\geq p \text{ and } \frac{1}{r}=\frac{1}{p}-\frac{1}{q},   
\end{align}
\end{corollary}

\begin{proof}
Let $0<q\leq p$ and fix $\frac{1}{r}=\frac{1}{q}-\frac{1}{p}$ and $ \frac{1}{2k}=\frac{1}{2r}+\frac{1
}{q}\geq 0$. Then for any two vectors $a,b$, s.t. $a_i,b_i\neq0$ whenever $v_i\neq0$ we have
\begin{align}
    \|v\|_{\ell_q}=\|aa^{-1}vb^{-1}b\|_{\ell_q}\leq \|a\|_{\ell_{2r}}\|a^{-1}vb^{-1}b\|_{\ell_{2k}}\leq\|a\|_{\ell_{2r}}\|b\|_{\ell_{2r}}\|a^{-1}vb^{-1}\|_{\ell_p} ,
\end{align} since $\frac{1}{q}=\frac{1}{2r}+\frac{1}{2k}=\frac{1}{2r}+\frac{1}{2r}+\frac{1}{p}=\frac{1}{r}+\frac{1}{p}$.
Achievability follows when choosing $a_i=b_i=|v_i|^s$, where $s=\frac{p}{2r}$, as then
\begin{align}
  \bigg(\sum_i|v_i|^{(1-2s)q}\bigg)^\frac{1}{q}  \bigg(\sum_i|v_i|^{2rs}\bigg)^{\frac{2}{2r}}=\bigg(\sum_i|v_i|^p\bigg)^\frac{1}{p},
\end{align} since $2rs=(1-2s)q=p$ and $\frac{1}{q}+\frac{1}{r}=\frac{1}{p}$.
The other case $0<q\leq p$ goes similarly. First by generalized Hölder we have
\begin{align}
    \|avb\|_{\ell_q}\leq \|a\|_{\ell_{2r}}\|vb\|_{\ell_{2k}}\leq \|a\|_{\ell_{2r}}\|b\|_{\ell_{2r}}\|v\|_{\ell_p} \Leftrightarrow \|v\|_{\ell_q}\geq \|a\|^{-1}_{\ell_{2r}}\|b\|_{\ell_{2r}}^{-1}\|avb\|_{\ell_q},
\end{align} since $\frac{1}{q}=\frac{1}{2r}+\frac{1}{2k}=\frac{1}{2r}+\frac{1}{2r}+\frac{1}{p}=\frac{1}{r}+\frac{1}{q}$. Achievability follows equally be setting $a_i=b_i=|v_i|^s$, where $s=\frac{p}{2r}$, since
\begin{align}
    \bigg(\sum_i|v_i|^{(1+2s)q}\bigg)^{\frac{1}{q}}    \bigg(\sum_i|v_i|^{2rs}\bigg)^{-\frac{2}{2r}} =     \bigg(\sum_i|v_i|^{p}\bigg)^{\frac{1}{p}},
\end{align} as $2rs=(1+2s)q=p$ and $\frac{1}{q}-\frac{1}{r}=\frac{1}{p}$.
\end{proof}

For the multi-indexed $\ell_p$-norms, these two variational formulas suffice to give a closed form expression.

\begin{proof}[Proof of \cref{thm:classical-2-indexed-spaces}]
The proof hinges on the simple observation that the $\ell_q[\ell_p]$-quasi-norm of a vector $v=(v_{ij})_{ij}$ is nothing but $\|w\|_{\ell_q}$ for $w=(w_i)_i$ with $w_i=(\sum_j|v_{ij}|^p)^\frac{1}{p}\equiv \|v_{i\cdot}\|_p$. Hence it follows directly that
\begin{align}
\|v\|_{\ell_q(\ell_p)}=\|w\|_{\ell_q}=\begin{cases}\inf_{a,b\in\ell_{2r}}\|a\|_{\ell_{2r}}\|b\|_{\ell_{2r}}\|a^{-1}wb^{-1}\|_{\ell_p} &\text{ for } q\leq p \\
\sup_{a,b\in\ell_{2r}}\|a\|^{-1}_{\ell_{2r}}\|b\|^{-1}_{\ell_{2r}}\|a^{}wb^{}\|_{\ell_p} \quad &\text{ for } q\geq p.
    \end{cases}
\end{align}
Now observing that
\begin{align}
    \|awb\|_{\ell_p}&=(\sum_i|a_iw_ib_i|^p)^\frac{1}{p} = (\sum_{i,j}|a_iv_{ij}b_i|^p)^\frac{1}{p} = \|(a\otimes\1)v(b\otimes\1)\|_{\ell_p} \\
     \|a^{-1}wb^{-1}\|_{\ell_p}&=(\sum_i|a^{-1}_iw_ib^{-1}_i|^p)^\frac{1}{p} = (\sum_{i,j}|a^{-1}_iv_{ij}b^{-1}_i|^p)^\frac{1}{p} = \|a_1^{-1}vb_1^{-1}\|_{\ell_p}
\end{align} finishes the proof.
\end{proof}

\subsubsection{Variational Formulas for Schatten-Quasi-Norms}\label{app:SpVariationalFormulas}

There exist well known variational formulas for computing Schatten quasi-norms, in particular \cite[Lemma 3.4]{Book.Tomamichel.2016} which are effectively consequences of reverse Hölder inequalities and factoring through the trace-norm. While for these it suffices to consider optimizing over one matrix valued variable, one can also rewrite these in terms of factorizations mimicking the above ones. To do so we will only require the generalized Hölder's inequality \cref{lem:GeneralizedHölder} for operators.


\begin{proof}[Proof of \cref{prop:SpVariationalFormulas}]
Consider first the case $q\leq p$, set $\frac{1}{r}=\frac{1}{q}-\frac{1}{p}$, then due to generalized Hölder's inequality we have for any $X,a,b$ s.t. $X=\Pi_aX\Pi_b$ that
\begin{align}
    \|X\|_q=\|aa^{-1}Xb^{-1}b\|_q\leq \|a\|_{2r}\|a^{-1}Xb^{-1}\|_p\|b\|_{2r},
\end{align} where $a^{-1},b^{-1}$ denote the generalized Moore-Penrose inverses of $a,b$.
To see that this is achievable let $X=UDV$ be the singular value decomposition of $X$. Then taking $a=UD^{s}U^*, b=V^*D^{s}V$ with $s=\frac{q}{2r}$ one quickly verifies equality.

In the case $q\geq p$, set $\frac{1}{r}=\frac{1}{p}-\frac{1}{q}$, then due to generalized Hölder's inequality we have for any $X,a,b$ that
\begin{align}
    \|aXb\|_p\leq \|a\|_{2r}\|X\|_q\|b\|_{2r}
\end{align}
which is achieved for $a=UD^{s}U^*, b=V^*D^{s}V$, where $s=\frac{q}{2r}$ and where the singular value decomposition of $X=UDV$. Note that these $a,b$ satisfy $X=\Pi_aX\Pi_b$.
\end{proof}

\subsection{Property of Optimizers for $q\geq p$ of \cref{lem:SupportConditions}}\label{app:MarginalSupports}

\begin{proof}
For completeness we will give an explicit argument that the optimizers in the supremum case ($q\geq p$) of \cref{lem:SupportConditions} when $\frac{1}{r}<1$ do actually possess this full marginal rank. This argument also nicely demonstrates the required continuity of our objective functional. The idea of this proof is that given any $a$ whose support is too small, we can perturb it which will both increase its support and the objective value.
For sake of contradiction take any operator $b$ and some operator $a$ with support projection $\Pi_a$ s.t. $X\neq (\Pi_a\otimes\1)X(\Pi_b\otimes\1)$. Let $\Pi_a^\perp:=\1-\Pi_a$, then $Z:=(\Pi_a^\perp\otimes\1)X(b\otimes\1)\neq0$ and for some $\epsilon>0$ consider the perturbation
\begin{align}
    a_\epsilon:=a+\epsilon\Pi_a^\perp.
\end{align}
Note that by construction $a_\epsilon$ has full support.
We will now show that the objective value under this perturbation does not decrease for small enough $\epsilon>0$, i.e. that 
\begin{align}
    \frac{\|(a_\epsilon\otimes\1)X(b\otimes\1)\|_p}{\|a_\epsilon\|_{2r}} \geq \frac{\|(a\otimes\1)X(b\otimes\1)\|_p}{\|a\|_{2r}}
\end{align} for small enough $\epsilon>0$.

Consider first the denominator under this perturbation. Then using the Taylor approximation $(1+\epsilon)^{s}=1+s\epsilon+o(\epsilon)$ for small $\epsilon$ we have
\begin{align}
    \|a_\epsilon\|_{2r}= \left(\|a\|^{2r}_{2r}+\epsilon^{2r}\|\Pi_a^\perp\|^{2r}_{2r}\right)^\frac{1}{2r} = \|a\|_{2r}\left(1+\epsilon^{2r}\frac{\|\Pi^\perp_a\|_{2r}^{2r}}{\|a\|^{2r}_{2r}}\right)^\frac{1}{2r}=\|a\|_{2r}(1+K_1\epsilon^{2r}+o(\epsilon^{2r}))
\end{align} where $0\leq K_1=\frac{\|\Pi^\perp_a\|_{2r}^{2r}}{2r\|a\|^{2r}_{2r}}<\infty$.

Analyzing the numerator under this perturbation is slightly more involved and requires some perturbation theory. The idea will be that this perturbation on the one hand changes existing eigenvalues which will result in an relative increase of the $p-$(quasi-)norm of order $\mathcal{O}(\epsilon^2)$ and on the other creates new eigenvalues which will result in an relative increase of the $p-$(quasi-)norm of order $\mathcal{O}(\epsilon^p)$. Then since $2r>p$ this will yield a net increase of the objective for small $\epsilon$.

To see this write $X^\prime:=X(b\otimes\1)$, see that
\begin{align}
    M_\epsilon:=(a_\epsilon\otimes\1)X^\prime=(a\otimes\1)X^\prime+\epsilon(\Pi_a^\perp\otimes\1)X^\prime \equiv Y+\epsilon Z
\end{align} and notice that
\begin{align}
    Y^*Z=X^{\prime^*}(a\otimes\1)(\Pi_a^\perp\otimes\1)X^\prime = 0 = Z^*Y
\end{align} which implies that
\begin{align}\label{equ:Suppoertproof.1}
    M_\epsilon^*M_\epsilon=Y^*Y+\epsilon^2Z^*Z.
\end{align}
This means in particular that $M_\epsilon^*M_\epsilon$ is monotonically increasing in Loewner order and hence also all its eigenvalues and its Schatten-$\frac{p}{2}$-(quasi-)norm, i.e. $\epsilon\mapsto\|M_\epsilon\|_{p}$ is monotonically non-decreasing. Morally this demonstrates that the increase to existing and new eigenvalues will in lowest order be of order $\mathcal{O}(\epsilon^2)$.

It is a standard result by Kato \cite{Kato.1995} that the perturbed eigenvalues of $H(\epsilon)=M_\epsilon^*M_\epsilon$ expand in a Puiseux series and from \eqref{equ:Suppoertproof.1} it becomes evident that the lowest order increase must be of order $\epsilon^2$. So the eigenvalues $\lambda_i(\epsilon)$ of $M_\epsilon^*M_\epsilon$ satisfy
\begin{align}
    \lambda_i(\epsilon)=\lambda_i+C_i\epsilon^2+o(\epsilon^2)
\end{align} where $C_i\geq0$ and $\lambda_i\geq0$ are the eigenvalues of $Y^*Y=M_0^*M_0$. Thus we have
\begin{align}
    \|M_\epsilon^* M_\epsilon\|^\frac{p}{2}_{\frac{p}{2}}&= \sum_i  \lambda_i(\epsilon)^\frac{p}{2} = \sum_{i,\lamdba_i>0}\lambda_i(\epsilon)^\frac{p}{2}+\sum_{i,\lambda_i=0}\lambda_i(\epsilon)^\frac{p}{2} \\&= \sum_{i,\, \lambda_i>0}(\lambda^\frac{p}{2}_i+\epsilon^2\lambda_i^{\frac{p}{2}-1}C_i+ o(\epsilon^2))+\sum_{i,\, \lambda_i=0}C^\frac{p}{2}_i\epsilon^p+o(\epsilon^p) \\
    &= \|M_0^*M_0\|_{\frac{p}{2}}^\frac{p}{2}+K\epsilon^2+K^\prime\epsilon^p+o(\epsilon^2,\epsilon^p)
\end{align} with $K,K^\prime\geq0$ and at least one strictly positive. This induces a relative increase
\begin{align}
    \|M_\epsilon\|_p&=\|M_\epsilon^* M_\epsilon\|^\frac{1}{2}_{\frac{p}{2}} = (\|M_0^*M_0\|_{\frac{p}{2}}^\frac{p}{2}+K\epsilon^2+K^\prime\epsilon^p+o(\epsilon^2,\epsilon^p))^\frac{1}{p}
    \\ &= \|M_0\|_{p}\left(1+\underbrace{\frac{K}{p\|M_0\|^p_p}}_{K_2\geq0}\epsilon^2+\underbrace{\frac{K^\prime}{p\|M_0\|_p^p}}_{K_3\geq0}\epsilon^p+o(\epsilon^2,\epsilon^p)\right)
\end{align} with at least one of $K_2,K_3>0$. Considering our objective function we have thus shown that this perturbation induces
\begin{align}
    \frac{\|M_\epsilon\|_p}{\|a_\epsilon\|_{2r}}=\frac{\|M_0\|_p(1+K_2\epsilon^2+K_3\epsilon^p+o(\epsilon^2,\epsilon^p))}{\|a_0\|_{2r}(1+K_1\epsilon^{2r}+o(\epsilon^{2r})}= \frac{\|M_0\|_p}{\|a_0\|_{2r}}(1+K_2\epsilon^2+K_3\epsilon^p-K_1\epsilon^{2r}+o(\epsilon^2,\epsilon^p,\epsilon^{2r}))
\end{align} which is an increase for small $\epsilon$ whenever $2r>\max\{2,p\}$ since without further assumptions we are only guaranteed that $\max\{K_2,K_3\}>0$.
This condition holds whenever $r>1$ as it always holds that $2r>r\geq p$ since $\frac{1}{r}=\frac{1}{p}-\frac{1}{q}\leq \frac{1}{p}$. By the same argument for $b$ instead of $a$ the overall claim follows hence whenever $r>1$.
\end{proof}

\begin{proposition}\label{prop:3index.sym.pos}
Let $X_{PQR}\geq0$ then the optimization in \cref{def:3.indexed.special} can be restricted to $a_P=b_P$, i.e.
\begin{align}
\|X\|_{(P:p,Q:q,R:p)}=\sup_{a\geq0}\|a\|^{-2}_{2r}\|a_PX_{PQR}a_P\|_{(PQ:q,R:p)}.
\end{align}
\end{proposition}

\begin{proof}
The proof will go analogously to the supremum case of \cref{lem:Positive.Symmetrie}. We first establish a suitable Hölder-type inequality for this 3-indexed quasi-norm and then use that to conclude the claim.
Let $A,B\in\cB(\cH_P\otimes\cH_Q\otimes\cH_R)$, then by \cref{def:3.indexed.special} and Hölder's inequality for Schatten-quasi-norms we have 
\begin{align}
   & \|AB\|_{(PQ:q,R:p)}\\
    &\qquad =\inf_{c,d\geq0}\|c\|_{2r}\|d\|_{2r}\|c^{-1}_{PQ}A_{PQR}B_{PQR}d^{-1}_{PQ}\|_{p} \\
    &\qquad\leq \inf_{c,d\geq0}\|c\|_{2r}\|d\|_{2r}\|c^{-1}_{PQ}A_{PQR}A^*_{PQR}c^{-1}_{PQ}\|^\frac{1}{2}_p\|d^{-1}_{PQ}B_{PQR}^*B_{PQR}d^{-1}_{PQ}\|^\frac{1}{2}_{p} \\
    &\qquad= \|A_{PQR}A_{PQR}^*\|_{(PQ:q,R:p)}^\frac{1}{2}\|B^*_{PQR}B_{PQR}\|_{(PQ:q,R:p)}^\frac{1}{2} \\
    &\qquad\leq\! \|AA^*\|^\frac{1}{2}_{(PQ:q,R:p)}\|B^*B\|_{(PQ:q,R:p)}^\frac{1}{2}
\end{align} where we used \cref{lem:Positive.Symmetrie}. Now it follows for any $X_{PQR}=xx^*\geq0$ that
\begin{align}
    \|a_PX_{PQR}b_P\|_{(PQ:q,R:p)}&\leq \|a_Pxx^*a_P\|^\frac{1}{2}_{(PQ:q,R:p)}\|b_Pxx^*b_P\|^\frac{1}{2}_{(PQ:q,R:p)}\\&\leq \max\{\|a_Pxx^*a_P\|_{(PQ:q,R:p)},\|b_Pxx^*b_P\|_{(PQ:q,R:p)}\}
\end{align} which finishes the proof.
\end{proof}

\section{Some Useful Optimization Tools}

\begin{lemma}\label{lem:continuity}
Let $0<s\leq m$ with $1\leq m$ and set $\frac{1}{r}:=\frac{1}{s}-\frac{1}{m}\geq0$, then for some $K\in\cS_{2r}(\cH)$ with $\cH$ being a finite dimensional Hilbert space, the map
\begin{align}
    f: &\cD(\cH):=\{a\in\cS_1(\cH)|a\geq0, \|a\|_1\leq 1\} \to [0,\|K\|^{2s}_{2r}] \\
    &a\mapsto \|Ka^\frac{1}{m}K^*\|_{s}^s=\|a^\frac{1}{2m}KK^*a^\frac{1}{2m}\|_{s}^s
\end{align}
is concave and (norm-)continuous.
\end{lemma}

Concavity follows from \cite[Theorem 1.1]{Zhang.2020} by the assumptions on $s$ and $m$. Continuity holds since in finite dimensions matrix multiplication, taking positive powers, and the trace are continuous operations on the compact set $\cD(\cH)$.

\subsection{Sion's Minimax Theorem}
A major technical part of this paper will be the switching of infima and suprema. There are various formulations of Sion's minimax theorem, going back to \cite{Sion.58} which give sufficiency criteria to do exactly this.
We will use the following formulation going back to Ky Fan, but in a slightly unusual manner, in that we assume the set $X$ over which we take the supremum to be compact, resulting in requiring the function to be upper semi-continuous in that respective variable $x\in X$.

\begin{theorem}[Ky Fan's Sion's minimax theorem, \cite{KyFan.53} Theorem 2]\label{thm:Sion} Let $X$ be a compact convex subset of a Hausdorff topological vector space and $Y$ a subset of a linear space. Let $f:X\times Y\to \mathbb{R}\cup\{-\infty\}$ be upper-semicontinous on $X$ for fixed $y\in Y$ and concave in $x$ and convex in $y$, then
\begin{align}
    \sup_{x\in X}\inf_{y\in Y}f(x,y)= \inf_{y\in Y}\sup_{x\in X}f(x,y).
\end{align}
\end{theorem}
We note here, that when minimizing and maximizing over multiple variables each, then the function needs to be jointly concave, respectively convex in them. Note also that it suffices that one of the two sets over which the function is optimized is required to be compact, as long as the function is suitably semicontinuous there. This follows from applying Sions minimax theorem to the function $-f$.

\subsection{Optimization over Products as Sums}

Two important tools we require are rewritings of optimizations over a product of functions into a sum of suitable powers of these functions. We will consider the cases of a maximization \cref{lem:Sup.Prod.to.Sum} and minimization \cref{lem:Inf.Prod.to.Sum} separately.
We present specific, rather then the most general such statements, but results like this are more or less known and generalize easily to products of arbitrary number of functions with different scalings, with the core proof tool being the (reversed) weighted AM-GM inequality.

For convenience we repeat them here in the specific case we require. 

\begin{proposition}[A weighted AM-GM inequality]
Let $\alpha\in(0,1)$ and $x,y,z\geq0$, then
\begin{align}
    \alpha x+\frac{1-\alpha}{2}y+\frac{1-\alpha}{2}z \geq x^\alpha y^{\frac{1-\alpha}{2}}z^{\frac{1-\alpha}{2}}.
\end{align}
\end{proposition}

\begin{lemma}[A reverse weighted AM-GM inequalities]
Let $\alpha>1$, and $x\geq0,\,y,z>0$ then
\begin{align} \label{equ:reversed.AM.GM}
    \alpha x+\frac{1-\alpha}{2}y+\frac{1-\alpha}{2}z \leq x^\alpha y^{\frac{1-\alpha}{2}}z^{\frac{1-\alpha}{2}}.
\end{align}
\end{lemma}
\begin{proof}
Denote the LHS with $X:=\alpha x+\frac{1-\alpha}{2}y+\frac{1-\alpha}{2}z$. If $X\leq 0$, then the inequality holds trivially, as the RHS is positive semidefinite. Now for $X\geq0$ apply the standard weighted AM-GM inequality with weights $0\leq \frac{1}{\alpha},\frac{\alpha-1}{2},\frac{\alpha-1}{2}\leq 1$ to 
\begin{align}
    x=\frac{1}{\alpha}X+\frac{\alpha-1}{2\alpha}y+\frac{\alpha-1}{2\alpha}z\geq X^\frac{1}{\alpha}y^{\frac{\alpha-1}{2\alpha}}z^{\frac{\alpha-1}{2\alpha}}.
\end{align} Rearranging yields the desired inequality \eqref{equ:reversed.AM.GM} $y^{\frac{1-\alpha}{2}}z^{\frac{1-\alpha}{2}}x^\alpha\geq X$.

\end{proof}

Note that equality in both of these holds if and only if $x=y=z$.

Explicitly for a product of three functions we now have the following.

\begin{lemma}[Supremum of products as sum]\label{lem:Sup.Prod.to.Sum}
Given three positive functions $g,h:\cB^+(\cH)\to [0,\infty],\, f:\cB^+(\cH\otimes\cH)\to [0,\infty]$ and numbers $r_2<0<r_1$ such that $\frac{1}{r_0}:=\frac{2}{r_1}+\frac{1}{r_2}<0$ and
\begin{align}
    g(\lambda C)=\lambda^{r_1}g(C)\quad h(\mu D)=\mu^{r_1}h(D)\quad f(\lambda C,\mu D)= (\lambda\mu)^{-r_2}f(C,D),
\end{align} holds for any $\lambda,\mu\geq0$.
Then 
\begin{align}
    \sup_{C,D\geq0}\left\{g(C)^{\frac{r_0}{r_1}}h(D)^{\frac{r_0}{r_1}}f(C,D)^{\frac{r_0}{r_2}}\right\} = \sup_{C,D\geq0}\left\{\frac{r_0}{r_1}g(C)+\frac{r_0}{r_1}h(D)+\frac{r_0}{r_2}f(C,D)\right\}.
\end{align}
\end{lemma}
The proof we present here is inspired by the proof of \cite[Lemma B.4]{Rubboli.2025}.
\begin{proof}
As $2\frac{r_0}{r_1}+\frac{r_0}{r_2}=1$ and $\frac{r_0}{r_1}<0$ we have $\frac{r_0}{r_2}>1$ and so by weighted reverse AM-GM inequality \eqref{equ:reversed.AM.GM} we get
\begin{align} \label{equ:proof.sup.prod.sum}
    g(C)^\frac{r_0}{r_1}h(D)^\frac{r_0}{r_1}f(C,D)^\frac{r_0}{r_2}\geq \frac{r_0}{r_1}g(C)+\frac{r_0}{r_1}h(D)+\frac{r_0}{r_2}f(C,D).
\end{align} This inequality is clearly preserved under taking suprema over $C,D\geq0$. For the other inequality consider the transformation 
\begin{align}
    (C,D)\mapsto (\lambda C,\mu D), \quad \lambda,\mu\geq0.
\end{align} The LHS of \eqref{equ:proof.sup.prod.sum} is invariant under it, since
\begin{align}
    g(\lambda C)^\frac{r_0}{r_1}h(\mu D)^\frac{r_0}{r_1}f(\lambda C,\mu D)^\frac{r_0}{r_2} &= (\lambda\mu)^{r_0}g(C)^\frac{r_0}{r_1}h(D)^\frac{r_0}{r_1}(\lambda\mu)^{-r_0}f(C,D)^\frac{r_0}{r_2}\\
    &= g(C)^\frac{r_0}{r_1}h(D)^\frac{r_0}{r_1}f(C,D)^\frac{r_0}{r_2}.
\end{align}
On the other hand, given any feasible $C,D\geq0$ setting
\begin{align}
    \tilde{\lambda}:=\tilde{\mu}\left(\frac{h(D)}{g(C)}\right)^\frac{1}{r_1}, \quad \tilde{\mu}:=\left(\left(\frac{h(D)}{g(C)}\right)^\frac{1}{r_1}\left(\frac{f(C,D)}{g(C)}\right)^\frac{1}{r_2-r_1}\right)^\frac{r_1-r_2}{2r_2-r_1}
\end{align} one verifies that $g(\tilde{\lambda}C)=f(\tilde{\lambda}C,\tilde{\mu}D)=h(\tilde{\mu}D)$, which implies equality in the inequality \eqref{equ:proof.sup.prod.sum} for $(\tilde{\lambda} C,\tilde{\mu} D)$ and thus finishes the proof of the lemma.
\end{proof}

Similarly for infima of a product of three functions we have
\begin{lemma}[Infimum of products as sum]\label{lem:Inf.Prod.to.Sum}
Given three positive functions $g,h:\cB^+(\cH)\to [0,\infty],\, f:\cB^+(\cH\otimes\cH)\to [0,\infty]$ and numbers $0<r_1,r_2$ s.t. 
\begin{align}
   g(\lambda C)=\lambda^{r_1}g(C), \quad  h(\mu D)=\mu^{r_1}h(D), \quad f(\lambda C,\mu D)=(\lambda\mu)^{-r_2}f(C,D)
\end{align} 
holds for any $\lambda,\mu\geq0$. 
Define $\frac{1}{r_0}:=\frac{2}{r_1}+\frac{1}{r_2}$, then it holds that
\begin{align}
    \inf_{C,D\geq0}g(C)^{\frac{r_0}{r_1}}h(D)^{\frac{r_0}{r_1}}f(C,D)^{\frac{r_0}{r_2}} = \inf_{C,D\geq0}\left\{\frac{r_0}{r_1}g(C)+\frac{r_0}{r_1}h(D)+\frac{r_0}{r_2}f(C,D)\right\}.
\end{align}
\end{lemma}

This proof goes analogous as the previous one, with the difference being the use of the usual weighted AM-GM inequality instead of the reversed one. 
\begin{proof}
On the one hand, by the weighted arithmetic-geometric mean inequality it holds that
\begin{align}\label{equ:w.am.gm.ineq}
    \frac{r_0}{r_1}g(C)+\frac{r_0}{r_1}h(D)+\frac{r_0}{r_2}f(C,D)\geq g(C)^{\frac{r_0}{r_1}}h(D)^{\frac{r_0}{r_1}}f(C,D)^{\frac{r_0}{r_2}},
\end{align} which still holds, when taking the infimum over all positive $C,D$. To show the achievability first notice that the RHS of \eqref{equ:w.am.gm.ineq} is invariant under the transformation 
\begin{align}
    (C,D)\mapsto (\lambda C,\mu D)
\end{align} for positive $\lambda,\mu\geq0$, since
\begin{align}
g(\lambda C)^{\frac{r_0}{r_1}}f(\lambda C,\mu D)^{\frac{r_0}{r_2}}h(\mu D)^{\frac{r_0}{r_1}} &= \lambda^{r_0}g(C)^{\frac{r_0}{r_1}}\lambda^{-r_0}\mu^{-r_0}f(C,D)^{\frac{r_0}{r_2}}\mu^{r_0}h(D)^{\frac{r_0}{r_1}}\\
&= g(C)^{\frac{r_0}{r_1}}f(C,D)^{\frac{r_0}{r_2}}h(D)^{\frac{r_0}{r_1}}.
\end{align}
Now given any feasible $C,D$, let 
\begin{align}
    \tilde{\lambda}&:=\big(f(C,D)g(C)^{-\frac{r_2}{r_1}-1}h(D)^{\frac{r_2}{r_1}}\big)^{\frac{1}{r_1+2r_2}}, \qquad
    \tilde{\mu} :=\tilde{\lambda}\left(\frac{g(C)}{h(D)}\right)^\frac{1}{r_1}.
\end{align} Then it is easily verified that
$g(\tilde{\lambda}C)=f(\tilde{\lambda}C,\tilde{\mu}D)=h(\tilde{\mu} D)$, which implies equality in \eqref{equ:w.am.gm.ineq} and hence the statement of the lemma.
\end{proof}

\section{Concavity-Convexity of a Certain Operator Space Functional}

Here we prove the following lemma, which was required in the proof of \cref{lem:3.index.formulas} leading to \Cref{thm:QuasiNormMinkowskiInequality}. It is a statement purely about operator valued Schatten norms and thus also of independent interest. The proof uses some properties of operator valued Schatten norms not introduced in this work. For an introduction see the references \cite{Book.Pisier.1998, Book.Pisier.2003, Beigi.2023}. We use the same notation as in \cite{Fawzi.2026}.

\begin{lemma}\label{lem.app.op.concavity}
Let $\cH$ be a finite dimensional Hilbert space, $\cX$ be an operator space and $s\geq r\geq1$, then it holds that 
\begin{align}
    \|(a^{\frac{1}{2s}}\otimes\1)A(a^{\frac{1}{2s}}\otimes\1)\|_{(r;\cX)}
\end{align} is concave and upper semi-continuous in $0\leq a\in\cS_1(\cH)$ and continuous and convex in $0\leq A\in\cS_\infty[\cH;\cX]$, where $\|\cdot\|_{(r;\cX)}$ denotes the norm of the operator valued Schatten space $\cS_r[\cH;\cX]$.
\end{lemma}

\begin{proof}
The convexity in $A$ follows simply from the fact that map $A\mapsto (a^\frac{1}{2s}\otimes\1)A(a^\frac{1}{2s}\otimes\1)$ is linear and $\|\cdot\|_{(r;\cX)}$ a norm by construction \cite{Book.Pisier.1998}. For the concavity in $a>0$ see that Pisier's theorem, \cite[Lemma 3.1]{Fawzi.2026}, implies that 
\begin{align}
    \|(a^{\frac{1}{2s}}\otimes\1)A(a^{\frac{1}{2s}}\otimes\1)\|_{(r;\cX)} 
    &= \inf_{c\geq0,\Pi_{c}=L_{(a^{\frac{1}{2r}}\otimes\1)A(a^{\frac{1}{2r}}\otimes\1)}}\|c\|^2_{2r}\|(c^{-1}a^{\frac{1}{2s}}\otimes\1)A(a^{\frac{1}{2s}}c^{-1}\otimes\1)\|_{(\infty;\cX)} \\ 
    &= \inf_{C}\|a^{\frac{1}{2s}}CC^*a^{\frac{1}{2s}}\|_r\|(C^{-1}\otimes\1)A(C^{-1}\otimes\1)\|_{(\infty;\cX)} \\
    &=\inf_{\tilde{C}\geq0}\|a^{\frac{1}{2s}}\tilde{C}a^{\frac{1}{2s}}\|_r\|(\tilde{C}^{-\frac{1}{2}}\otimes\1)A(\tilde{C}^{-\frac{1}{2}}\otimes\1)\|_{(\infty;\cX)} 
\end{align} 
Now observe that due to \cite[Theorem 1.1 (1)]{Zhang.2020} $a\mapsto\|a^{\frac{1}{2s}}\tilde{C}a^{\frac{1}{2s}}\|_r^r$ is concave in $a>0$ for $0<\frac{1}{s}\leq 1$ and $0<r\leq s$, i.e.\ for $s\geq 1$ and $r\leq s$; see \cite[Theorem 1.1 (1)]{Zhang.2020}, or \cite{Epstein.1973}. Composing with the concave and monotonically increasing map $u\mapsto u^{1/r}$ on $[0,\infty)$, which uses $r\geq1$, shows that $a\mapsto\|a^{1/2s}\tilde C a^{1/2s}\|_r$ is concave; since the feasible set of $\tilde C$ does not depend on $a$, this is a pointwise infimum of concave functions and hence concave.  Continuity extends the statement to singular $a$.

Upper semi-continuity in $a$ follows since $a\mapsto\|a^{1/2s}\tilde C a^{1/2s}\|_r$ is continuous, hence also upper semi-continuous, and upper semi-continuity is preserved under infima. Continuity in $A$ holds because in finite dimensions all norms are equivalent and $A\mapsto(a^{1/2s}\otimes\1)A(a^{1/2s}\otimes\1)$ is linear.
\end{proof}

\end{document}